\documentclass[11pt]{article}
\usepackage[nohead,margin=1in]{geometry}
\usepackage{palatino}
\usepackage{amssymb,amsmath,amsfonts,eurosym,geometry,ulem,graphicx,caption,color,setspace,comment,footmisc,caption,natbib,pdflscape,subfigure,array}
\usepackage[colorlinks = true,
            linkcolor = blue,
            urlcolor  = blue,
            citecolor = blue,
            anchorcolor = blue]{hyperref}
\usepackage{natbib}
\usepackage[inline]{enumitem}
\usepackage{float}
\usepackage{pgfplots}
\usepackage{tikz}
\usepgfplotslibrary{fillbetween}
\newtheorem{theorem}{\color{blue}Theorem}
\newtheorem{corollary}{\color{blue}Corollary}
\newtheorem{proposition}{\color{blue}Proposition}
\newtheorem{lemma}{\color{blue}Lemma}

\newenvironment{proof}[1][Proof]{\noindent\textbf{#1.} }{\ \rule{0.5em}{0.5em}}
\newtheorem{example}{\color{blue}Example}

\newcolumntype{L}[1]{>{\raggedright\let\newline\\arraybackslash\hspace{0pt}}m{#1}}
\newcolumntype{C}[1]{>{\centering\let\newline\\arraybackslash\hspace{0pt}}m{#1}}
\newcolumntype{R}[1]{>{\raggedleft\let\newline\\arraybackslash\hspace{0pt}}m{#1}}
\pgfplotsset{compat=1.18}
\DeclareMathOperator\supp{supp}

\begin{document}

\title{Price Impact of Insurance\footnote{We are deeply grateful to Piotr Dworczak, Jeff Ely and Alessandro Pavan for their guidance. All errors remain our own.}}
\author{Andrea Di Giovan Paolo\footnote{Department of Economics, Northwestern University. Email: \href{mailto:andreadigiovanpaolo@u.northwestern.edu}{andreadigiovanpaolo@u.northwestern.edu}} \and Jose Higueras\footnote{Department of Economics, Northwestern University. Email: \href{mailto:josehiguerascorona2025@u.northwestern.edu}{josehiguerascorona2025@u.northwestern.edu}}}
\date{\today}

\maketitle
\vspace{-2em}
\begin{center}
    \Large \textbf{Preliminary Draft}
\end{center}
\vspace{2em}
\begin{abstract}
This paper analyzes optimal insurance design when the insurer internalizes the effect of coverage on third-party service prices. A monopolistic insurer contracts with risk-averse agents who have sequential two-dimensional private information and preferences represented by Yaari’s dual utility. Insurance contracts shape service demand and, through a market-clearing condition, determine equilibrium third-party prices. We characterize the structure of optimal contracts and show they take simple forms---either full coverage after a deductible is paid or limited coverage with an out-of-pocket maximum---closely mirroring real-world insurance plans. Technically, we formulate the problem as a sequential screening model and solve it using tools from optimal transport theory.
\end{abstract}

\section{Introduction}
Insurance plays a crucial role in mitigating risks by covering a portion of the monetary losses customers incur following adverse events. These monetary losses typically depend on prices set by third-party service markets. For example, in healthcare insurance, the monetary loss is directly tied to the price of hospital services. In car insurance, the loss is determined by the price of auto repair services, whereas in home insurance, it depends on the price for property restoration services. While the contract theory literature has extensively studied the optimal design of insurance contracts under adverse selection and moral hazard \citep{stiglitz1977monopoly, grossman1983analysis, chade2012optimal,gershkov2023optimal}, less attention has been paid to the equilibrium effects of insurance. The type and level of insurance coverage can influence the prices charged by third-party providers. In healthcare insurance, for instance, offering more generous coverage can increase policyholders’ demand for medical services, which in turn drives up hospital prices and insurer payouts. Similar dynamics arise in auto and home insurance, where contract design influences the demand for repair and restoration services, thereby affecting market prices and claims. Given that large\footnote{According to the latest report on competition in health insurance \citep{AMA2024CompetitionHealthInsurance}, the American Medical Association (AMA) finds that, based on the DOJ/FTC Merger Guidelines, 95\% of metropolitan commercial markets were highly concentrated in 2023, with a Herfindahl--Hirschman Index (HHI) above 1800. The average commercial market exhibited an HHI of 3458. Moreover, in 89\% of metropolitan markets at least one insurer held a market share of 30\% or more, and in 47\% of markets a single insurer’s share reached at least 50\%. These findings indicate that many health insurance markets display monopolistic characteristics.} insurance companies are likely to internalize these forces when designing contracts, a key question arises: how does the interaction between insurance and third-party service markets shape both the \textit{types} of contracts offered and the \textit{price} for these services?

This paper addresses these questions by incorporating equilibrium effects into a mechanism design framework. We study the problem faced by a risk-neutral monopolistic insurer offering contracts to a population of risk-averse agents, and building on \citet{gershkov2023optimal} we model their preferences using Yaari’s dual utility \citep{yaari1987dual}. In contrast to the standard expected utility framework---where risk-averse agents evaluate lotteries by applying a concave utility function to outcomes---dual utility models assume linear utility over outcomes and capture risk aversion through probability distortions. 

Agents are sequentially and privately informed about two characteristics: their ex-ante risk type---such as the probability of developing an illness or suffering a car accident---and their ex-post loss-related valuation, which reflects the monetary loss they associate with not treating or repairing the damages resulting from an adverse event. This valuation determines their willingness to pay for third-party services. 

In the first period, when the agent knows only her risk type and her loss-related valuation has not yet been realized, she can purchase an insurance contract from the monopolistic provider. An insurance contract has two parts. First, a premium that the agent must pay to obtain coverage. Second, a set of options, where each option specifies a quantity or fraction of services covered by insurance together with the out-of-pocket payment the agent would face if she later uses that option.

In the second period, once the adverse event occurs and the agent learns her loss-related valuation, she can repair damages or seek treatment from the third-party service provider. Services are supplied competitively at a market price, and the agent may demand any quantity. If she purchased an insurance contract in the first period, she may exercise one of the coverage options: doing so entitles her to obtain the specified quantity of services by paying only the associated out-of-pocket cost. 

The insurance provider chooses a menu of contracts to maximize profits. Importantly, the insurer internalizes that each contract affects an agent’s decision of whether to seek third-party services and how much to consume. As a result, the menu of contracts shapes aggregate demand in the competitive service market. Through the market-clearing condition---requiring demand to equal available supply---this demand determines the equilibrium service price.

%In our extensions, we show that this market structure can also be interpreted as one where the service price is determined through Nash bargaining between the insurer and the third-party provider. In this formulation, the provider faces a capacity constraint that limits aggregate demand, while the insurer is assumed to hold all bargaining power.

We adopt a two-stage approach to determine the structure of the contracts offered by the insurer and the resulting equilibrium service price:
\begin{enumerate*}[label=\arabic*)]
\item In the first-stage problem, we fix a price for services and identify, among the menus that sustain that price in equilibrium, the one that maximizes the insurer's profits.
\item In the second-stage problem we associate each achievable price in equilibrium with the corresponding profits obtained by the insurer, and then solve for the price that maximizes these profits.
\end{enumerate*}

Our first main result, Theorem~\ref{thm:optimal-menu}, characterizes the optimal contracts offered by the insurer in the first-stage problem. Despite the complexity of the setting, we find that optimal contracts take a simple form and resemble real-world insurance plans. Crucially, one of these contract forms arises \textit{only because of the insurer’s price impact}. Specifically, the contract offered to any given risk type takes one of the following two forms:
\begin{enumerate}
    \item \textbf{Simple deductible contract.} This contract is analogous to standard deductible-based insurance: once the agent pays a fixed amount out-of-pocket, insurance covers the remaining cost of the service in full.
    
    \item \textbf{Limited coverage deductible contract.} Under this contract, after the agent pays a fixed amount out-of-pocket, insurance covers a fraction $\alpha \in (0,1)$ of the service. If the agent wants the remaining fraction $1-\alpha$, she must pay the full service price. Alternatively, by paying a higher out-of-pocket amount, the agent can receive full coverage without limits.
\end{enumerate}

In our model, the insurer faces a novel tradeoff beyond the classic screening frictions studied in traditional insurance models (see, for instance, \cite{stiglitz1977monopoly}). Offering more generous coverage---by paying a higher share of service costs---raises agents’ demand for services, which in turn drives up the equilibrium service price. In other words, every contract the insurer writes has a price impact. This price effect provides an economic rationale for the \textit{limited-coverage deductible contract}. The role of this contract is to impose service restrictions, or more generally cost sharing, as a way to regulate agents’ utilization of services, thereby shaping aggregate demand and the resulting equilibrium price. Importantly, such contracts emerge only because of equilibrium effects. If the insurer ignores his price impact---treating service prices as fixed---the optimal contract collapses to the \textit{standard deductible contract}.

This theoretical prediction echoes real-world practice. Coinsurance or cost-sharing arrangements typically cover only a fraction $\alpha$ of treatment costs, with policyholders responsible for the remainder. Many plans also include an out-of-pocket maximum, beyond which the insurer covers all costs. More generally, the parameter $\alpha$ can capture common service restrictions. Medicare Part~A covers only up to 90 days of inpatient hospital care per benefit period; additional days must be paid entirely out-of-pocket. Private plans often cap the number of physical therapy or mental health sessions, drug plans restrict the quantity of medication dispensed or the frequency of refills, and advanced diagnostic procedures such as MRI, CT, or PET scans may not be fully covered.

%Similar coverage limits appear in other insurance markets. In car insurance, companies typically cover auto repair costs only up to the vehicle's actual cash value and may exclude repairs to aftermarket parts or custom equipment. In home insurance, many policies include sub-limits on “trace and access” coverage: the insurer may pay to repair water damage caused by a burst pipe, but not to locate the leak or restore removed flooring or walls. Some contracts also exclude repairs to the plumbing infrastructure itself, leaving the homeowner responsible for replacing damaged pipes or fixtures.

From a technical perspective, our contribution lies in extending sequential screening to environments where the mechanism designer affects market outcomes---a combination that we analyze through applications of tools from optimal transport theory. In our setting, the insurer engages in sequential screening: in the first period, the agent knows only her risk type, and in the second period she learns her loss-related valuation. This introduces an additional challenge compared to standard one-dimensional mechanism design problems, where the designer’s allocation rule maps a single type to a final outcome. To address these challenges, we formulate the insurer’s problem in quantity space—a common technique in mechanism design \citep{bulow1989simple}—and further transform it into a problem that closely resembles an optimal transport problem. In this formulation, the insurer effectively matches risk types to quantities of services demanded. By leveraging duality results from optimal transport theory, we derive several properties of the optimal contracts, including Theorem~\ref{thm:optimal-menu}.

Our second main result, Theorem~\ref{thm:envelope}, characterizes how the insurer’s profits change as the service price varies. This result is central to the analysis of the second-stage problem, where the insurer endogenously influences the equilibrium price through the menu of contracts offered to agents. In particular, we identify two economic forces that drive the insurer’s marginal value of raising the price.

The first force concerns \textit{surplus extraction}. As the service price rises, the distribution of realized losses among insured agents shifts upward in the sense of first-order stochastic dominance, as higher losses occur with positive probability. Covered agents are protected against these losses, and under dual utility they distort probabilities and assign a non-linear value to this shift. The insurer can capture this value—net of information rents—through premiums.

The second force reflects the \textit{cost of increasing the price}. A unit increase in the price leads to a proportional increase in the insurer’s payouts for a given level of coverage.

After identifying these economic forces, we introduce the concept of a \textit{price-taking benchmark}. This benchmark is the price that clears the service market and solves the insurer’s first-stage problem in the absence of equilibrium effects. In other words, it is the market-clearing price that would arise if the insurer treated prices as exogenous and ignored its impact on the allocation of services. We then ask: under what conditions does an insurer with price impact set a service price above or below this benchmark? We provide conditions under which it is sufficient to compare the surplus extraction force and the cost effect locally at the price-taking benchmark. In particular, if the surplus extraction force dominates the marginal cost force at this point, the insurer raises the price above the benchmark; if it is dominated, the insurer lowers the price below the benchmark.

Unlike standard monopoly pricing models, where the monopolist typically sets a price above the price-taking benchmark, in our setting the optimal service price may lie either above or below it. The key distinction lies in the nature of the trade-offs faced by the insurer. In standard models, the monopolist benefits unambiguously from raising the price at the benchmark: doing so increases profits per unit sold. In contrast, the insurer in our framework balances a potentially positive gain from surplus extraction against a negative effect from higher payouts. As a result, the net benefit of increasing the price at the benchmark can be either positive or negative, leading to qualitatively different pricing behavior.

Finally, we present three extensions of our baseline model that preserve its core economic insights and main results. In the first extension, we show how our framework applies to settings in which the insurance provider acts as a social planner with redistributive objectives, in the spirit of \cite{akbarpour2024redistributive}. That is, the insurer places positive weight not only on revenue, but also on agents' utility. This setting captures environments where the state is the main provider of insurance or where insurance markets are highly regulated. In the second extension, we reinterpret our baseline model in a setting where the service price is not determined by a competitive market but instead negotiated through Nash bargaining between the insurer and a capacity-constrained third-party provider, with the insurer holding all bargaining power. We also discuss how the main results extend to settings in which the provider has some bargaining power. The third extension considers a setting in which the provider offers multiple tiered services each with its own price.

\paragraph{Related literature.} The relationship between insurance and third-party service prices has received considerable empirical attention since the seminal work of \citet{feldstein1970rising}. More recent empirical work has estimated hospital–insurer bargaining models, enabling researchers to address important questions such as the division of surplus under hospital capacity constraints \citep{ho2009insurer}, the consequences of hospital mergers \citep{gowrisankaran2015mergers}, and the effects of reduced insurer competition on welfare, hospital prices, and premiums \citep{ho2017insurer}, among others. However, this line of research generally analyzes bargaining dynamics while taking the form of insurance contracts as given, since its aim is not to characterize optimal contract design while accounting for price feedback between insurers and hospitals---one of this paper’s key objectives. By contrast, the contract theory and mechanism design literature, although centered on optimal insurance design, has paid relatively little attention to this interaction. The main contribution of this paper is to develop a mechanism design framework that characterizes the types of contracts insurers optimally offer when accounting for their price impact.

More precisely, this paper contributes to two main strands of literature. The first is the extensive theoretical body of work on insurance. A large literature following \citet{borch1960safety} and \citet{arrow1963welfare} studies welfare-maximizing insurance policies in the absence of information frictions. For example, \citet{raviv1979design} shows that coinsurance arrangements are welfare maximizing when the insurer faces an exogenous cost of provision that is increasing and convex in coverage.

This line of work does not incorporate adverse selection. Models with information frictions follow the seminal contributions of \citet{rothschild1978equilibrium} and \citet{stiglitz1977monopoly}. Particularly relevant is the recent work of \citet*{gershkov2023optimal}, who revisit the classical monopoly insurance problem under adverse selection. Their framework departs from the canonical model of \citet{stiglitz1977monopoly} in two key respects.

First, they allow for general loss distributions that may be correlated with agents’ private information. Second, they assume agents are risk-averse with preferences represented by a dual utility function \citep{yaari1987dual}. Unlike the expected utility framework, dual utility captures risk aversion through a non-linear transformation of probabilities rather than payoffs. As shown by \cite{yaari1987dual}, this representation arises from a single deviation from expected utility theory: rather than assuming independence concerning probabilistic mixtures of lotteries, it assumes independence concerning mixtures of payments associated with those lotteries.

\cite{gershkov2023optimal} argue that dual utility is a compelling framework for several reasons. First, empirical studies \citep{sydnor2010over,barseghyan2013nature,barseghyan2016inference} document insurance choices that are difficult to reconcile with expected utility but align well with models incorporating probability weighting---a feature inherent to dual utility. Second, dual utility disentangles risk aversion from wealth effects, allowing the framework to capture risk aversion while preserving constant marginal utility of wealth, a property that enhances tractability relative to traditional frameworks. As a result, \cite{gershkov2023optimal} can rationalize real-world contract features---such as deductibles---that expected utility models often struggle to explain.

Motivated by these insights, we also endow agents with Yaari risk preferences. However, relative to \citet{gershkov2023optimal} and the broader theoretical literature, our focus is different: we study how the structure of optimal insurance contracts changes once we account for the impact of coverage on third-party service prices. To capture this tradeoff, we analyze a fundamentally different model. First, monetary losses are not exogenous but arise from adverse events that can only be treated by purchasing third-party services. This makes insurance effective only if it stimulates utilization, thereby creating the feedback loop between coverage, service demand, and market prices. Second, contracts shape demand for services, which must match supply through market clearing. This introduces an equilibrium constraint in addition to incentive constraints, forcing the insurer to internalize how contract design translates into service prices. Third, agents possess two dimensions of private information revealed sequentially, which captures richer selection dynamics and broadens the space of feasible contract forms. Finally, recognizing the insurer’s price impact allows us to provide an economic rationale for limited-coverage deductible contracts. Together, these features highlight that equilibrium effects fundamentally reshape optimal insurance design and help explain contract structures observed in practice.

The second strand of literature we contribute to is the growing work on partial mechanism design, where designers lack full control over the primary market and must consider the effects of their mechanisms on secondary trading venues. This framework has been applied to financial markets \citep{philippon2012optimal,tirole2012overcoming}, monopoly pricing with resale \citep{calzolari2006monopoly,loertscher2022monopoly}, information transmission to aftermarkets \citep{calzolari2006optimality,dworczak2020mechanism}, optimal redistribution \citep{kang2023public,kang2024optimal}, contracting between firms \citep{calzolari2015exclusive,kang2022contracting}, and policy interventions based on market outcomes \citep{valenzuela2020market}. We extend this literature by applying partial mechanism design to insurance markets with risk-averse agents holding non-standard preferences, where third-party service prices are endogenously determined by the insurer’s mechanism and the available supply. Technically, our framework differs from prior studies in one key respect: we address the insurer’s price impact within a sequential screening framework, similar to \citet{courty2000sequential}, where mechanisms resemble “menus of menus.” To tackle this complexity, we reformulate the problem using an optimal transport approach.

The remainder of the paper is organized as follows. Section~\ref{sec:model} introduces the model. Section~\ref{sec:first-stage} defines the first-stage problem and presents our first main result, Theorem~\ref{thm:optimal-menu}. Section~\ref{sec:second-stage} turns to the second-stage problem and presents our second main result, Theorem~\ref{thm:envelope}. Section~\ref{sec:extensions} discusses extensions of the baseline model, and Section~\ref{sec:conclusion} concludes.

\section{Model}\label{sec:model}
\paragraph{Sequential Private Information.}
A unit mass of agents faces the risk of experiencing an adverse event---such as developing an illness or being involved in a car accident---that may require third-party services tomorrow. Each agent is characterized by a privately observed ex-ante \textit{risk type} $\theta \in \Theta = [\underline{\theta}, \bar{\theta}]$, which parametrizes the distribution of her also privately observed ex-post loss-related \textit{valuation} $b \in B = [0, \bar{b})$, where $\bar{b}$ may be finite or infinite. Specifically, an agent with risk type $\theta$ draws her valuation from the cumulative distribution function $H_\theta: B \to [0,1]$.

An agent who draws valuation $b$ experiences a ex-post payoff of $-b$ (i.e., a \textit{loss} of $b$). We interpret $b$ as the monetary loss the agent associates with the negative consequences of the adverse event in the absence of treatment or restoration services.

We assume that $H_\theta$ increases in $\theta$ in the sense of first-order stochastic dominance: higher risk types are more likely to draw higher valuations. The distribution of risk types is denoted by $F: \Theta \to [0,1]$, with continuous density function $f: \Theta \to (0, \infty)$.

We impose the following regularity conditions. The function $H_\theta$ is continuously differentiable in $\theta$, and for each fixed $\theta$, $H_\theta(\cdot)$ is continuously differentiable on $B$, with a strictly positive derivative $h_\theta: B \to (0, \infty)$. Moreover, for each $(\theta, b)$, the cross partial derivative $\frac{\partial^2 H_\theta(b)}{\partial b , \partial \theta}$ exists and is continuous.

\begin{example}\label{ex:distribution}
    Suppose that risk type $\theta$ represents the probability of developing an illness. Conditional on becoming ill, the agent's loss-related valuation is uniformly distributed on $[0,1]$. In this case, the distribution function $H_\theta$ is given by:
    \[
    H_\theta(b) = 1 - \theta + \theta b.
    \]
\end{example}

\paragraph{Third-Party Services and the Agent's Ex-Post Payoff.}
After drawing a valuation $b$, the agent can visit the service provider to receive treatment or repair damages at a price $p \in \mathbb{R}_+$. In that case, her loss becomes $p$ instead of $b$. Therefore, the agent seeks the service whenever $b > p$. That is, whenever the monetary loss she associates with the adverse event in the absence of of the service exceeds the service price.

We also allow for the possibility of \textit{partial treatment}, modeled as the agent’s ability to purchase any fraction $x \in [0,1]$ of the service for a total cost of $xp$. In other words, the service provider uses linear pricing.

An agent with valuation $b$ who chooses to consume $x$ units of service incurs a loss of $xp + (1 - x)b$. This loss is additive in two components. The first is the monetary payment of $xp$ made to the provider for the $x$ units of service. The second is the residual loss $(1 - x)b$, which reflects the untreated portion of the adverse event. Since the agent forgoes $1 - x$ units of service and values the full untreated event at $b$, we assume that these effects also interact linearly.

Fractional service becomes relevant only with insurance; without it, it is inconsequential. If $b \leq p$, any $x > 0$ yields a loss of $(1 - x)b + xp \geq b$, so the agent prefers $x = 0$. If $b > p$, any $x < 1$ yields a loss strictly above $p$, so the agent chooses $x = 1$. Thus, the agent’s loss is $b_p = \min\{b, p\}$.

Finally, we assume that the market for services is competitive and characterized by a strictly increasing supply function, denoted by $S: \mathbb{R}_+ \to \mathbb{R}_+$, where $S(p)$ represents the quantity of services supplied at price $p$. Additionally, we impose the following technical assumptions: $S$ is continuously differentiable, unbounded and satisfies $S(0)=0$.

\paragraph{The Agent's Ex-Ante Payoff.}
As in \citet{gershkov2023optimal}, agents are endowed with a \textbf{Yaari (dual) utility} determined by a probability distortion function $g : [0,1] \to [0,1]$, where $g$ is increasing, continuously differentiable, and satisfies $g(q) \leq q$ with $g(0) = 0$ and $g(1) = 1$.

The \textit{certainty equivalent} of suffering a random loss distributed according to $H_{\theta}$ is given by:
\[
    CE(H_{\theta}(b))=-\int_{0}^{\bar{b}} [1 - g(H_{\theta}(b))] \, db.
\]
In particular, if services are provided at price $p$, so that the agent's loss is given by $b_p=\min\{b,p\}$, then the certainty equivalent is:
\[
   -\int_{0}^{p} [1 - g(H_{\theta}(b))] \, db,
\]
where we use the fact that $\mathbb{P}(b_p \leq b) = H_{\theta}(b)$ for all $b < p$, $\mathbb{P}(b_p \leq b) = 1$ for all $b \geq p$, and $g(1) = 1$.

As discussed in \citet{gershkov2023optimal}, the agent is risk-averse in the weak sense: the certainty equivalent of any lottery is less than the lottery's expected value. That is:
\[
    CE(H_{\theta}(b))=-\int_{0}^{\bar{b}} [1 - g(H_{\theta}(b))] \, db \leq -\int_{0}^{\bar{b}} [1 - H_{\theta}(b)] \, db = -\mathbb{E}_{H_\theta}[b],
\]
where the inequality follows from $g(q) \leq q$, and $\mathbb{E}_{H_\theta}[b]$ represents the expected value of $b$ under $H_{\theta}$.

While in the expected utility framework, this form of risk aversion is equivalent to aversion to mean-preserving spreads, in the dual utility framework, aversion to mean-preserving spreads is stronger and is equivalent to $g$ being convex. As in \citet{gershkov2023optimal}, the weak form of risk aversion, which requires only $g(q) \leq q$, is sufficient for our analysis.

To better understand dual utility, we can use integration by parts to express the certainty equivalent of the lottery $H_{\theta}$ as:
$$CE(H_{\theta}(b))=-\int_{0}^{\bar{b}} [1 - g(H_{\theta}(b))] \, db = -\int_{0}^{\bar{b}} b g'(H_{\theta}(b)) h_{\theta}(b) \, db$$
In contrast to the expected utility framework—where the standard expectation operator is modified by distorting payoffs through the von Neumann–Morgenstern utility function—dual utility modifies the expected value of the lottery by distorting probabilities through $g'$. Here, every loss $b$ is weighted by $g'(H_{\theta}(b))$. The term $g(H_{\theta}(b))$ represents the cumulative weight assigned to losses below $b$, while $1 - g(H_{\theta}(b))$ represents the cumulative weight assigned to losses above $b$. The assumption $g(q) \leq q$ implies that the agent overweights the cumulative probability of larger losses and underweights the cumulative probability of smaller losses.

Figure~\ref{figCE} illustrates the certainty equivalents of a risk-neutral agent and a risk-averse agent with Yaari utility given by $g(q) = q^2$, when both face the distribution $H_{\theta}$ described in Example~\ref{ex:distribution}, and $p$ is the price of services. The blue curve represents the distribution function $H_\theta(b)$ truncated at $p$, while the red curve shows $g(H_\theta(b))$, also truncated at $p$. The area above the blue curve corresponds to the certainty equivalent of the risk-neutral agent. The area above the red curve and below the horizontal dashed line corresponds to the certainty equivalent of the risk-averse agent. The area between the red and blue curves reflects the additional amount the risk-averse agent is willing to pay to eliminate the lottery risk.

\begin{figure}[H]
\begin{center}
\begin{tikzpicture}[scale=0.6]
    \begin{axis}[
        axis lines = middle,
        xlabel style={at={(current axis.right of origin)}, anchor=north, xshift=1.9cm, yshift=-0.05cm},
        ylabel style={at={(current axis.above origin)}, anchor=north west, xshift=-1.8cm, yshift=2.4cm},
        xlabel = \( b \),
        ylabel = {\( g(H_{\theta}(b)) \), \( H_{\theta}(b) \)},
        ymin=0, ymax=1.2,
        xmin=0, xmax=1.2,
        xtick={0,0.7,1},
        ytick={0,1},
        xticklabels={\(0\), \(p\), \(\bar{b}\)},
        yticklabels={\(0\), \(1\)},
        domain=0:1,
        samples=100,
        width=10cm, height=8cm,
        grid = none,
        smooth,
    ]

    % Plot the curve g(H^p_theta(b))
    \addplot [
        name path=curve,
        domain=0:0.7,
        samples=100,
        thick,
        red,
    ]
    {(1 - 0.7 + 0.7*x)^2};

    % Path for the horizontal line at y = 1
    \path[name path=topline] (axis cs:0,1) -- (axis cs:0.7,1);

    % Shading the area above the curve
    \addplot [
        fill=red,
        fill opacity=0.2,
    ]
    fill between[
        of=curve and topline,
    ];

    % Horizontal line for b > p
    \addplot [
        domain=0.7:1,
        samples=100,
        thick,
        red,
    ]
    {1};

    % Dashed lines
    \draw[dashed, thick] (axis cs:0,1) -- (axis cs:0.7,1);
    \draw[dashed, thick] (axis cs:0.7,0) -- (axis cs:0.7,1);
    
    % Points
    \filldraw[black] (axis cs:0.7,1) circle (3pt);
    \filldraw[white, thick] (axis cs:0.7,0.6241) circle (3pt);
    \draw[black] (axis cs:0.7,0.6241) circle (3pt);

    % Add the additional H^p_theta(b) line
    \addplot [
        name path=Hp,
        domain=0:0.7,
        samples=100,
        thick,
        blue,
    ]
    {(1 - 0.7 + 0.7*x)};

    % Horizontal part of H^p_theta(b) for b > p
    \addplot [
        domain=0.7:1,
        samples=100,
        thick,
        blue
    ]
    {1};

    % Shading between g(H^p_theta(b)) and H^p_theta(b)
    \addplot [
        blue,
        fill=blue,
        fill opacity=0.2,
    ]
    fill between[
        of=Hp and topline,
    ];

    \end{axis}
\end{tikzpicture}
\end{center}
\caption{$H_\theta(b)=1-\theta+\theta b$, \ $\theta=0.7$, \ $g(q)=q^2$} \label{figCE}
\end{figure}

Finally, another key property of Yaari utility is its additivity with respect to constant random variables $t$:
\[
    CE(H_{\theta} + t) = CE(H_{\theta}) + t.
\]
As noted by \citet{yaari1987dual}, in the expected utility framework, the agent's attitude toward risk and their attitude toward wealth are inherently linked. Specifically, a risk-averse agent will always exhibit decreasing marginal utility with respect to wealth. In contrast, under dual utility, an agent can maintain constant marginal utility over wealth while still displaying risk aversion toward uncertain outcomes. 

\paragraph{Insurance.}
A risk-neutral monopolistic insurance provider (he) offers a mechanism to the agents. The mechanism is offered when an agent knows only her ex-ante risk type $\theta$ before her ex-post valuation $b$ has been realized. We assume that the mechanism specifies the following for each agent:
\begin{enumerate*}[label=\arabic*)]  
\item the \textit{quantity of services}, $x$, covered by insurance;  
\item the \textit{per-unit price}, or \textit{per-unit out-of-pocket cost} (OPC), denoted by $D$, that the agent must pay for the fraction of services covered by insurance; and
\item an upfront payment or \textit{premium}, $t$, that an agent must pay to access her insurance.  
\end{enumerate*}  

The insurer can set terms only for the insured units; it cannot dictate the agent’s consumption of additional units at the prevailing market price. Agents remain free to purchase extra units at the market price, and the insurer cannot restrict or condition access to services (for instance, by paying the agent not to seek treatment).

Additionally, we restrict our attention to:  
\begin{enumerate*}[label=\arabic*)]  
\item direct mechanisms and  
\item non-randomized mechanisms.  
\end{enumerate*}  
Focusing on direct mechanisms is without loss of generality, as the revelation principle applies in our setting. While \citet{gershkov2023optimal} also study non-randomized mechanisms, they note that this restriction is not without loss: lotteries over insurance contracts could help the provider better screen risk-averse agents. However, focusing on non-randomized mechanisms is appropriate for insurance applications, as lotteries over contracts are rarely observed in practice.

More formally, a mechanism (or \textit{menu}) is represented as: 
\[
    \left\{\left(\left\{\left(x(\theta,b),D(\theta,b)\right)\right\}_{b \in B}, t(\theta)\right)\right\}_{\theta \in \Theta},
\]  
where each \textit{menu option} includes a \textit{contract} $\left\{\left(x(\theta,b),D(\theta,b)\right)\right\}_{b \in B}$ and a \textit{premium} $t(\theta) \in \mathbb{R}$. Upon visiting the service provider, an agent with the contract $\left\{\left(x(\theta,b),D(\theta,b)\right)\right\}_{b \in B}$ can choose any of its items. If the agent selects the \textit{contract item} $\left(x(\theta,b),D(\theta,b)\right)$, insurance covers $x(\theta,b) \in [0,1]$ units of service for a per-unit price or OPC of $D(\theta,b) \in \mathbb{R}_+$. Equivalently, a contract item specifies $x(\theta,b)$ and the total out-of-pocket payment the agent would incur for those units, given by $\hat{D}(\theta,b)=x(\theta,b)D(\theta,b)$. For reasons that will become clear later in the analysis, however, it is more convenient to work with the per-unit out-of-pocket cost.

We present two examples of insurance contracts within our framework that share key features with real-world plans and will become relevant in our later analysis:

\begin{itemize}
    \item \textit{Simple deductible contract} $\{(1, D)\}$: This contract is analogous to standard deductible-based insurance. If the agent pays the OPC $D$, she receives the full unit of services, and the insurer covers the remaining $p - D$.
    
   \item \textit{Limited coverage deductible contract} $\{(\alpha, D), (1, M)\}$: Under this contract, insurance covers up to a fraction $\alpha$ of the service at an OPC of $D$. If the agent chooses to consume the remaining $1 - \alpha$ units, she must pay the per-unit service price $p$. Alternatively, by paying a higher OPC $M > D$, the agent can bypass the coverage cap $\alpha$, receive the full unit of service, and have the insurer cover the remaining amount, $p - M$.
\end{itemize}

\paragraph{Timing.} We now describe in more detail the timing of the interaction between the insurer and the agents.

\begin{enumerate}
    \item The insurer commits to and offers the agents a menu:
    \[
        \left\{\left(\left\{\left(x(\theta,b),D(\theta,b)\right)\right\}_{b \in B}, t(\theta)\right)\right\}_{\theta \in \Theta}.
    \]
    \item The agents privately observe their ex-ante risk type $\theta$, decide whether to participate in the mechanism and, if so, they select an option from the menu, and pay the corresponding premium.
    \item The ex-post loss-related valuation $b$ is realized and privately observed by the agent. An agent holding the contract $\left\{\left(x(\theta,b), D(\theta,b)\right)\right\}_{b \in B}$ decides whether to visit the service provider or remain at home. If she visits the provider, she may demand any quantity of service. Payment can be made either by paying the per-unit service price $p$ or by selecting a contract item. If the agent selects $\left(x(\theta,b), D(\theta,b)\right)$, she is entitled to $x(\theta,b)$ units at a per-unit price of $D(\theta,b)$. For any additional units beyond $x(\theta,b)$, the full per-unit price $p$ must be paid.
    \item The contracts chosen by the agents, along with the items selected from the contracts, generate a demand for services at a given price. The market for services clears at the price where demand equals supply, and payoffs are accrued.
\end{enumerate}

The objective of the insurer is to find a menu that maximizes his profits. However, since a menu influences the equilibrium service price, which in turn affects both the agents' behavior and the insurer's profits, we adopt a two-stage approach to determine the optimal menu:
\begin{enumerate*}[label=\arabic*)]
\item In the first stage, we fix a price for services, $p$, and identify, among the menus that sustain that price in equilibrium, the one that maximizes the insurer's profits.
\item Once the first stage is complete, we associate each achievable price in equilibrium with the corresponding profits obtained by the insurer, and then solve for the price that maximizes these profits.
\end{enumerate*}

\subsection{Discussion of the Model}
\paragraph{Ex-Post Private Information.} In our model, $b$ represents the agent’s willingness to pay for third-party services. Since this valuation is subjective and shaped by many personal and contextual factors, we assume that $b$ is private information.

One important determinant of $b$ is the \textit{severity} of the agent’s condition: a more severe illness or greater property damage typically generates higher perceived losses and, consequently, higher valuations for treatment or restoration services. If higher values of $b$ correspond to more serious events, it is plausible that service costs are also higher.

In Section~\ref{sec:multiple-prices}, we present an extension of our baseline model that captures this idea while preserving our main results. The basic setup in this extension is that the third-party provider offers tiered services: a basic one (i.e., a routine checkup) at price $p_1$, and a more advanced one (i.e., a specialist visit) at price $p_2$. Only agents with sufficiently high valuations choose the advanced service in addition to the basic one and therefore face higher treatment costs. Since the service received is verifiable by the insurer, he can condition contracts on the type of service provided.

\paragraph{Ex-Post and Ex-Ante Payoff.} An important assumption that brings tractability to the model is the linearity of the agent’s ex-post payoff in the quantity of services $x$. As we will show later, this assumption plays a crucial role in ensuring that the insurer’s objective is linear in $x$, which allows us to represent the insurer’s problem as an optimal transport problem and to use duality results to characterize the optimal insurance contracts.

The linearity of the agent’s ex-post payoff arises from our assumption that the service provider uses linear pricing. As previously mentioned, in Section~\ref{sec:multiple-prices} we present an extension in which the provider offers multiple services, each with potentially different base prices. However, we maintain the assumption that, for any given service, partial treatment is possible and the provider charges linearly in the amount of service received.

Another key assumption that ensures linearity of the insurer’s objective in $x$ is that the agents are endowed with a Yaari (dual) utility. This implies that the agent’s ex-ante utility is also linear in $x$, since dual utility distorts probabilities rather than payoffs. Moreover, since dual utility is additive over constant random variables, we will show that premiums can be expressed as a linear function in $x$.

Beyond its analytical convenience, we believe that dual utility is compelling for several substantive reasons. A number of empirical studies (i.e., \citealp{sydnor2010over, barseghyan2013nature, barseghyan2016inference}) document patterns in insurance choice that are difficult to reconcile with expected utility theory but are well explained by models involving probability weighting.

\paragraph{Market Structure.} The assumption that the insurance provider is a monopolist with price impact is motivated by many settings where insurers hold significant market power. According to the latest report on competition in health insurance \citep{AMA2024CompetitionHealthInsurance}, the American Medical Association (AMA) finds that, based on the DOJ/FTC Merger Guidelines, 95\% of metropolitan commercial markets were highly concentrated in 2023, with a Herfindahl--Hirschman Index (HHI) above 1800. The average commercial market exhibited an HHI of 3458. Moreover, in 89\% of metropolitan markets at least one insurer held a market share of 30\% or more, and in 47\% of markets a single insurer’s share reached at least 50\%. 

Regarding the assumption of a competitive service market, in Section~\ref{sec:nash-bargaining} we show that our baseline model can also be interpreted as a setting where prices are determined not by competition, but through Nash bargaining between the insurer and the third-party provider, with the insurer holding all the bargaining power. We also discuss how our main results extend to cases where the provider has bargaining power.

\section{First Stage: Optimal Insurance Menus}\label{sec:first-stage}
Recall that in the first stage we fix the service price $p \in \mathbb{R}_+$, and our goal is to identify the menu of insurance contracts that sustains $p$ in equilibrium and maximizes the insurer's profits. We begin by presenting the first main result of the paper, which characterizes the form of profit-maximizing contracts for any given price $p$.

\begin{theorem}\label{thm:optimal-menu}
    For any price $p$, the optimal contract offered to risk type $\theta$ takes one of the following forms: a \textbf{simple deductible contract} $\{(1, D(\theta))\}$ with $D(\theta) \in [0,p]$, or a \textbf{limited-coverage deductible contract} $\{ (\alpha,D(\theta)),(1, M(\theta))\}$ where $0\leq D(\theta) < M(\theta) \leq p$ and $\alpha \in (0,1)$.
\end{theorem}

In the next sections, we present the proof of Theorem~\ref{thm:optimal-menu}. The argument proceeds in three steps:

\begin{enumerate}[label=\textbf{Step \arabic*:}, leftmargin=*]
    \item \textbf{Characterization of implementable menus.} We first characterize contracts for which it is ex-post incentive compatible for an agent with loss-related valuation $b$ to select the contract item corresponding to her true valuation. Building on this, we characterize menus for which it is ex-ante individually rational for type $\theta$ to select an option, and ex-ante incentive compatible to choose the option corresponding to her true type. We rely on standard techniques for this analysis, though the ex-post stage requires additional care: the agent must also decide whether to seek services at all and whether to pay using the contract or at the full service price. Finally, we define the equilibrium constraint that any menu sustaining $p$ must satisfy.
    
    \item \textbf{Rewriting the insurer’s problem in the quantity space.} With Step 1 complete, we can formally state the insurer’s problem. In this formulation, the insurer chooses an allocation function $x(\theta,\cdot)$ for each type $\theta$. Using standard techniques from mechanism design \citep{bulow1989simple}, we transform this into the quantity space, where the insurer directly chooses the quantity of services $q$ demanded by each risk type. We then illustrate this formulation with an example (Section~\ref{sec:example}) involving two risk types, which provides intuition for Theorem~\ref{thm:optimal-menu}.
    
    \item \textbf{Optimal transport formulation.} Finally, we address the technical challenges that arise in the infinite-dimensional type space. To handle these, we transform the problem into one that closely resembles an optimal transport problem, in which the insurer matches risk types to service quantities. Leveraging known duality results from optimal transport theory, we derive several properties of the solution, including Theorem~\ref{thm:optimal-menu}.
\end{enumerate}

\subsection{Implementable Menus}\label{sec:implementable}

\paragraph{Period 2 Incentive Compatibility.}
After an agent with risk type $\theta'$ selects an option from the menu and her valuation $b$ is realized, she makes a series of decisions. First, she chooses whether to visit the service provider or remain at home. If she decides to seek the service, she has two payment options: covering the service bill entirely out of pocket or utilizing an item from her contract $\left\{\left(x(\theta,b), D(\theta,b)\right)\right\}_{b \in B}$. The agent’s \textit{loss function} when engaging with the item $\left(x(\theta,b'), D(\theta,b')\right)$ is defined as follows:
\[
    L(b,b';\theta) = x(\theta,b') \min\{D(\theta,b'), b_p\} + (1 - x(\theta,b')) b_p,
\]
where $b_p = \min\{b, p\}$.

We now explain why this expression represents the agent's loss. The agent will seek the service and utilize the fraction of treatment covered by the contract item, $x(\theta,b)$, if the OPC $D(\theta,b')$ is less than $b_p$. In this case, insurance covers a fraction $x(\theta,b)$ of the service and the agent pays $x(\theta,b)D(\theta,b')$.

If $b$ is higher than $p$, the agent will also demand the remaining fraction of the service, $1 - x(\theta,b)$, at the full per-unit price $p$. The additional service cost is $(1 - x(\theta,b))p$, resulting in a total loss of $x(\theta,b)D(\theta,b') + (1 - x(\theta,b))p$.

However, if $b$ is lower than $p$, the agent will not demand the remaining portion, and her loss will instead be $x(\theta,b)D(\theta,b') + (1 - x(\theta,b))b$, where $x(\theta,b)D(\theta,b')$ is the monetary payment for the $x(\theta,b)$ units of service and $(1 - x(\theta,b))b$ is the residual loss, which reflects the untreated portion of the adverse event.

Conversely, if $D(\theta,b')$ is higher than $b_p$, two scenarios may arise: 
\begin{enumerate*}[label=\arabic*)]
    \item If the service price $p$ is lower than her valuation $b$, the agent prefers to get the full unit of the service and pay the entire service cost out of pocket. In this situation, her loss is $p$.
    \item If her valuation $b$ is lower than $p$, the agent prefers to remain at home, and her loss is $b$.
\end{enumerate*}

Note that whenever the agent decides whether to seek services or to pay using a contract item, she compares the per-unit price under insurance, $D(\theta,b')$, with $b_p$. For this reason, it is more convenient to work with $D(\theta,b')$ rather than the total out-of-pocket payment $\hat{D}(\theta,b)=x(\theta,b)D(\theta,b)$.

Finally, it is also important to note that the agent's risk type $\theta'$ is irrelevant in determining her loss function. At this stage, the only relevant variables are her actual valuation $b$, the contract chosen from the menu (indexed by $\theta$), and the item selected from the contract (indexed by $b'$). Figure \ref{figMH} illustrates the agent's behavior and her resulting loss when she selects the item from the contract corresponding to her true valuation.

\begin{figure}[H]
        \centering
        \includegraphics[width=0.7\linewidth]{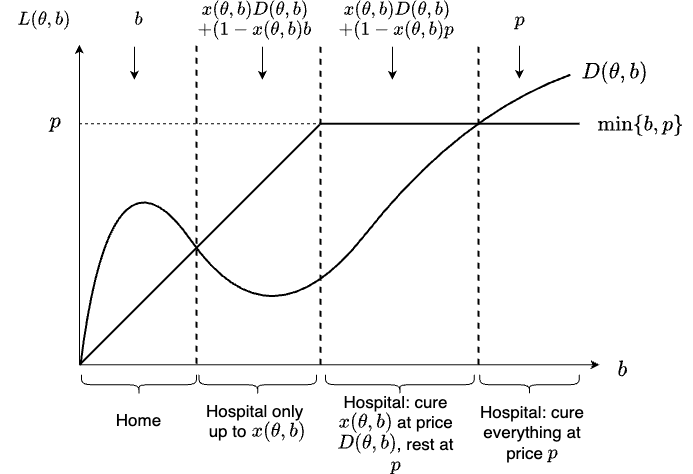}
        \caption{Loss and behavior of an agent that truthfully reports her valuation in the context of health insurance.}
        \label{figMH}
\end{figure}

Having introduced the agent's period 2 loss function, we can now define the usual incentive compatibility constraint that a contract must satisfy. Specifically, it suffices to focus on contracts where the agent truthfully reports her valuation:
\begin{equation}
    L(b;\theta) \equiv L(b,b;\theta) \leq L(b,b';\theta) \; \; \text{for all $b,b' \in B$.} \tag{$IC_2$}
    \label{IC2}
\end{equation}

We now show some properties of contracts that satisfy~\ref{IC2}:

\begin{lemma}\label{lemma:IR-contract}
    Without loss of generality we can restrict attention to contracts that satisfy $D(\theta,b) \leq b_p$.
\end{lemma}

\begin{lemma}\label{lemma:effective-b-type}
    A contract satisfies~\ref{IC2} only if $L(p;\theta)=L(b;\theta)$ for all $b > p$. Furthermore, without loss of generality, we can restrict attention to contracts that satisfy $x(\theta,b)=1$ for all $b > p$.
\end{lemma}
The formal proofs of Lemmas~\ref{lemma:IR-contract} and~\ref{lemma:effective-b-type} can be found in Appendix~\ref{sec:proof-implmentable}.

Lemmas~\ref{lemma:IR-contract} and~\ref{lemma:effective-b-type} are intuitive. For Lemma~\ref{lemma:IR-contract}, if $D(\theta,b) > b_p$, the agent will not use her insurance. The insurer can replicate this outcome by setting $D'(\theta,b) \leq b_p$ but restricting the agent's access to the service by setting $x'(\theta,b) = 0$.

For Lemma~\ref{lemma:effective-b-type}, note that all agents with $b \geq p$ incur the same loss from any allocation. These are the agents who would always seek the service in the absence of insurance, effectively behaving as if they have the same valuation. To satisfy~\ref{IC2}, all such agents must experience the same loss. 

Moreover, for any agent with a valuation $b > p$, and for any contract that results in a loss of $L(p; \theta)$ for such an agent, the insurer can replicate the same outcome by setting $x'(\theta,b) = 1$ and $D'(\theta,b) = L(p; \theta)$. This choice preserves the agent’s demand for services, their loss (or ex-post payoff), and the insurer’s cost.

Lemmas~\ref{lemma:IR-contract} and~\ref{lemma:effective-b-type} allow us to express the loss of an agent with $b \in [0,p]$ when truthfully reporting as:  
\[
    L(b;\theta) = x(\theta,b)D(\theta,b) + (1 - x(\theta,b))b.
\]  
This loss function is linear in $b$ and separable in the OPC. Using standard arguments from mechanism design, we can characterize contracts that satisfy~\ref{IC2} by applying the envelope formula and the monotonicity of $x(\theta,\cdot)$:  

\begin{lemma}\label{lemma:IC2-contract}
    A contract satisfies~\ref{IC2} if and only if:  
    \begin{enumerate}[label=(\arabic*)]
        \item \label{condition:envelope-contract} $L(b;\theta) = \int_0^{b}(1 - x(\theta,l))\,dl$ for all $b \in [0, p]$, and $L(b;\theta) = L(p;\theta)$ for all $b > p$.
        \item \label{condition:monotonicity-contract} $x(\theta,\cdot): B \to [0,1]$ is a non-decreasing function, with $x(\theta,b) = 1$ for all $b > p$.
    \end{enumerate}
\end{lemma} 
The formal proof of Lemma~\ref{lemma:IC2-contract} can be found in Appendix~\ref{sec:proof-implmentable}.

\paragraph{Period 1 Incentive Compatibility and Individual Rationality.}
By Lemma~\ref{lemma:IC2-contract}, we know that if an agent with risk type $\theta$ selects the menu option: 
\[
    \left(\left\{\left(x(\theta',b),D(\theta',b)\right)\right\}_{b \in B},t(\theta')\right),
\]
the agent’s loss function is given by:
\[
    L(b;\theta') = \int_0^b (1 - x(\theta',l)) \, dl.
\]
Observe that $L(\cdot;\theta')$ is non-decreasing. Consequently, the lottery faced by risk type $\theta$ in period 1 is represented by $H_{\theta}(L^{-1}(l;\theta'))$, where $L^{-1}(l;\theta) = \sup\{z : L(z;\theta) \leq l\}$. 

The certainty equivalent that risk type $\theta$ gets from her menu option is given by:
\begin{align*}
    U(\theta,\theta')&=-t(\theta')-\int_0^{L(\bar{b};\theta')}\left[1-g\left(H_{\theta}\left(L^{-1}(l;\theta')\right)\right)\right] \, dl \\
    &=-t(\theta')-\int_0^{\bar{b}}\frac{\partial L(b;\theta')}{\partial b}[1-g(H_{\theta}(b))] \, db\\
    &=-t(\theta')-\int_0^{p}(1-x(\theta',b))[1-g(H_{\theta}(b))]db.
\end{align*}
The additivity of the premium arises from the fact that dual utility is additive in constant random variables. The second equality follows from the change of variable $b=L^{-1}(l;\theta)$, while the third equality results from $\frac{\partial L(b;\theta')}{\partial b}=1-x(\theta',b)$ (Lemma~\ref{lemma:IC2-contract}, Condition~\ref{condition:envelope-contract}), along with the condition $x(\theta',b)=1$ for $b>p$ (Lemma~\ref{lemma:IC2-contract}, Condition~\ref{condition:monotonicity-contract}).

The certainty equivalent is linear in $x$, a property that provides crucial tractability for the analysis. This linearity arises from the fact the agent's ex-post payoff is linear in $x$ and the fact that dual utility distorts probabilities rather than payoffs. 

If the agent were instead a risk-averse expected utility maximizer, two important aspects would change. First, because the agent’s loss $L(b; \theta)$ is evaluated through a concave von Neumann–Morgenstern utility function, the certainty equivalent would no longer be linear in $x$. Second, since this distortion also applies to the premium, $t$ would no longer be separable from $x$.

The period 1 incentive compatibility constraint is:
\[
    U(\theta) \equiv U(\theta,\theta) \geq U(\theta,\theta') \; \; \text{for all $\theta, \theta' \in \Theta$,} \tag{$IC_1$} \label{IC1}
\]

Lemma~\ref{lemma:IC1-menu} provides a partial characterization of menus that satisfy~\ref{IC1}:

\begin{lemma}\label{lemma:IC1-menu} \leavevmode
    \begin{enumerate}[label=(\arabic*)]
        \item \label{condition:IC1-necessity} A menu satisfies~\ref{IC1} only if 
        \[
            U(\theta) = U(\underline{\theta}) + \int_{\underline{\theta}}^{\theta}\int_0^{p}(1-x(s,b))g'(H_s(b))\frac{\partial H_s(b)}{\partial s} \, db \, ds. 
        \]
        \item \label{condition:IC1-sufficiency} If a menu satisfies the above condition and $x(\cdot,b):\Theta \to [0,1]$ is a non-decreasing function for all $b \in B$, then the menu satisfies~\ref{IC1}.
    \end{enumerate}
\end{lemma}
The formal proof of Lemma~\ref{lemma:IC1-menu} is in Appendix~\ref{sec:proof-implmentable}.

Condition~\ref{condition:IC1-necessity} of Lemma~\ref{lemma:IC1-menu} is the standard envelope formula. Importantly, the additivity of dual utility in constant random variables, combined with Condition~\ref{condition:IC1-necessity}, enables us to determine premiums (up to a constant) as a function of $x$. Additionally, the monotonicity of $x(\cdot,b)$ for all $b$, together with Condition~\ref{condition:IC1-necessity}, ensures that a menu satisfies~\ref{IC1}.

However, a menu that satisfies Condition~\ref{condition:IC1-necessity} does not necessarily require $x(\cdot,b)$ to be monotonic for all $b$. This result is common in screening environments where the allocation chosen by the designer for a fixed agent type is infinite---or multi-dimensional. Nevertheless, monotonicity remains a robust condition, as it does not depend on the functional forms of $g$ or $H_{\theta}$.

It remains to define the period 1 participation constraint. We begin by specifying the agent’s outside option (non-participation payoff), denoted $U_{NP}(\theta)$ for risk type $\theta$. We assume that an uninsured agent faces a service price of $\beta + p$ with $\beta \ge 0$; that is, the uninsured pays a weakly positive (possibly zero) premium $\beta$ over the baseline price $p$. This reflects empirical evidence that uninsured individuals typically face higher prices: for example, \citet{anderson2007soak} document that uninsured hospital patients are charged, on average, 2.5 times more than insured patients for the same services. It also captures a less visible but important benefit of insurance—access to lower prices \citep{cook2020uninsured}. 

\[
U_{NP}(\theta) = -\int_0^{\beta+p} [1 - g(H_\theta(b))]\,db.
\]

Note that the insurer can always offer a contract with zero premium and no coverage, ensuring $U(\theta) \geq U_{NP}(\theta)$. Therefore, it is without loss of generality to restrict attention to mechanisms in which all agents participate in period 1. This leads to the following period 1 individual rationality constraint:

\[
U(\theta) \geq U_{NP}(\theta) \quad \text{for all } \theta \in \Theta. \tag{$IR$} \label{IR}
\]

%Regardless of which specification of $U_{NP}(\theta)$ applies, an incentive compatible mechanism satisfies~\ref{IR} if and only if it binds for the lowest type:

Then, we can show that an incentive compatible mechanism satisfies~\ref{IR} if and only if it binds for the lowest type:
\begin{lemma}\label{lemma:IR-menu}
    Suppose a menu satisfies~\ref{IC1}. Then, the~\ref{IR} condition is satisfied for all types if and only if it is satisfied for the lowest type, $\underline{\theta}$.
\end{lemma}
The formal proof of Lemma~\ref{lemma:IR-menu} is in Appendix~\ref{sec:proof-implmentable}.

Lemma~\ref{lemma:IR-menu} follows from the fact that a higher-risk type faces a riskier lottery and, therefore, has a lower certainty equivalent. This implies that both the agent’s outside option, $U_{NP}$, and their utility from participation, $U$, decrease in $\theta$. However, $U$ decreases more slowly, since the agent receives some coverage (i.e., $x(\theta,b) \in [0,1]$). Therefore, $U'(\theta) - U'_{NP}(\theta) \geq 0$, and
\[
U(\theta) - U_{NP}(\theta) 
= U(\underline{\theta}) - U_{NP}(\underline{\theta}) 
+ \int_0^{\theta} [U'(s) - U'_{NP}(s)] \, ds 
\geq 0,
\]
whenever~\ref{IR} holds for type $\underline{\theta}$.

\paragraph{Period 3 Equilibrium Constraint.}
If each contract in a menu satisfies~\ref{IC2} and the menu satisfies~\ref{IC1} and~\ref{IR}, the demand for services by an agent with risk type $\theta$ and valuation $b$ is $x(\theta,b)$. This follows because the menu and its contracts are both incentive compatible and individually rational, ensuring that the agent selects the contract item $\left(x(\theta,b), D(\theta,b)\right)$. 

Moreover, if $b \leq p$, the agent will not demand any additional units of the service, and her demand is fully determined by $x(\theta,b)$. Conversely, if $b > p$, she would, in principle, demand the remaining portion, $1 - x(\theta,b)$. However, by Lemma~\ref{lemma:effective-b-type}, $x(\theta,b) = 1$, meaning the remaining portion is zero. Thus, her demand is again fully determined by $x(\theta,b)$.

The demand for services at price $p$ is given by:  
\begin{align*}
    \mathcal{D}(x;p) &= \int_{\underline{\theta}}^{\bar{\theta}} \int_0^{\bar{b}} x(\theta,b)h_{\theta}(b)\,dbf(\theta) \, d\theta \\
           &= \int_{\underline{\theta}}^{\bar{\theta}} \int_{0}^{p} x(\theta,b)h_{\theta}(b)\,db f(\theta) \, d\theta 
           + \int_{\underline{\theta}}^{\bar{\theta}} [1 - H_{\theta}(p)] f(\theta) \, d\theta,
\end{align*}
where the second equality follows from the fact that $x(\theta,b) = 1$ for all $b > p$.

It is convenient to define the residual supply of services faced by the insurer at price $p$ as:  
\[
    RS(p) = S(p) - \int_{\underline{\theta}}^{\bar{\theta}} [1 - H_{\theta}(p)] f(\theta) \, d\theta,
\]
where the second term reflects the portion of demand that the insurer cannot influence, as no matter what menu the insurer offers, he cannot affect the demand arising from valuations higher than $p$. 

The equilibrium constraint can then be expressed as:  
\[
    \int_{\underline{\theta}}^{\bar{\theta}} \int_{0}^{p} x(\theta,b)h_{\theta}(b)\,db f(\theta) \, d\theta = RS(p). \tag{$EQ$} \label{EQ}
\]

The lowest demand that the insurer can generate occurs with a menu that sets $\underline{x}(\theta, b) = 0$ for all $\theta$ and $b \in [0, p]$. In this case, the insurer provides no coverage to any agents, and demand is given by:
\[
    \mathcal{D}(\underline{x};p) = \int_{\underline{\theta}}^{\bar{\theta}} [1 - H_{\theta}(p)] f(\theta) \, d\theta.
\]

Moreover, by our technical assumptions on the supply function $S$ and the distributions of valuations $H_{\theta}$, there exists a unique strictly positive price $p^N$ that satisfies $\mathcal{D}\left( \underline{x};p^N\right) = S\left(p^N\right)$. That is, $p^N$ is the price that clears the service market in the absence of coverage.

Similarly, the highest demand that the insurer can generate occurs with a menu that sets $\overline{x}(\theta, b) = 1$ for all $\theta$ and $b \in [0, p]$. In this case, the insurer provides full insurance---or full access to services at zero OPC. Demand in this case is $\mathcal{D}\left(\overline{x};p\right) = \int_{\underline{\theta}}^{\bar{\theta}}[1-H_{\theta}(0)]f(\theta)\,d\theta$, where $1-H_{\theta}(0)$ is the probability that a risk type $\theta$ experiences an adverse event. Once again, by our technical assumptions, there exists a unique positive price $p^F$ such that $\mathcal{D}\left( \overline{x};p^F\right) = S\left(p^F\right)$, and $p^F > p^N$.

We then have the following result:

\begin{lemma}\label{lemma:feasible-prices}
    There exists a function $x$ that satisfies Condition~\ref{condition:monotonicity-contract} of Lemma~\ref{lemma:IC2-contract} and~\ref{EQ} if and only if $p \in \left[p^N, p^F\right]$.
\end{lemma}
The formal proof of Lemma~\ref{lemma:feasible-prices} is provided in Appendix~\ref{sec:proof-implmentable}.

From now on, we assume that $p$ belongs to $\left[p^N, p^F\right]$, ensuring that there exists a feasible menu that induces it as an equilibrium price in the service market.

\subsection{Insurer's Problem}\label{sec:insurer-problem}
We are now ready to formulate the first-stage insurer's problem. To do so, the following Lemma provides an expression for the insurer's objective:
\begin{lemma}\label{lemma:objective}
    Suppose that each contract in a menu satisfies~\ref{IC2}, and the menu satisfies~\ref{IC1} and~\ref{IR}. Then, the insurer's profits from the menu are:
    \[
        \int_{\underline{\theta}}^{\bar{\theta}} \int_{0}^{p} x(\theta,b) J(\theta,b) h_{\theta}(b)\, db f(\theta) \, d\theta 
        - U(\underline{\theta}) - \int_{0}^{p} [1 - g(H_{\underline{\theta}}(b))] \, db,
    \]
    where
    \[
        J(\theta,b) = \underbrace{\left[1 - g(H_{{\theta}}(b)) + \frac{1 - F(\theta)}{f(\theta)} g'(H_{\theta}(b)) \frac{\partial H_{\theta}(b)}{\partial \theta}\right]}_{\text{Period 1 Virtual Value}}\frac{1}{h_{\theta}(b)} +\underbrace{b - p - \frac{1 - H_{\theta}(b)}{h_{\theta}(b)}}_{\text{Period 2 Virtual Value}}.
    \]
\end{lemma}
The formal proof of Lemma~\ref{lemma:objective} is in Appendix~\ref{sec:proof-insurer-problem}.

The insurer's cost of providing any contract is derived from the fact that the portion of the service bill paid by the insurer, when facing an agent with risk type $\theta$ and valuation $b$, is $x(\theta,b)(p - D(\theta,b))$. Using the identity $x(\theta,b)D(\theta,b) = L(b;\theta) - (1 - x(\theta,b))b$ and the expression for $L(b;\theta)$ obtained in Lemma~\ref{lemma:IC2-contract}, Condition~\ref{condition:envelope-contract}, along with integration by parts, allows us to determine the insurer's cost.

Once we have an expression for the insurer's cost, the only remaining step to determine profits is to find an expression for the revenue the insurer receives from premiums. To obtain this, we can apply standard techniques. Specifically, as mentioned earlier, since Yaari utility is additive in constant random variables, $t$ can be expressed as a function of $x$, up to a constant $U(\underline{\theta})$.

To interpret the insurer's profits, recall that in screening environments with quasi-linear preferences, where a designer allocates an object or good, the classical agent's virtual value is a key concept in determining whether it is profitable to assign the good to the agent. It is defined as the marginal value from an increase in the allocation minus the information rents. Information rents, in turn, are captured by the inverse hazard rate multiplied by the derivative of the agent's marginal value with respect to her type.

In period 1, the agent’s private information, or risk type, is denoted by $\theta$, and her utility from an allocation, $x$, is given by $U(\theta)$. Based on this utility, we see that the agent values an additional unit of hospital services (if her future valuation is $b$) by $1 - g(H_{\theta}(b))$. This valuation forms the first term in $J(\theta, b)$, referred to as the \textit{period 1 virtual value}. The derivative of this marginal valuation with respect to the agent's risk type is given by:  
\[
-g'(H_{\theta}(b)) \frac{\partial H_{\theta}(b)}{\partial b}.
\]  
This term is then multiplied by the inverse hazard rate of the distribution $F$.

In period 2, the agent’s private information shifts to her valuation, denoted by $b$, and her payoff from an allocation is determined by the loss function $L(b; \theta)$. Here, the second term in $J(\theta, b)$, denoted as the \textit{period 2 virtual value}, is also analogous to the classical virtual value in screening problems but accounts for provision costs. Specifically, the agent's marginal valuation for an additional unit of services corresponds to her valuation $b$. On the other hand, the insurer’s marginal cost for providing this additional unit is $p$.

The \textit{(relaxed) first-stage insurer's problem} is given by:
\begin{align*}
    \Pi(p)\equiv&\max_{x} \quad \int_{\underline{\theta}}^{\bar{\theta}}\int_{0}^{p}x(\theta,b)J(\theta,b)h_{\theta}(b)\,dbf(\theta)\,d\theta +K(p)\tag{$P1$} \label{P1} \\
    &\text{s.t.} \quad \int_{\underline{\theta}}^{\bar{\theta}}\int_{0}^{p}x(\theta,b)h_{\theta}(b)\,dbf(\theta)\,d{\theta} = RS(p), \tag{$EQ$} \\
    &\;\;\;\;\;\;\;\; x(\theta,\cdot):B \rightarrow [0,1] \; \text{is non-decreasing for all } \theta, \tag{$MON$} \label{MON}
\end{align*}
where:
\[
K(p)=- U_{NP}(\underline{\theta})- \int_{0}^{p} [1 - g(H_{\underline{\theta}}(b))] \, db.
\]

In the formulation of~\ref{P1}, we have already used the fact that, in an optimal menu, the insurer will always set the utility of the lowest risk type equal to her outside option. We refer to~\ref{P1} as a relaxed version of the first-stage insurer's problem because it incorporates all necessary conditions that an implementable menu must satisfy. However, a solution to~\ref{P1} is not necessarily implementable. Specifically, it is possible that $x$ does not satisfy~\ref{IC1}.

The approach, therefore, is to solve~\ref{P1} and identify conditions under which the solution is implementable. For instance, by Lemma~\ref{lemma:IC1-menu}, Condition~\ref{condition:IC1-sufficiency}, we know that if the solution satisfies the property that $x(\cdot, b)$ is non-decreasing for all $b$, then $x$ is implementable.

The function $\Pi:\left[p^N,p^F\right] \rightarrow \mathbb{R}$ captures the profits obtained by the insurer from the optimal menu that sustains $p$ in equilibrium.

\subsection{Representing the Insurer's Problem in the Quantity Space}\label{sec:quantity-space}
In this section we reformulate the insurer's problem in the quantity (or quantile) space, a common technique in mechanism design (see, for instance, \cite{bulow1989simple}). For this purpose, let $H_{\theta}(\cdot \lvert p): [0, p] \rightarrow [0, 1]$ denote the distribution of valuations for risk type $\theta$, conditional on being strictly greater than $0$ and less than or equal to $p$. Specifically, 
\[
H_{\theta}(b \mid p) = \frac{H_{\theta}(b) - H_{\theta}(0)}{H_{\theta}(p) - H_{\theta}(0)},
\]

By our assumption on $H_{\theta}$, it follows that $H_{\theta}(\cdot \mid p)$ is invertible. Any non-decreasing, left-continuous function $x(\theta, \cdot): [0,p] \rightarrow [0, 1]$ can then be represented as:
\[
x(\theta, b) = \int_0^1 \mathbf{1}_{[b > b_{\theta}(q)]} \, dG_{\theta}(q),
\]
where $b_{\theta}(q) = H_{\theta}^{-1}(1 - q \mid p)$ and $G_{\theta}$ is a distribution over $[0, 1]$.

Economically, this representation allows us to express a mechanism in the quantity space. In the absence of insurance, only the agents with a valuation exceeding $p$ seek the service. As a result, the total demand for services, conditional on risk type $\theta$, is given by $1 - H_{\theta}(p)$. To induce additional demand of $[H_{\theta}(p) - H_{\theta}(0)]q$ units, the insurer can offer a \textit{simple deductible contract} of the form  $\{(1, b_{\theta}(q))\}$. Under this contract, in addition to agents with valuations greater than $p$, those with valuations in the range $\left(b_{\theta}(q), p\right]$ will also seek the service. For risk type $\theta$, the mass of agents within the interval $\left(b_{\theta}(q), p\right]$ corresponds precisely to $[H_{\theta}(p) - H_{\theta}(0)]q$. In other words, the insurer can only influence the demand of agents with valuations within $(0,p]$. Among these agents, the insurer selects a fraction $q$ who will get the service. Moreover, by appropriately randomizing over $q$ (or equivalently, over simple deductible contracts), the insurer can replicate any feasible demand schedule $x(\theta, \cdot)$. Thus, without loss of generality, the optimization can be performed over $G_{\theta}$ instead of directly over $x(\theta, \cdot)$.

Then, the insurer's problem becomes:
\begin{align*}
    &\max_{G} \quad \int_{\underline{\theta}}^{\bar{\theta}}\int_0^1 \Phi(\theta, q) \, dG_{\theta}(q) f(\theta) \, d\theta,  \tag{$P1'$} \label{P1'}\\
    &\text{s.t.} \quad \int_{\underline{\theta}}^{\bar{\theta}}\int_{0}^{1}[H_{\theta}(p) - H_{\theta}(0)]q\,dG_{\theta}(q)f(\theta)\,d{\theta} = RS(p), \tag{$EQ$}  
\end{align*}
where:
\[
\Phi(\theta, q) = [H_{\theta}(p)-H_{\theta}(0)]\int_0^q\phi(\theta,s)ds,
\]
and:
\[
\phi(\theta, q) = J(\theta,b_{\theta}(q)).
\]

Note that in formulation~\ref{P1'} we no longer include the constant $K(p)$. In the first stage, this is without loss of generality, as any function $G$ that solves~\ref{P1'} will also be a maximizer if the constant $K(p)$ is included. Therefore, to simplify notation, we omit $K(p)$ for the remainder of this section.

\subsection{Two-type Example}\label{sec:example}
In this section, we illustrate the main intuition behind Theorem~\ref{thm:optimal-menu} with an example. Specifically, we consider a setting with only two risk types, $\theta \in \{\theta_L, \theta_H\}$, where $\theta_H>\theta_L$. This simplified case already reveals the economic forces that drive the insurer to offer a \textit{limited coverage deductible contract}.

Using the representation of the insurer's problem derived in Section~\ref{sec:quantity-space}, we obtain the following expression for the insurer's problem in the quantity space:
\begin{align*}
    &\max_{G} \quad \int_0^1 \Phi(\theta_L,q) \, dG_{\theta_L}(q) \hat{f}(\theta_L)+\int_0^1 \Phi(\theta_H,q) \, dG_{\theta_H}(q) \hat{f}(\theta_H)\\
    &\text{s.t.} \quad \int_{0}^{1}q\,dG_{\theta_L}(q)\hat{f}(\theta_L)+\int_{0}^{1}q\,dG_{\theta_H}(q)\hat{f}(\theta_H) = RS(p),
\end{align*}

where 
\[
\Phi(\theta,q)=\int_0^q\phi(\theta,s)ds,
\]
\[
\phi(\theta_L,q)=b_{\theta_L}(q)-p+\frac{H_{\theta_L}(b_{\theta_L}(q))-g(H_{\theta_L}(b_{\theta_L}(q)))}{h_{\theta_L}(b_{\theta_L}(q))}+\frac{1-f(\theta_L)}{f(\theta_L)h_{\theta_L}(b_{\theta_L}(q))}g(H_{\theta_H}(b_{\theta_H}(q)))-g(H_{\theta_L}(b_{\theta_L}(q))),
\]
\[
    \phi(\theta_H,q)=b_{\theta_H}(q)-p+\frac{H_{\theta_H}(b_{\theta_H}(q))-g(H_{\theta_H}(b_{\theta_H}(q)))}{h_{\theta_H}(b_{\theta_H}(q))},
\]
\[
RS(p)=S(p)-f(\theta_L)(1-H_{\theta_L}(p))-f(\theta_H)(1-H_{\theta_H}(p)),
\]
and 
\[
\hat{f}(\theta)=\left[H_{\theta}(p)-H_{\theta}(0)\right]f(\theta).
\]

Observe that $\phi(\theta, q)$ accounts for the period 1 and period 2 virtual values adjusted to a discrete risk type space. Specifically, we account for the discrete version of information rents. These rents no longer arise from the envelope formula but instead follow from the standard result that, with two types, the~\ref{IC1} constraint for the high type and the~\ref{IR} constraint for the low type always bind.

We now analyze the insurer's problem. For any distribution $G_{\theta}$ that results in an average $q$ given by $q_{\theta}$, there exists a distribution $G'_{\theta}$ that preserves the same average $q_{\theta}$ while attaining the concave envelope of $\Phi(\theta, \cdot)$ evaluated at $q_{\theta}$. In other words, let $\overline{\Phi}(\theta, \cdot)$ denote the concave envelope of $\Phi(\theta, \cdot)$. Then, we have:

\[
\int_0^1q\,dG'_{\theta}(q)=q_{\theta}=\int_0^1q\,dG_{\theta}(q), 
\]
and
\[
\int_0^1\Phi(\theta,q)\,dG'_{\theta}(q)=\overline{\Phi}(\theta,q_{\theta})\geq\int_0^1\Phi(\theta,q)\,dG_{\theta}(q).
\]

$G'$ weakly increases the insurer's payoff while still satisfying the~\ref{EQ} condition. Thus, we can also express the insurer's problem as:
\begin{align*}
    &\max_{q} \quad \overline{\Phi}(\theta_L,q_{\theta_L})\hat{f}(\theta_L) + \overline{\Phi}(\theta_H,q_{\theta_H})\hat{f}(\theta_H) \\
    &\text{s.t.} \quad q_{\theta_L}\hat{f}(\theta_L)+q_{\theta_H}\hat{f}(\theta_H) = RS(p),
\end{align*}
In the above formulation, for each risk type, the insurer selects the fraction of agents with a valuation in $(0,p]$ who seek the service, denoted by $q_{\theta}$. This formulation already demonstrates why a solution exists in which the contract offered to each risk type is either a simple deductible contract or a limited coverage deductible contract. Indeed, any point in $\overline{\Phi}(\theta, \cdot)$ can be attained by a distribution $G_{\theta}$ with at most two elements in its support. This, in turn, results in an allocation function $x(\theta, \cdot)$ with at most two jumps, as illustrated in Figure~\ref{fig:allocation}. 

Moreover, if $\Phi(\theta,\cdot)$ is concave, or equivalently, $\frac{\partial \phi(\theta,q)}{\partial q} \leq 0$ for all $q$, then the insurer can always achieve $\overline{\Phi}(\theta, q_{\theta})$ with a distribution that puts probability one on $q_{\theta}$. That is, the insurer can always offer risk type $\theta$ the simple deductible contract $\{(1,b_{\theta}(q_{\theta}))\}$. This observation is formalized in Lemma~\ref{lemma:simple-contract}.

To illustrate the economic forces that lead the insurer to include a limited coverage deductible contract, we assume $\theta_L = 0.25$, $\theta_H = 0.35$, $f(\theta_H) = \frac{1}{5}$, and $H_{\theta}(b) = 1 - \theta + \theta Q(b)$, where $Q(b) = b^2 \cdot \mathbf{1}_{\left[0 \leq b \leq 1\right]}$. In this case, we have $q\Phi(\theta,1) > \Phi(\theta,q)$ for all $q \in (0,1)$. This implies that the concave envelope of $\Phi(\theta, \cdot)$ is the line segment connecting $\Phi(\theta,0)$ and $\Phi(\theta,1) $, i.e., $\overline{\Phi}(\theta,q) = q\Phi(\theta,1)$, as shown in Figure~\ref{fig:concaveenvelope}.
\begin{figure}[H]
    \centering
    \includegraphics[width=0.5\linewidth]{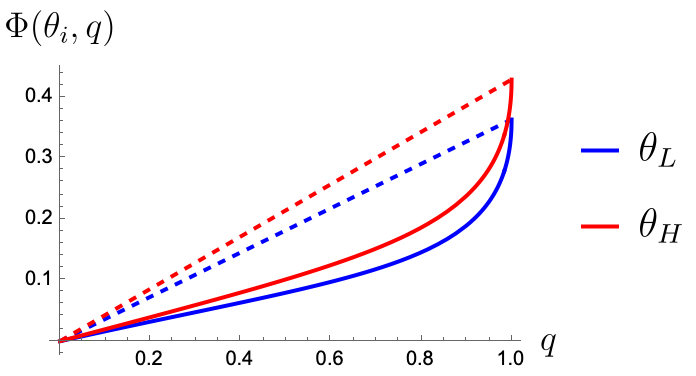}
    \caption{The continuous lines illustrate $\Phi(\theta_L,q)$ (in blue) and $\Phi(\theta_H,q)$ (in red). The dashed lines are the respective concave envelopes. The figure is drawn assuming that $p=0.8$.}
    \label{fig:concaveenvelope}
\end{figure}

Under these parametric assumptions, the insurer solves:  
\begin{align*}
    &\max_{q} \quad q_{\theta_L}\Phi(\theta_L,1)\hat{f}(\theta_L) + q_{\theta_H}\Phi(\theta_H,1)\hat{f}(\theta_H) \\
    &\text{s.t.} \quad q_{\theta_L}\hat{f}(\theta_L) + q_{\theta_H}\hat{f}(\theta_H) = RS(p).
\end{align*}

The solution thus depends solely on the relationship between $\Phi(\theta_L,1)$ and $\Phi(\theta_H,1)$. In particular, we can verify that $\Phi(\theta_H,1) > \Phi(\theta_L,1)$. Therefore, the optimal allocation is:  
\[
    q^*_{\theta_H}=\min\left\{1,\frac{RS(p)}{\hat{f}(\theta_H)}\right\},\quad 
    q^*_{\theta_L}=\frac{RS(p)}{\hat{f}(\theta_L)}-\min\left\{1,\frac{RS(p)}{\hat{f}(\theta_H)}\right\} \cdot \frac{\hat{f}(\theta_H)}{\hat{f}(\theta_L)}.
\]
In other words, the insurer favors the riskier type, allocating as many service units as possible to $\theta_H$ first. If the maximum insurance-driven demand, $\hat{f}(\theta_H)$, is sufficient to meet the residual supply, then $q^*_{\theta_H} = RS(p)/\hat{f}(\theta_H)$ and $q^*_{\theta_L} = 0$. Otherwise, we set $q^*_{\theta_H} = 1$ and further increase demand by including some low-risk types until the market clears.  

Furthermore, if $q^*_{\theta} \in (0,1)$, the contract for risk type $\theta$ is a limited coverage deductible contract. To achieve such $q^*_{\theta}$, the insurer mixes between the two extremes of $q = 0$ and $q = 1$, corresponding to an OPC of $p$ and $0$, respectively. For instance, if $p = 0.8$, we obtain $q^*_{\theta_L} = 0.67$ and $q^*_{\theta_H} = 1$. In this case, the high-risk type receives a simple deductible contract with a OPC equal to zero, while the low-risk type faces a limited coverage deductible contract: she can access up to 67\% of the service with a zero OPC or receive full treatment for higher OPC equal to 0.26.

Additionally, since the insurer always favors the riskier type, the optimal menu satisfies~\ref{IC1}. Indeed, in this case, it is easy to verify that whenever the high-risk type receives a limited coverage deductible contract, the low-risk type receives no insurance. Moreover, whenever the low-risk type receives some form of insurance, the high-risk type gets full insurance---access to the full unit of the service with a zero OPC. Thus, the optimal allocation of hospital services is increasing in the risk type. That is, $x^*(\theta_H,b) \geq x^*(\theta_L,b)$ for all $b$. Therefore, by Lemma~\ref{lemma:IC1-menu}, Condition~\ref{condition:IC1-sufficiency}, we conclude that $x^*$ satisfies~\ref{IC1}. 

The economic forces that lead the insurer to favor the riskier type follow from the fact that the profits of providing any units of hospital services are higher for the high-risk type than for the low-risk type, as illustrated in Figure~\ref{fig:concaveenvelope}. The latter occurs because $\phi(\theta_H,q) > \phi(\theta_L,q)$, i.e., the sum of the period 1 and period 2 virtual values is higher for the high-risk type than for the low-risk type for any $q$. This observation is formalized in Lemma~\ref{lemma:implementable-menu}.

\subsection{Optimal Transport and Duality}\label{sec:optimal-transport}  
We now return to the analysis of the insurer's problem in the quantity space, focusing on the infinite-dimensional risk-type space. Consider a collection of distributions over $[0,1]$, denoted by $\{G_{\theta}\}_{\theta \in \Theta}$, which the insurer selects in problem~\ref{P1'}. Additionally, recall that $F$ denotes the prior distribution over risk types. Together, these define a probability measure $\pi$ on $\Theta \times [0,1]$, where the marginal of $\pi$ on $\Theta$ is $F$. We can then reformulate problem~\ref{P1'} as follows:  

\begin{align*}
    &\max_{\pi} \quad \int_{\underline{\theta}}^{\bar{\theta}}\int_0^{1}\Phi(\theta,q)\,d\pi(\theta,q) \tag{$OT$} \label{OT} \\
    &\text{s.t.} \quad \int_0^{\theta}\int_0^1\,d\pi(s,q)=F(\theta), \quad \forall \theta \in \Theta, \tag{$BP$} \label{BP} \\
    &\;\;\;\;\;\;\;\int_{\underline{\theta}}^{\bar{\theta}}\int_0^{1}[H_{\theta}(p) - H_{\theta}(0)]q\,d\pi(\theta,q)=RS(p). \tag{$EQ$} 
\end{align*}  

In the \textit{primal problem}~\ref{OT}, the insurer is effectively matching risk types to service demands. If risk type $\theta$ is matched with $q$, then among the agents that the insurer can influence to seek the service---those with a valuation in the range $(0,p]$---a fraction $q$ will get the service. However, this allocation is subject to two constraints. The first constraint,~\ref{BP}, ensures that the insurer does not assign more individuals of risk type $\theta$ than the total available, as determined by the prior distribution $F$. In other words, the marginal of $\pi$ on $\Theta$ must coincide with $F$, mirroring a Bayes-plausibility constraint in Bayesian persuasion models. The second constraint,~\ref{EQ}, is the equilibrium condition described earlier.  

The similarity between~\ref{OT} and an optimal transport problem is advantageous because it allows us to leverage known duality results to analyze the insurer's optimization problem more effectively. To this end, we introduce the corresponding \textit{dual problem}:  

\begin{align*}
    &\min_{v,\lambda} \quad \int_{\underline{\theta}}^{\bar{\theta}}v(\theta)f(\theta)d\theta  \tag{$D$} \label{D} \\
    &\text{s.t.} \quad v(\theta) \geq \max_{q \in [0,1]} \;\Phi(\theta,q)+\lambda\left([H_{\theta}(p)-H_{\theta}(0)]q-RS(p)\right), \quad \text{for $f$-almost all $\theta$}. \tag{$SP$} \label{SP}
\end{align*}  

In this dual formulation, $v(\theta)$ represents the shadow price of matching risk type $\theta$. The multiplier $\lambda$ corresponds to the value of relaxing the equilibrium constraint~\ref{EQ}. Furthermore, constraint~\ref{SP} ensures that $v(\theta)$ is at least as large as the insurer’s valuation of matching risk type $\theta$ with any $q$, where this valuation consists of the insurer’s profit, $\Phi(\theta,q)$, plus the term multiplied by $\lambda$, which captures the excess demand or supply generated by risk type $\theta$ when assigned to $q$.  

Then, we have the following duality result: 

\begin{theorem}[Monge-Kantorovich Duality]\label{thm:zero-duality}
    A solution tho both the primal problem~\ref{OT} and the dual problem~\ref{D} exists. Furthermore, there is no duality gap: the optimal value of the primal problem~\ref{OT} equals the optimal value of the dual problem~\ref{D}.
\end{theorem}
The formal proof of Theorem~\ref{thm:zero-duality} can be found in the technical  Appendix~\ref{sec:proof-technical}.

The absence of a duality gap is a direct consequence of the Monge-Kantorovich duality in optimal transport problems (see, for example, \citealp{villani2009optimal}). The existence of a solution to the primal problem~\ref{OT} follows from our technical assumptions that guarantee the continuity of the function $\Phi$ and the fact that the maximization is performed over a compact space. Finally, leveraging the absence of a duality gap and the existence of a solution to the primal problem, we explicitly construct a solution to the dual problem~\ref{D}.

\paragraph{Implications of Zero Duality Gap.} 
An immediate consequence of the absence of a duality gap is that, given a solution to the dual problem~\ref{D}, we obtain necessary and sufficient conditions that any joint distribution $\pi$ must satisfy to be a solution of the primal problem~\ref{OT}. To establish these conditions, consider any $\pi$ that satisfies~\ref{BP}. Let $\pi(\cdot \mid \theta)$ denote the distribution over $[0,1]$ conditional on risk type $\theta$, defined as:  
\[
\pi(q \mid \theta) = \frac{\pi(q, \theta)}{f(\theta)}.
\]  

With this, we obtain the following complementary slackness result:
\begin{corollary}[Complementary Slackness]\label{coro:complementary-slackness}
    Suppose that $(v^*,\lambda^*)$ is a solution to the dual problem~\ref{D}. Then, any $\pi$ that satisfies conditions~\ref{BP} and~\ref{EQ} is a solution to the primal problem~\ref{OT} if and only if $f$-almost every $\theta$ has $\pi(\cdot \mid \theta)$ supported on:  
    \[
   \arg\max_{q \in [0,1]} \quad \Phi(\theta,q)+\lambda^*\left([H_{\theta}(p)-H_{\theta}(0)]q-RS(p)\right) \label{supp} \tag{$Supp$}.
    \]
\end{corollary}
The formal proof of Corollary~\ref{coro:complementary-slackness} is in the technical Appendix~\ref{sec:proof-technical}.

The proof of Corollary~\ref{coro:complementary-slackness} follows directly from the absence of a duality gap and the fact that if $(v^*,\lambda^*)$ is a solution to~\ref{D}, then we have  
\[
v^*(\theta) = \max_{q \in [0,1]} \quad \Phi(\theta,q) + \lambda^* \left( [H_{\theta}(p) - H_{\theta}(0)]q - RS(p) \right),
\]  
where the maximum is well-defined due to the continuity of $\Phi(\theta, \cdot)$ and the compactness of the set $[0,1]$.  

Corollary~\ref{coro:complementary-slackness} is particularly useful for proving Theorem~\ref{thm:optimal-menu}. To this end, we introduce the following lemma, from which Theorem~\ref{thm:optimal-menu} follows immediately:

\begin{lemma}[Solution to \ref{OT}]\label{lemma:optimal-menu}
    There exists a solution $\pi^*$ to the primal problem~\ref{OT} such that, for all $\theta$ and $q$, $\pi^*(q \mid \theta) \in \{0,\alpha,1\}$ for some $\alpha \in (0,1)$. Moreover, this $\pi^*$ implies that each risk type $\theta$ is assigned either a \textbf{simple deductible contract} $\{(1, D(\theta))\}$ with $D(\theta) \in [0,p]$, or a \textbf{limited-coverage deductible contract} $\{(\alpha,D(\theta)),(1, M(\theta))\}$ where $0 \leq D(\theta) < M(\theta) \leq p$ and $\alpha \in (0,1)$.
\end{lemma}
The proof of Lemma~\ref{lemma:optimal-menu} in Appendix~\ref{sec:proof-optimal-transport}.

Theorem~\ref{thm:optimal-menu} follows directly from Lemma~\ref{lemma:optimal-menu}. Once we know that $\pi^*$ is a solution to the primal problem~\ref{OT}, it immediately implies that the following function $x^*$ solves the insurer's problem~\ref{P1}:
\[
x^*(\theta,b) = \int_0^1 \mathbf{1}_{\left[b > b_{\theta}(q)\right]} \, d\pi^*(q \mid \theta).
\]

Since $\pi^*(\cdot \mid \theta)$ consists of at most two jumps, and the size of these jumps is independent of the risk type, the allocation function $x^*(\theta,\cdot)$ can take one of two forms for any risk type $\theta$. It may be a step function that jumps from $0$ to $1$ at some value of $b$ given by $D(\theta)$, which corresponds to the \textit{simple deductible contract}. Under this contract, all agents with valuation above $D(\theta)$ pay the OPC. Alternatively, $x^*(\theta,\cdot)$ may exhibit two jumps, as shown in Figure~\ref{fig:allocation}, in which case the contract for type $\theta$ is a \textit{limited-coverage deductible contract}. Specifically:
\begin{itemize}
    \item $x^*(\theta,b)=0$ for valuations $b \leq D(\theta)$---these agents forgo treatment and stay at home.
    \item $x^*(\theta,b)=\alpha$ for valuations $b \in \left(D(\theta), \frac{M(\theta) - \alpha D(\theta)}{1 - \alpha}\right]$---these agents pay the OPC and receive partial coverage.
    \item $x^*(\theta,b)=1$ for valuations $b > \frac{M(\theta) - \alpha D(\theta)}{1 - \alpha}$---these agents pay the higher OPC and receive full treatment.
\end{itemize}

\begin{figure}[H]
    \centering
    \includegraphics[width=0.5\linewidth]{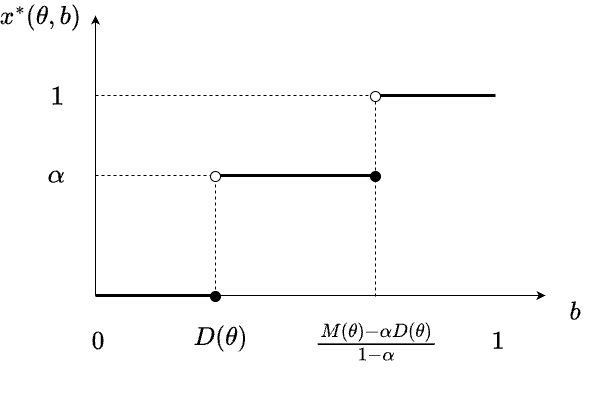}
    \caption{Optimal allocation $x^*$ corresponding to the menu described in Theorem \ref{thm:optimal-menu}.}
    \label{fig:allocation}
\end{figure}

The intuition behind the proof of Lemma~\ref{lemma:optimal-menu} is as follows: By Theorem~\ref{thm:zero-duality} and Corollary~\ref{coro:complementary-slackness}, we know that there exists a $\pi'$ that solves~\ref{OT}, a $(v^*, \lambda^*)$ that solves~\ref{D}, and that $\pi'(\cdot \mid \theta)$ is supported on~\ref{supp}.

By Carathéodory's theorem, we can construct a new distribution $\pi^*$ such that $\supp \pi^*(\cdot \mid \theta)$ contains at most two elements, is a subset of $\supp \pi'(\cdot \mid \theta)$, and results in the same average quantity of services for type $\theta$. This implies that $\pi^*$ still satisfies the~\ref{EQ} constraint. Furthermore, since $\supp \pi^*(\cdot \mid \theta) \subset \supp \pi'(\cdot \mid \theta)$, we can again apply Corollary~\ref{coro:complementary-slackness} to conclude that $\pi'$ is optimal.

Next, we examine the primitive conditions under which the insurer optimally offers a simple deductible contract:

\begin{lemma}[Simple Deductible Contracts]\label{lemma:simple-contract}
    Fix any risk type $\theta$ and suppose that $\frac{\partial \phi(\theta,q)}{\partial q} \leq 0$ for all $q$. Then there exists a solution $\pi^*$ to the primal problem~\ref{OT} such that $\pi^*(q \mid \theta) \in \{0,1\}$ for all $q$. Under this solution, all agents are assigned a \textbf{simple deductible contract}.
\end{lemma}
The proof of Lemma~\ref{lemma:simple-contract} is provided in Appendix~\ref{sec:proof-optimal-transport}.  

Lemma~\ref{lemma:simple-contract} once again relies on our complementary slackness result. To illustrate, suppose that $\frac{\partial \phi(\theta,q)}{\partial q} > 0$ holds with strict inequality. Then, the objective in~\ref{supp} is strictly concave in $q$, implying that the set~\ref{supp} contains only one element. Moreover, by our complementary slackness result, we know that $\pi^*(\cdot \mid \theta)$ is supported on~\ref{supp}, leading to the desired conclusion.  

Finally, we check when a solution to~\ref{OT} is fully implementable. That is, we have to verify that the solution $x^*$ to~\ref{P1} obtained from the solution $\pi^*$ described in Lemma~\ref{lemma:optimal-menu} satisfies~\ref{IC1}. 

\begin{lemma}\label{lemma:implementable-menu}
    Suppose that $\frac{\partial \phi(\theta,q)}{\partial \theta} > 0$ for all $q$, and let $\pi^*$ be a solution to the primal problem. Then, for any $\theta > \theta'$,  
    \[
        \min \supp \pi^*(\cdot \mid \theta) \;\geq\; \max \supp \pi^*(\cdot \mid \theta').
    \]
   Moreover, under $\pi^*$ higher risk types receive greater coverage at every loss-related valuation: $\pi^*$ induces an allocation function $x^*$ such that $x^*(\cdot,b)$ is non-decreasing for all $b$, and hence $x^*$ satisfies~\ref{IC1}.
\end{lemma}
The proof of Lemma~\ref{lemma:implementable-menu} is provided in Appendix~\ref{sec:proof-optimal-transport}.  
Lemma~\ref{lemma:implementable-menu} implies that the allocation $x^*$ obtained through $\pi^*$, which solves the insurer's problem~\ref{P1}, satisfies the monotonicity condition: $x^*(\cdot, b)$ is non-decreasing for all $b$. Therefore, by Lemma~\ref{lemma:IC1-menu} and Condition~\ref{condition:IC1-sufficiency}, $x^*$ satisfies~\ref{IC1}.  

The intuition behind Lemma~\ref{lemma:implementable-menu} is straightforward. By Corollary~\ref{coro:complementary-slackness}, we know that $\pi^*(\cdot \mid \theta)$ is supported on~\ref{supp}. Then, applying Topkis's monotonicity theorem, we conclude that the objective in~\ref{supp} strictly increases in $\theta$ in the strong set order, yielding the desired result.

\section{Second-Stage: Price of Services}\label{sec:second-stage}
Throughout section~\ref{sec:first-stage}, we took the price of services, $p$, as given and characterized the shape of the optimal menu that sustains $p$ in equilibrium. This allowed us to associate each achievable $p$ with the corresponding profits obtained by the insurer, denoted by $\Pi(p)$. 

We now proceed to the second-stage analysis, which consists of solving for the price that maximizes $\Pi$. Formally, the second-stage problem is:

\begin{equation}
 \max_{p \in \left[p^N, p^F\right]} \, \Pi(p). \tag{$P2$} \label{P2}
\end{equation}

To analyze problem~\ref{P2}, we introduce the following result, which characterizes the derivative of $\Pi$:

\begin{theorem}\label{thm:envelope}
   Let $x^*$ be a solution to the first-stage problem~\ref{P1} and $\lambda^*$ be a solution to the dual problem~\ref{D}. Then, the function $\Pi: \left[p^N, p^F\right] \rightarrow \mathbb{R}_+$ is differentiable in $p$, and its derivative is given by:
    \begin{align*} \Pi'(p)=&\int_{\underline{\theta}}^{\overline{\theta}}x^*(\theta,p)J(\theta,p)h_{\theta}(p)f(\theta)\,d\theta- \int_{\underline{\theta}}^{\bar{\theta}}\int_0^px^*(\theta,b)h_{\theta}(b)\,dbf(\theta)\,d\theta\\
    &-\lambda^* \left( S'(p) + \int_{\underline{\theta}}^{\overline{\theta}}(1-x^*(\theta,p))h_\theta(p)f(\theta)d\theta\right)-[g(H_{\underline{\theta}}(p+\beta))-g(H_{\underline{\theta}}(p))] %K'(p),
    \end{align*}
    %where, if the agent’s outside option $U_{NP}$ corresponds to the endogenous case, then $K'(p) = 0$, while if $U_{NP}$ corresponds to the exogenous case, then $K'(p) = -[1 - g(H_{\underline{\theta}}(p))]$.
\end{theorem}
The formal proof of Theorem~\ref{thm:envelope} is in Appendix~\ref{sec:proof-second-stage}. We invoke Theorem 4 of \cite{milgrom2002envelope}, which provides an envelope formula for saddle points. In particular, our duality result (Theorem \ref{thm:zero-duality}) implies that a solution $\pi^*$ to the primal problem~\ref{OT}, together with the multiplier $\lambda^*$ that solves the dual problem~\ref{D}, form a saddle point. Under our technical assumptions, this allows us to apply Theorem 4 of \citet{milgrom2002envelope} to obtain a formula for the derivative of the insurer’s profit in terms of $\pi^*$ and $\lambda^*$. Finally, using the relationship between $\pi^*$ and $x^*$, we derive the desired result. We postpone the economic interpretation of $\Pi'$ in the case where the insurer offers a simple deductible contract to all risk types, as this setting allows for a clearer explanation.

Since $\Pi$ is differentiable, it is also continuous. Moreover, because problem~\ref{P2} involves maximizing $\Pi$ over a compact domain, the \textit{Extreme Value Theorem} ensures that a maximizer exists. Whenever the maximizer lies in the interior of the domain, it must satisfy the first-order condition:

\begin{corollary}\label{coro:maximum-price}
A solution $p^*$ to second-stage problem~\ref{P2} exists. Moreover, if $p^* \in \left(p^N, p^F\right)$, then it satisfies $\Pi'(p^*)=0$.
\end{corollary}

Before turning to the case where the insurer offers only simple deductible contracts, we show that as the premium $\beta$ decreases, the equilibrium service price is weakly higher.

%Before turning to the case where the insurer offers only simple deductible contracts, we present a result showing that when the insurer can influence the agent’s outside option---the \textit{endogenous case}---he can always charge a weakly higher price than an insurer who cannot---the \textit{exogenous case}.

\begin{corollary}\label{coro:endo-exo}
    Suppose that $\beta > \beta'$. Then, there does not exist a price $p$ that solves the second-stage problem~\ref{P2} under $\beta$ and is strictly greater than all prices that solve the second-stage problem~\ref{P2} under $\beta'$.
    %There does not exist a price $p$ that solves the second-stage problem~\ref{P2} in the exogenous case and is strictly greater than all prices that solve the second-stage problem~\ref{P2} in the endogenous case.
\end{corollary}
The formal proof of Corollary~\ref{coro:endo-exo} is provided in Appendix~\ref{sec:proof-second-stage}. The result is stated to accommodate the possibility of multiple maximizers. When the maximizer of~\ref{P2} is unique under both $\beta$ and $\beta'$, the corollary simply states that the price solving~\ref{P2} under $\beta$ is lower than the one under $\beta'$.

%The formal proof of Corollary~\ref{coro:endo-exo} is provided in Appendix~\ref{sec:proof-second-stage}. The result is stated to accommodate the possibility of multiple maximizers. However, when the maximizer of~\ref{P2} is unique in both the endogenous and exogenous cases, the corollary simply states that the insurer will choose a weakly higher price when he can influence the agent’s outside option, compared to when he cannot.

The mathematical intuition is straightforward: the derivative of the insurer’s profit under $\beta$ is always weakly lower than under $\beta'$. The economic intuition is similarly simple: when the premium $\beta$ is very high, an increase in $p$ has little effect on the agent’s outside option. Consider the extreme case in which $\beta$ is high enough that $\beta + p \ge \bar{b}$; in this case the agent will never seek the service in the absence of insurance, so further increases in $p$ do not change this fact, and her outside option---and thus her willingness to pay for insurance---remains the same. Then, any incentive to increase $p$ in order to worsen the agent’s outside option, raise her willingness to pay for insurance, and extract more surplus is shut down.

%The mathematical intuition is straightforward: the derivative of the insurer’s profit in the exogenous case is always weakly lower than in the endogenous case. The economic intuition is similarly simple: when the insurer can affect the agent’s outside option, he has an additional incentive to raise the price. By increasing the price, the insurer shifts the agent’s outside distribution over losses in the first-order stochastic dominance sense, as losses between the previous and new prices now occur with positive probability. This increases the agent’s ex-ante willingness to pay for insurance and enables the insurer to charge higher premiums.

\paragraph{Simple Deductible Contracts.} Throughout this section, we assume that for all $p$ and $\theta$, the function $\phi(\theta,\cdot)$ is strictly decreasing. Recall that $\phi(\theta,q) = J(\theta, b_{\theta}(q))$, where $b_{\theta}(q) = H_{\theta}^{-1}(1 - q \mid p)$, and $H_{\theta}(\cdot \mid p)$ denotes the distribution of valuations for risk type $\theta$, conditional on being in $(0, p]$. It is easy to verify that $b_{\theta}(\cdot)$ is strictly decreasing. Therefore, the assumption is equivalent to requiring that the virtual value function $J(\theta, \cdot)$ is strictly increasing for all $\theta$.

Then, Lemma~\ref{lemma:simple-contract} implies that for all $p$, it is optimal for the insurer to offer each risk type a simple deductible contract, denoted by $\{(1, D(\theta; p))\}$, where $D(\theta; p) \in [0, p]$. With a slight abuse of notation, we now make explicit the dependence of the contract on the price $p$. Similarly, we use the notation $\lambda^*(p)$ to reflect the dependence of the multiplier solving the dual problem~\ref{D} on $p$. Applying Theorem~\ref{thm:envelope}, we obtain the following result:

\begin{corollary}\label{coro:envelope-simple}
Suppose that $J(\theta,\cdot)$ is strictly increasing for all $\theta$. Then, the function $\Pi: \left[p^N, p^F\right] \rightarrow \mathbb{R}+$ is differentiable in $p$, and its derivative is given by:
\begin{align*}
    \Pi'(p) = &\int_{\underline{\theta}}^{\bar{\theta}} \mathbf{1}_{\left[D(\theta; p) < p\right]} \left[\underbrace{1 - g(H_\theta(p)) + \frac{1 - F(\theta)}{f(\theta)} g'(H_{\theta}(p)) \frac{\partial H_{\theta}(p)}{\partial \theta}}_{\text{Virtual value}}-\underbrace{\left(1 - H_\theta(D(\theta; p)) \right)}_{\text{Payouts}}\right] f(\theta) d\theta \\
    & - \underbrace{\lambda^*(p) \left[ S'(p) + \int_{\underline{\theta}}^{\overline{\theta}} \mathbf{1}_{\left[D(\theta; p) = p\right]} h_\theta(p) f(\theta) d\theta \right]}_{\text{Equilibrium effects}}-\underbrace{[g(H_{\underline{\theta}}(p+\beta))-g(H_{\underline{\theta}}(p))]}_{\text{Outside Option}}.
\end{align*}
\end{corollary}
We omit the proof of Corollary~\ref{coro:envelope-simple} as it is an immediate consequence of Theorem~\ref{thm:envelope} using the fact that $x^*(\theta,b)$ is equal to zero for values of $b$ lower or equal than $D(\theta;p)$, it is equal to one for values of $b$ strictly higher than $D(\theta;p)$ and $D(\theta;p) \in [0,p]$.

Corollary~\ref{coro:envelope-simple} allows us to provide a clear interpretation of the economic forces that determine the value of inducing a marginal increase in the service price. For each $\theta$, the term $1 - g(H_\theta(p))$ captures the ex-ante increase in the value of insurance. Agents who receive coverage—i.e., those for whom $D(\theta; p) < p$—now face potentially higher losses, which are protected by insurance. Since $1 - g(H_\theta(p))$ reflects the cumulative weight assigned to losses above $p$, this term measures the agent’s marginal increase in willingness to pay for insurance.

If $\theta$ were observable, the insurer could fully appropriate this increase in willingness to pay by raising premiums accordingly, thereby increasing profits. However, when $\theta$ is private information, the insurer can only appropriate the agent’s first-period virtual value---that is, her willingness to pay net of information rents. Recall that the term $\frac{1 - F(\theta)}{f(\theta)} g'(H_\theta(b)) \frac{\partial H_\theta(b)}{\partial \theta}$ is negative, reflecting these rents.

The second term, $1 - H_\theta(D(\theta; p))$, captures the insurer’s cost from increasing the price. If a risk type $\theta$ receives a simple deductible contract $\{(1, D(\theta; p))\}$ and is covered, then all agents with valuations in the range $(D(\theta; p), \bar{b})$ will choose to seek the service and pay the OPC. This implies a total demand equal to $1 - H_\theta(D(\theta; p))$. Therefore, a unit increase in price leads to a proportional increase in payouts.

In addition, we must account for equilibrium effects. The sign of this component is governed by the multiplier $\lambda^*(p)$, since the bracketed term is always positive and represents the increase in residual supply. Specifically, a higher price increases the supply of services, captured by $S'(p)$, and reduces demand among those not covered under the optimal contract (i.e., those for whom $D(\theta; p) = p$).

%Finally, the term $K'(p)$ reflects the agent’s outside option. Recall that in the endogenous case, $K'(p) = 0$. In this case, the insurer can fully appropriate the increase in the agent’s virtual value, as a higher price shifts the agent’s distribution over losses in the first-order stochastic dominance sense both with and without insurance. In contrast, in the exogenous case, $K'(p) \leq 0$. When strictly negative, this term indicates that the insurer cannot capture the full increase in virtual value, since the distribution over losses without insurance remains unchanged.

Whenever $J(\theta,\cdot)$ is strictly increasing, we can derive a comparative statics result describing how the multiplier $\lambda^*(\cdot)$ and the deductible $D(\theta;\cdot)$ change with the price.

\begin{proposition}\label{prop:comparative-statics}
Suppose that $J(\theta,\cdot)$ is strictly increasing for all $\theta$. Then, $\lambda^*(\cdot)$ is  continuous and strictly increasing, and for all $\theta$, $D(\theta;\cdot)$ is continuous and strictly decreasing whenever $D(\theta;p) \in (0,p)$ and is non-increasing when $D(\theta;p) = 0$. 
\end{proposition}
The proof of Proposition~\ref{prop:comparative-statics} is in Appendix~\ref{sec:proof-second-stage}. The intuition is straightforward: when the price increases, supply also increases. To clear the market, demand must rise as well. To incentivize higher demand, the multiplier must increase. In turn, agents require greater coverage, which the insurer provides by lowering the deductible.

Using our comparative statics result, we now address the question of whether the insurer chooses a ``high'' or ``low'' price in the second-stage problem~\ref{P2}. To do so, we must first define what constitutes a high or low price. For this purpose, we introduce the concept of a \textit{price-taking benchmark}, which is the price that clears the services market and solves the first-stage problem~\ref{P1}, ignoring the equilibrium constraint~\ref{EQ}. We denote this benchmark price by $p^U$.

Note that the value of relaxing the equilibrium constraint~\ref{EQ} at $p^U$ is zero, i.e., $\lambda^*(p^U) = 0$. Moreover, $p^U$ has a clear economic interpretation: it is the market price that would arise if the insurer took prices as given and did not internalize its impact on the market for services. The following lemma establishes the existence and uniqueness of the price-taking benchmark.

\begin{lemma}\label{lemma:price-taking}
    Suppose that $J(\theta,\cdot)$ is strictly increasing for all $\theta$. Then, there exists a unique price $p^U$ such that $\lambda^*\left(p^U\right)=0$.
\end{lemma}
The existence and uniqueness of $p^U$ follow essentially from the fact that $\lambda^*(\cdot)$ is strictly increasing and continuous. The proof of Lemma~\ref{lemma:price-taking} is provided in Appendix~\ref{sec:proof-second-stage}.

Before stating the result comparing the price chosen by the insurer in problem~\ref{P2} to $p^U$, let $J_1(\theta, b)$ denote the period-1 virtual value, defined as:
\[
J_1(\theta, b) = 1 - g(H_{\theta}(b)) + \frac{1 - F(\theta)}{f(\theta)} g'\left(H_{\theta}(b)\right) \frac{\partial H_{\theta}(b)}{\partial \theta}.
\]

\begin{proposition}\label{prop:price-taking}
Suppose that $\beta=0$ (uninsured agents face the same service price); $J(\theta,\cdot)$ is strictly increasing for all $\theta$; $J_1(\theta,\cdot)$ is non-increasing for all $\theta$; $H_{\theta}(\cdot)$ is concave for all $\theta$. Then it is sufficient to evaluate $\Pi'$ locally at the benchmark $p^U$ to determine whether the equilibrium service price lies above, below, or at this benchmark:
\begin{itemize}
    \item If the surplus extraction force is strictly greater than the cost effect force at the benchmark, the equilibrium service price is strictly higher than $p^U$.
    
    \item If the surplus extraction force is strictly smaller than the cost effect force at the benchmark, the equilibrium service price is strictly lower than $p^U$.

    \item If the surplus extraction force equals the cost effect force at the benchmark, the equilibrium service price is exactly $p^U$.
\end{itemize}
\end{proposition}
The formal proof of Proposition~\ref{prop:price-taking} is provided in Appendix~\ref{sec:proof-second-stage}, but we offer a brief sketch here. Suppose that $p < p^U$. The first possibility is that an agent receives coverage at price $p$. Then, by Proposition~\ref{prop:comparative-statics}, she also receives coverage under the price-taking benchmark $p^U$. Moreover, the same proposition implies that the deductible is lower at $p^U$, which in turn leads to higher payouts.

Additionally, since $J_1(\theta,\cdot)$ is assumed to be non-increasing, the marginal virtual value extracted at $p^U$ is weakly lower than that extracted at $p$. In other words, at the higher price $p^U$, payouts increase while the marginal surplus extracted weakly decreases.

The second possibility is that an agent does not receive coverage at price $p$, but does so at $p^U$. In this case, the concavity of $H_{\theta}(\cdot)$ ensures that the equilibrium effects at the lower price are strong enough to offset any benefit from expanding coverage at $p^U$. These effects are captured by the term $-\lambda^*(p)(S'(p) + h_{\theta}(p))$. By Proposition~\ref{prop:comparative-statics}, we know that $-\lambda^*(p) > -\lambda^*\left(p^U\right) = 0$.

This allows us to conclude that $\Pi'(p) > \Pi'\left(p^U\right)$. A symmetric argument shows that the inequality reverses when $p > p^U$.

Proposition~\ref{prop:price-taking} shows that, under the given assumptions, it suffices to compare the virtual value extracted from the agent with the payouts at $p^U$ to determine whether the insurer will set a price above or below the price-taking benchmark.

Unlike in standard monopoly pricing models, where the monopolist typically sets a price above the price-taking benchmark, in our setting the optimal price may lie above or below it. The key difference is that, in standard models, the monopolist benefits solely from a positive effect at the benchmark price: higher profits per unit sold. In contrast, the insurer in our model faces a potentially positive impact from increased virtual value extraction, but also a negative effect due to higher payouts. As a result, the net effect of a price increase is ambiguous and depends on the model's primitive assumptions.

The next result shows that, under additional structure on $H_{\theta}(\cdot)$ and the supply function $S$, the insurer's objective function $\Pi$ becomes concave, with strictly decreasing derivative $\Pi'$:

\begin{proposition}\label{prop:concave-profits}
Suppose that $\beta=0$ (uninsured agents face the same service price); $J(\theta,\cdot)$ is strictly increasing for all $\theta$; $J_1(\theta,\cdot)$ is non-increasing for all $\theta$; $H_{\theta}(\cdot)$ is affine for all $\theta$, and $S$ is also affine. Then, $\Pi'$ is strictly decreasing and $\Pi$ is concave.
\end{proposition}
The formal proof of Proposition~\ref{prop:concave-profits} is provided in Appendix~\ref{sec:proof-second-stage}.

The following example illustrates a case in which both $S$ and $H_{\theta}(\cdot)$ are affine and $\Pi'$ is strictly decreasing. Moreover, in this example, we have that the cost effect force dominates the surplus extraction force, implying that the insurer’s optimal price lies strictly below the price-taking benchmark.
 
\begin{example}\label{ex:concave-profits}
    Let $\theta \sim U[0,1]$, $H_{\theta}(b) = 1 - \theta + \theta Q(b)$, and $Q(b) = b\cdot\mathbf{1}_{\left[0 \leq b \leq 1\right]}$. Furthermore, for simplicity let us assume that the risk type $\theta$ of the agent is common knowledge. Then, $\Pi'(p)$ is as shown in Figure \ref{fig:example_mkteql}:

\begin{figure}[H]
    \centering
    \includegraphics[width=0.55\linewidth]{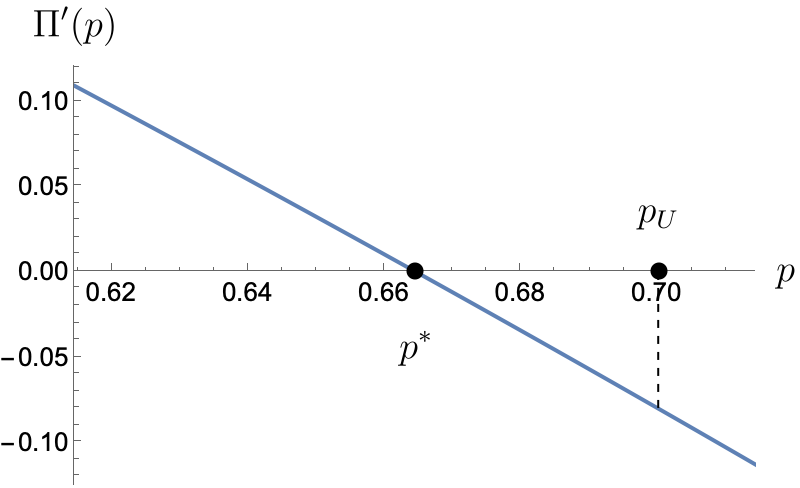}
    \caption{In this case, the optimal price is $p^*=0.66$ and the price-taking benchmark is $p^U=0.7$}
    \label{fig:example_mkteql}
\end{figure}
\end{example}

\section{Extensions}\label{sec:extensions}

\subsection{Social Planner}\label{sec:social-planner}
In this section, we show how our model can also be used to analyze the price impact of a public provider of insurance. We are particularly motivated by countries where the state is the main provider of insurance, while the majority of service providers are privately owned, or countries where insurance markets are highly regulated.\footnote{For instance, in Germany, the government mandates insurance coverage through a statutory health insurance system (SHI), while roughly 50\% of hospitals are privately owned. In Japan, there is universal public insurance through government-managed or employer-based plans, with state regulation, while the majority of hospitals are privately owned.}

To analyze this setting, we consider a social planner with an objective similar to that proposed by \cite{akbarpour2024redistributive}. Specifically, in the first-stage problem, given a price $p$, the planner designs an insurance menu that maximizes a weighted average of insurer profits and social surplus:

\[
\delta\underbrace{\left[\int_{\underline{\theta}}^{\bar{\theta}}\int_0^px(\theta,b)J(\theta,b)h_{\theta}(b)\,dbf(\theta)\,d\theta-U(\underline{\theta})-\int_0^p[1-g(H_{\underline{\theta}}(b))]\,db\right]}_{\text{Profits}}+\underbrace{\int_{\underline{\theta}}^{\bar{\theta}}\omega(\theta)U(\theta)f(\theta)\,d\theta}_{\text{Social Surplus}},
\]
where the planner assigns a weight $\delta$ to profits and a weight $\omega(\theta)$ to an agent with ex-ante risk type $\theta$. Profits are given by Lemma~\ref{lemma:objective}, and correspond to those generated by an implementable insurance menu---that is, a menu that satisfies conditions~\ref{IC2},~\ref{IC1}, and~\ref{IR}. Similarly, $U(\theta)$ denotes the ex-ante payoff---or certainty equivalent---that an agent of type $\theta$ obtains from an implementable menu (see Lemma~\ref{lemma:IC1-menu}, Condition~\ref{condition:IC1-necessity}).

Using the expression for $U(\theta)$ provided in Lemma~\ref{lemma:IC1-menu}, Condition~\ref{condition:IC1-necessity}, we can simplify the objective as follows:
\[
\int_{\underline{\theta}}^{\bar{\theta}}\int_0^px(\theta,b)V(\theta,b)h_{\theta}(b)\,dbf(\theta)\,d\theta+\int_{\underline{\theta}}^{\bar{\theta}}\int_0^p\Omega(\theta)I(\theta,b)\,dbf(\theta)\,d\theta+(\bar{\omega}-\delta)U(\underline{\theta})-\delta\int_0^p[1-g(H_{\underline{\theta}}(b))]\,db,
\]
where $I(\theta,b)$ represents the period 1 information rents:
\[
I(\theta,b)=\frac{1-F(\theta)}{f(\theta)}g'(H_{\theta}(b))\frac{\partial H_{\theta}(b)}{\partial \theta}.
\]

$\Omega(\theta)$ represents the average social weight assigned to all agents with a risk type greater than $\theta$:
\[
\Omega(\theta)=\mathbb{E}_F[\omega(s) \mid s \geq \theta]=\frac{1}{1-F(\theta)}\int_{\theta}^{\bar{\theta}}\omega(s)f(s)\,ds.
\]

$V(\theta,b)$ represents the social value of allocating a unit of the service to an agent with risk type $\theta$ and valuation $b$:
\[
V(\theta,b)=\delta J(\theta,b)-\Omega(\theta)\frac{I(\theta,b)}{h_{\theta}(b)}.
\]
This social value function has a similar economic interpretation to the one in \cite{akbarpour2024redistributive}, but in the context of insurance. In particular, it consists of a weighted sum of the period 1 and period 2 virtual values, captured by $J(\theta, b)$ (corresponding to revenue), and the period 1 information rents, $I(\theta, b)$, weighted by $\Omega(\theta)$.
The weight $\Omega(\theta)$ on the period 1 information rents arises from the envelope formula (Lemma~\ref{lemma:effective-b-type}, Condition~\ref{condition:IC1-necessity}), which implies that, for the menu to satisfy condition~\ref{IC1}, any increase in the ex-ante utility of risk type $\theta$ must also be reflected in the utility of all higher risk types. 

Finally, $\bar{\omega}$ represents the average social weight across all risk types:
\[
\bar{\omega}=\mathbb{E}_F[\omega(\theta)]=\int_{\underline{\theta}}^{\bar{\theta}}\omega(\theta)f(\theta)\,d\theta,
\]
where, whenever $\bar{\omega}$ exceeds the weight on profits, the planner assigns the highest possible ex-ante utility to the lowest type and, through premiums, implements the largest possible transfer to all risk types.

Then, the first-stage social planner's problem is given by:
\begin{align*}
    W(p)\equiv&\max_{x,U(\underline{\theta})} \quad \int_{\underline{\theta}}^{\bar{\theta}}\int_{0}^{p}x(\theta,b)V(\theta,b)h_{\theta}(b)\,dbf(\theta)\,d\theta +(\bar{\omega}-\delta)U(\underline{\theta})+K(p)\tag{$P1-SP$} \label{P1-SP} \\
    &\text{s.t.} \quad \int_{\underline{\theta}}^{\bar{\theta}}\int_{0}^{p}x(\theta,b)h_{\theta}(b)\,dbf(\theta)\,d{\theta} = RS(p), \tag{$EQ$} \\
    &\;\;\;\;\;\;\;\; x(\theta,\cdot):B \rightarrow [0,1] \; \text{is non-decreasing for all } \theta, \tag{$MON$} \\
     &\;\;\;\;\;\;\;\;U(\underline{\theta}) \in \left[U_{NP}(\underline{\theta}),\bar{U}\right], \tag{$IR$}
\end{align*}
where:
\[
K(p)=\int_{\underline{\theta}}^{\bar{\theta}}\int_0^p\Omega(\theta)I(\theta,b)\,dbf(\theta)\,d\theta)-\delta\int_0^p[1-g(H_{\underline{\theta}}(b))]\,db.
\]

The function $W:\left[p^N,p^F\right] \rightarrow \mathbb{R}$ represents the welfare attained by the planner from the optimal menu that sustains price $p$ in equilibrium. Note that, even in the context of a social planner, the implementable prices remain constrained to the interval $\left[p^N,p^F\right]$ (see Lemma~\ref{lemma:feasible-prices}).

Nonetheless, unlike the first-stage insurer's problem~\ref{P1}, it is not immediate that the planner will set the ex-ante utility of type $\underline{\theta}$ equal to her no-participation utility, $U_{NP}(\underline{\theta})$. As argued above, if $\bar{\omega} > \delta$, the planner prefers to set $U(\underline{\theta})$ as high as possible. Therefore, to ensure the maximization problem is well-defined, we assume that the maximum ex-ante utility the planner can assign to type $\underline{\theta}$ is bounded above by some $\bar{U}$ satisfying $U_{NP}(\underline{\theta}) \leq \bar{U} < \infty$.

Once we define problem~\ref{P1-SP}, it is straightforward to see that we can apply the techniques described in Sections~\ref{sec:quantity-space} and~\ref{sec:optimal-transport} to represent \ref{P1-SP} in the quantity space and rewrite it as an optimal transport problem, along with its corresponding dual formulation. Moreover, Lemma~\ref{lemma:optimal-menu} allows us to characterize the solutions to this optimal transport problem. As in the case of a monopolist insurer, we conclude that the social planner will offer to each risk type either a simple deductible contract or a limited coverage deductible contract:

\begin{corollary}[Optimal Contracts Social Planner]
    There exists an $x^*$ that solves the problem~\ref{P1-SP} such that the contract offered to risk type $\theta$ takes one of the following forms: a \textbf{simple deductible contract} $\{(1, D(\theta))\}$ with $D(\theta) \in [0,p]$, or a \textbf{limited coverage deductible contract} $\{ (\alpha,D(\theta)),(1, M(\theta))\}$ where  $0\leq D(\theta) < M(\theta) \leq p$ and $\alpha \in (0,1)$. 
\end{corollary}\label{coro:optimal-menu-SP}

Next, we turn to the second-stage planner's problem, which consists of choosing the price that maximizes the welfare function $W$:
\begin{equation}
    \max_{\left[p^N,p^F\right]}\;W(p) \tag{$P2-SP$} \label{P2-SP}
\end{equation}

Using an argument analogous to the one employed in the proof of Theorem~\ref{thm:envelope}, we can compute the derivative of the welfare function with respect to $p$. For simplicity and clarity of exposition, assume that  $\bar{\omega} > \delta$ and that $V(\theta, \cdot)$ is strictly increasing for all $\theta$. The first assumption implies that $U(\underline{\theta}) = \bar{U}$, and the second implies, via Lemma~\ref{lemma:simple-contract}, that the planner will offer each risk type a simple deductible contract of the form $\{(1, D(\theta; p))\}$, where we again make explicit the dependence of the contract on the price $p$. Then, the derivative of $W$ is given by:

\begin{corollary}[Welfare Derivative]
   The function $W$ is differentiable in $p$ and its derivative is given by:
   \begin{align*}
    W'(p) = &\int_{\underline{\theta}}^{\bar{\theta}} \mathbf{1}_{\left[D(\theta; p) < p\right]} \delta\left[\underbrace{1 - g(H_\theta(p)) + I(\theta,p)}_{\text{Virtual value}}-\underbrace{\left(1 - H_\theta(D(\theta; p)) \right)}_{\text{Payouts}}\right] f(\theta) d\theta \\
    &- \underbrace{\lambda^*(p) \left[ S'(p) + \int_{\underline{\theta}}^{\overline{\theta}} \mathbf{1}{\left[D(\theta; p) = p\right]} h_\theta(p) f(\theta) d\theta \right]}_{\text{Equilibrium effects}}+\int_{\underline{\theta}}^{\bar{\theta}}\underbrace{\mathbf{1}_{\left[D(\theta; p) = p\right]}\Omega(\theta)I(\theta,p)f(\theta)\,d\theta}_{\text{Information rents}}\\
    &-\underbrace{\delta[1-g(H_{\underline{\theta}}(p))]}_{\text{Outside Option}}.
\end{align*}
\end{corollary}\label{coro:envelope-SP}
Compared to the derivative of profits, $\Pi'(p)$, the derivative of welfare assigns an additional negative weight to the period 1 information rents of agents who do not receive coverage (i.e., those for whom $D(\theta; p) = p$). This is captured by the term $\Omega(\theta) I(\theta, p)$, where, recall, $I(\theta, p) \leq 0$. Intuitively, these are forgone rents that the planner could have paid to uncovered types to increase their ex-ante utility.

As in the case of a monopolist provider, when $V(\theta, \cdot)$ is strictly increasing for all $\theta$, we can apply an argument analogous to the one in Proposition~\ref{prop:comparative-statics} to show that $\lambda^*(\cdot)$ is strictly increasing and continuous, and that $D(\theta; \cdot)$ is continuous and strictly increasing whenever $D(\theta; p) \in (0, p)$, while it is non-increasing when $D(\theta; p) = 0$. Moreover, we define the price-taking benchmark $p^U$ as the price that clears the services market and solves problem~\ref{P1-SP} without the~\ref{EQ} constraint (i.e., such that $\lambda^*(p^U) = 0$). Existence and uniqueness of $p^U$ then follow from the argument in Lemma~\ref{lemma:price-taking}. Finally, using the derivative of the welfare function, we can show that:

\begin{proposition}\label{prop:price-taking-planner}
   Suppose that $V(\theta, \cdot)$ is strictly increasing for all $\theta$, $\delta = 0$, and $p^U \in (p^N, p^F)$. Then, the price chosen by the planner in the second-stage problem~\ref{P2-SP} is strictly lower than $p^U$.
\end{proposition}
We omit the proof of Proposition~\ref{prop:price-taking-planner}, as it follows from the fact that $W'(p) < 0$ for all $p \geq p^U$. Indeed, since $\delta = 0$, the planner assigns zero weight to revenue. As a result, the derivative of $W$ depends only on the term capturing forgone information rents, which is always negative, and on the term capturing equilibrium effects, which is also negative above the price-taking benchmark. This is because the bracketed term is always positive, and $\lambda^*(\cdot)$ is strictly increasing. Hence, $-\lambda^*(p) < -\lambda^*(p^U) = 0$ for all $p > p^U$.

\subsection{Nash Bargaining}\label{sec:nash-bargaining}
In this section, we show how our baseline model can be extended to a setting in which, in the second stage, the insurer does not necessarily have full control over the service price, as the service provider may also possess bargaining power. Under this formulation, our baseline model corresponds to the special case in which the insurer holds all the bargaining power.

More precisely, the timing and interaction between the service provider and insurer is as follows:
\begin{enumerate}
    \item First, the insurer and the service provider bargain over the service price $p$ and a capacity cap $S$, which specifies the maximum number of patients the provider is willing to accept from the insurer. The pair $(p,S)$ is determined through simultaneous Nash bargaining, with the outcome maximizing the Nash product:
    \[
        \Pi_i(p,S)^{\beta} \Pi_s(p,S)^{1-\beta},
    \]
    where $\Pi_i$ and $\Pi_s$ denote the profits of the insurer and the service provider, respectively, and $\beta \in [0,1]$ is the insurer's bargaining weight.

    The fact that the insurer and the provider bargain over $S$ is motivated by empirical evidence showing that providers often face real or strategic capacity constraints. For example, \cite{ho2009insurer} analyzes a setting in which hospitals encounter such constraints during bargaining, either due to physical limitations---such as limited numbers of beds or nurses per bed---or because a hospital is highly attractive to consumers and thus less reliant on insurers to secure patient volume. \cite{ho2009insurer} finds that hospitals expected to fill their beds can leverage this position to obtain approximately \$6{,}900 more per patient than average. Similarly, \cite{town2001hospital} argue that hospitals with little excess capacity tend to have higher reservation prices and, as a result, stronger bargaining positions.

    Moreover, in many countries such constraints are explicitly written into contracts. For instance, in European health systems, contracts often include explicit volume caps, sometimes called ``activity caps'' or ``contracted volumes'' \citep{appleby2012payment}.

    Finally, this formulation also accommodates a simpler version in which only the price $p$ is negotiated. In that case, $p$ indirectly determines the provider’s capacity constraint through a function $S:\mathbb{R}_+ \to \mathbb{R}_+$, where, for example, $S(p)=S$ for all $p$, implying that the hospital’s capacity constraint is invariant to price.

    \item Once $p$ and $S$ are agreed upon, the insurer offers an insurance menu to the agents:
    \[
        \left\{\left(\left\{ \left(x(\theta,b), D(\theta,b) \right) \right\}_{b \in B}, \, t(\theta) \right)\right\}_{\theta \in \Theta}.
    \]
    The menu is offered at a stage when each agent observes only her ex-ante risk type $\theta$, before her loss-related valuation $b$ is realized. The agent also observes the agreed price $p$, and decides whether to select a contract from the menu. If she does, she pays the corresponding premium.

    \item The agent's private loss-related valuation $b$ is then realized. Based on this information, she chooses whether to the seek the service or remain at home. If she seeks the service, she can either use the benefit specified in her contract to cover the service costs, or pay entirely out of pocket.

    \item The collection of contracts selected by the agents, along with the specific items used within those contracts, generates total demand for services at the agreed price $p$. This demand must satisfy the capacity constraint $S$—that is, the total volume of customers referred by the insurer must not exceed $S$.
\end{enumerate}

Note that for any fixed price $p$, we can still apply Lemmas~\ref{lemma:IC2-contract} and~\ref{lemma:IC1-menu} to characterize an implementable insurance menu---that is, a menu that satisfies conditions~\ref{IC2},~\ref{IC1}, and~\ref{IR}. Moreover, this implies that the insurer’s profits from any implementable menu are still determined by Lemma~\ref{lemma:objective}.

Then, the insurer’s problem is given by:

\begin{align*}
    \Pi_i(p)\equiv&\max_{x} \quad \int_{\underline{\theta}}^{\bar{\theta}}\int_{0}^{p}x(\theta,b)J(\theta,b)h_{\theta}(b)\,dbf(\theta)\,d\theta +K(p)\tag{$P1-NB$} \label{P1-NB} \\
    &\text{s.t.} \quad \int_{\underline{\theta}}^{\bar{\theta}}\int_{0}^{p}x(\theta,b)h_{\theta}(b)\,dbf(\theta)\,d{\theta}+\int_{\underline{\theta}}^{\bar{\theta}}[1-H_{\theta}(p)]f(\theta)\,d\theta \leq S, \tag{$CAP$} \label{CAP} \\
    &\;\;\;\;\;\;\;\; x(\theta,\cdot):B \rightarrow [0,1] \; \text{is non-decreasing for all } \theta, \tag{$MON$}
\end{align*}
where:
\[
K(p)=- U_{NP}(\underline{\theta})- \int_{0}^{p} [1 - g(H_{\underline{\theta}}(b))] \, db.
\]
The main difference between problem~\ref{P1} and problem~\ref{P1-NB} is that we replace the equilibrium constraint~\ref{EQ} with the capacity constraint~\ref{CAP}. This capacity constraint requires that the volume of agents referred by the insurer to the service provider not exceed $S$.

It is straightforward to apply the techniques described in Sections~\ref{sec:quantity-space} and~\ref{sec:optimal-transport} to represent problem~\ref{P1-NB} in quantity space and rewrite it as an optimal transport problem, along with its corresponding dual formulation. Moreover, Lemma~\ref{lemma:optimal-menu} allows us to characterize the solutions to this optimal transport problem. As in the case of a competitive market for services, we conclude that the insurer offers each risk type $\theta$ either a simple deductible contract or a limited coverage contract:

\begin{corollary}[Optimal Contracts under Nash Bargaining]
    There exists an $x^*$ that solves the problem~\ref{P1-NB} such that the contract offered to risk type $\theta$ takes one of the following forms: a \textbf{simple deductible contract} $\{(1, D(\theta))\}$ with $D(\theta) \in [0,p]$, or a \textbf{limited coverage deductible contract} $\{ (\alpha,D(\theta)),(1, M(\theta))\}$ where  $0\leq D(\theta) < M(\theta) \leq p$ and $\alpha \in (0,1)$. 
\end{corollary}\label{coro:optimal-menu-NB}

Using the solution to problem~\ref{P1-NB}, we can associate each price $p$ and capacity cap $S$ with the corresponding insurer profits, denoted $\Pi_i(p,S)$. We then define the \textit{Nash bargaining price and capacity cap} as the solution to the problem:
\begin{equation}
    \max_{p,S} \; \Pi_i(p,S)^{\beta} \Pi_s(p,S)^{1 - \beta}, \tag{$P2$-$SP$} \label{P2-NB}
\end{equation}
Note that when $\beta = 1$---that is, when the insurer holds all the bargaining power---and when the capacity cap $S$ is indirectly determined by the agreed price $p$, problem~\ref{P2-NB} reduces to the second-stage problem~\ref{P2} of the baseline model.

Moreover, we can use the argument made in Theorem \ref{thm:envelope} to express the derivative of $\Pi_i$ with respect to $p$ and $S$:

\begin{corollary}[Insurer’s Profit Derivative under Nash Bargaining]
   The function $\Pi_i$ is differentiable in $p$ and $S$; its partial derivatives are:
    \begin{align*} \frac{\partial \Pi_i(p,S)}{\partial p}=&\int_{\underline{\theta}}^{\overline{\theta}}x^*(\theta,p)J(\theta,p)h_{\theta}(p)f(\theta)\,d\theta- \int_{\underline{\theta}}^{\bar{\theta}}\int_0^px^*(\theta,b)h_{\theta}(b)\,dbf(\theta)\,d\theta\\
    &-\lambda^* \left(\int_{\underline{\theta}}^{\overline{\theta}}(1-x^*(\theta,p))h_\theta(p)f(\theta)d\theta\right)-[g(H_{\underline{\theta}}(p+\beta))-g(H_{\underline{\theta}}(p))],\\
    \frac{\partial \Pi_i(p,S)}{\partial S}=&-\lambda^*.
    \end{align*}
\end{corollary}\label{coro:envelope-NB}

Corollary~\ref{coro:envelope-NB} is key for analyzing the Nash-bargaining problem~\ref{P2-NB}. For this problem, we have not yet imposed any assumptions on the service provider’s profit function $\Pi_s$. In practice, depending on the setting, these profits may even depend on the insurance menu chosen by the insurer in problem~\ref{P1-NB}. For example, one could assume:
\[
\Pi_s(p) = \int_{\underline{\theta}}^{\bar{\theta}} \int_{0}^{p} x^*(\theta, b) (p - mc) h_{\theta}(b) \, db \, f(\theta) \, d\theta + \int_{\underline{\theta}}^{\bar{\theta}} [1 - H_{\theta}(p)] (p - mc) f(\theta) \, d\theta,
\]
where the service provider faces constant marginal cost $mc$, and $x^*(\theta, b)$ denotes the units of the service provided to an agent of type $\theta$ with loss-related valuation $b \leq p$ in the optimal menu solving problem~\ref{P1-NB}. Recall from Lemma~\ref{lemma:IC2-contract} that $x^*(\theta, b) = 1$ for all $b > p$.

\subsection{Pricing Tiers}\label{sec:multiple-prices}

Service providers typically offer \textit{tiered services}, differentiated by the level of care, specialization, and quality provided. For example, a hospital may offer primary care or routine checkups at a base price $p_1$, while secondary care---such as treatment from a specialist---is available for an additional cost $p_2$. Similarly, basic lab panels may be priced at $p_1$, whereas more advanced diagnostics like imaging require an extra payment of $p_2$. In terms of amenities, a standard room might be available at price $p_1$, while private rooms or more luxurious accommodations come at an added cost $p_2$.

In our initial model, it is natural to assume that agents with higher valuations are more likely to demand additional services, and thus incur higher treatment costs. To formalize this idea, we continue to assume that each agent is characterized by her risk type $\theta \in \Theta$ and valuation $b \in B$. The service provider now offers \emph{tiered pricing}: a basic service is available for a price $p_1 \in \mathbb{R}_+$ and additional (or higher-quality) services are offered for an extra price $p_2 \in \mathbb{R}_+$.

An agent’s loss when her valuation is $b$ and she receives only basic services is given by $p_1 + \max\{b - b^*, 0\}$, for some cap $b^* \in B$. Conversely, her loss when she only receives additional  services is $\min\{b,b^*\} + p_2$. This implies that the agent's valuation for basic services is $\min\{b, b^*\}$, reflecting a cap or threshold $b^*$ on the loss that basic care can prevent. The valuation for additional services is then $\max\{b - b^*, 0\}$---only agents with $b > b^*$ value these services. 

Recall that $b$ represents the agent’s willingness to pay for third-party services. Since this valuation is subjective and shaped by many personal and contextual factors, we model $b$ as private information.

One key factor influencing $b$ is the \textit{severity} of the agent’s condition: more severe illnesses typically generate higher perceived losses and, consequently, higher valuations for treatment. Thus, a higher value of $b$ corresponds to a more serious health condition. This link between severity and valuation is explicitly captured in this extension of our baseline model.

In particular, when an agent's valuation $b$ exceeds a threshold $b^*$, we can interpret this as the agent suffering a more severe adverse event than one with $b < b^*$. In such cases, a basic service (i.e., a routine checkup, denoted as service~1) is insufficient to fully treat or repair damages. The agent also requires a more specialized intervention (i.e., seeing a specialist, denoted as service~2). As a result, treating the more severe condition is costlier: the agent must pay the combined cost $p_1 + p_2$ rather than just $p_1$.

We assume that $p_1 \in [0,b^*]$ such that this structure yields three types of behavior:
\begin{itemize}
    \item Agents with $b \leq p_1$ do not seek the service and incur a loss of $b$.
    \item Agents with $b \in (p_1, b^* + p_2]$ receive only basic services and incur a loss of $p_1 + \max\{b - b^*, 0\}$.
    \item Agents with $b > b^* + p_2$ receive both basic and additional services and incur a loss of $p_1 + p_2$.
\end{itemize}

It is easy to verify that the certainty equivalent for a risk type $\theta$ agent when a service provider offers tiered services is given by:
\begin{align*}
    -\sum_{j=1}^{2}\int_{b_j}^{b_j+p_j}[1-g(H_{\underline{\theta}}(b))]\,db,
\end{align*}
where $b_1=0$ and $b_2=b^*$.

Moreover, an insurance menu is now represented as:
\[
    \left\{ \left( \left\{ \left( x_j(\theta, b), D_j(\theta, b) \right) \right\}_{b \in B,\; j \in \{1,2\}},\; t(\theta) \right) \right\}_{\theta \in \Theta},
\]
where a contract $\left\{ \left(x_j(\theta, b), D_j(\theta, b) \right) \right\}_{b \in B,\; j \in \{1,2\}}$ specifies, for each service $j$, a fixed per unit price $D_j(\theta, b)$ that must be paid for the $x_j(\theta, b)$ units of the service covered by the insurer.

The agent's loss function when her true valuation is $b$ and she picks contract item $\left\{(x_j(\theta, b'), D_j(\theta, b'))\right\}_{j \in \{1,2\}}$ is defined as follows:

\[
L(b,b';\theta)=L_1(b,b';\theta)+L_2(b,b';\theta),
\]
where 
\[
L_j(b,b';\theta)=x_j(\theta,b')\min\{D_j(\theta,b'),b_{p_j}\}+(1-x_j(\theta,b'))b_{p_j},
\]
and $b_{p_j}=\min\{b-b_j,p_j\}$.

We can then use an argument similar to the one used in the proof of Lemma~\ref{lemma:IC2-contract} to characterize the contracts that satisfy the period 2 incentive compatibility constraint~\ref{IC2}. In particular, a contract satisfies~\ref{IC2} if and only if:
\begin{enumerate}
    \item The agent's loss function is given by:
    \[
        L(b;\theta)\equiv L(b,b;\theta) =\int_0^{\min\{b,p_1\}}(1-x_1(\theta,l)\,dl+\int_0^{\min\{\max\{b^*,b\},b^*+p_2\}}(1-x_2(\theta,l)\,dl,
    \]
    \item $x_j(\theta,\cdot):B\to[0,1]$ is a non-decreasing function, with $x_j(\theta,b)=1$ for all $b>b_j+p_j$ and $x_2(\theta,b)=0$ for all $b \leq b^*$.
\end{enumerate}

Then, if risk type $\theta$ selects the menu option:
\[
    \left( \left\{ \left( x_j(\theta', b), D_j(\theta', b) \right) \right\}_{b \in B,\; j \in \{1,2\}},\; t(\theta') \right),
\]
the agent's certainty equivalent is given by:
\[
    U(\theta,\theta')=-t(\theta')-\sum_{j=1}^{2}\int_{b_j}^{b_j+p_j}(1-x_j(\theta,b))[1-g(H_{\theta}(b))]\,db,
\]

Applying an argument similar to the one used in Lemma~\ref{lemma:IC1-menu}, we can show that any menu satisfying the period 1 incentive compatibility constraint~\ref{IC1} must also satisfy the envelope formula:
\[
U(\theta)\equiv U(\theta,\theta)=U(\underline{\theta})+\sum_{j=1}^{2}\int_{\underline{\theta}}^{\theta}\int_{b_j}^{b_j+p_j}(1-x_j(s,b))g'(H_{s}(b))\frac{\partial H_{s}(b)}{\partial s}\,db\,ds.
\]

For simplicity, we also assume that $\beta_j=0$, so that if the insurer implements price $p_j$, then $p_j$ is the price agents face if they choose not to participate in the mechanism. Using the argument from Lemma~\ref{lemma:IR-menu}, we can then show that a menu satisfies the period 1 individual rationality constraint~\ref{IR} if and only if:
\[
U(\underline{\theta}) \geq -\sum_{j=1}^{2}\int_{b_j}^{b_j+p_j}[1-g(H_{\underline{\theta}}(b))]\,db
\]

Additionally, if the supply of service $j$ is given by a strictly increasing function $S_j: \mathbb{R}+ \to \mathbb{R}+$, then we can define the residual supply of service $j$ faced by the insurer at price $p_j$ as:
\[
RS_j(p)=S_j(p)-\int_{\underline{\theta}}^{\bar{\theta}}[1-H_{\theta}(b_j+p_j)]f(\theta)\,d\theta.
\]
This residual supply accounts for the fact that, regardless of the menu offered by the insurer, all agents with a valuation $b > b_j + p_j$ will demand the full unit of service $j$.

Then, the equilibrium constraint faced by the insurer for service $j$ is given by:
\[
\int_{\underline{\theta}}^{\bar{\theta}}\int_{b_j}^{b_j+p_j}x_j(\theta,b)h_{\theta}(b)\,dbf(\theta)\,d\theta=RS_j(p).
\]

We can then apply an argument similar to the one used in Lemma~\ref{lemma:feasible-prices} to show that the set of implementable prices for service $j$ lies in the compact interval $\left[p_j^N, p_j^F\right]$, where, recall, $p_j^N$ is the price that clears the market for service $j$ when $x_j(\theta, b) = 0$ for all $b \in [b_j, b_j + p_j]$, and $p_j^F$ is the price that clears the market when $x_j(\theta, b) = 1$ for all $b \in [b_j, b_j + p_j]$.

Additionally, to remain consistent with our assumption that $p_1 \in [0, b^*]$, we assume that $p_1^F \leq b_2 = b^*$.

Next, applying the same logic as in Lemma~\ref{lemma:objective}, we can show that the insurer's profits are given by:
\[
\sum_{j=1}^2\left[\int_{\underline{\theta}}^{\bar{\theta}}x_j(\theta,b)J_j(\theta,b)h_{\theta}(b)\,dbf(\theta)\,d\theta-\int_{b_j}^{b_j+p_j}[1-g(H_{\underline{\theta}}(b))]\,db\right]-U(\underline{\theta}),
\]
where:
\[
J_j(\theta,b)=\left[1 - g(H_{{\theta}}(b)) + \frac{1 - F(\theta)}{f(\theta)} g'(H_{\theta}(b)) \frac{\partial H_{\theta}(b)}{\partial \theta}\right]\frac{1}{h_{\theta}(b)} +b-b_j - p_j - \frac{1 - H_{\theta}(b)}{h_{\theta}(b)}.
\]

From here, it is straightforward to see that the first-stage problem can be separated by service $j$, with the corresponding problem given by:
\begin{align*}
    \Pi_j(p)\equiv&\max_{x_j} \quad \int_{\underline{\theta}}^{\bar{\theta}}\int_{b_j}^{b_j+p_j}x_j(\theta,b)J_j(\theta,b)h_{\theta}(b)\,dbf(\theta)\,d\theta \tag{$P1_j$} \label{P1_j} \\
    &\text{s.t.} \quad \int_{\underline{\theta}}^{\bar{\theta}}\int_{b_j}^{b_j+p_j}x_j(\theta,b)h_{\theta}(b)\,dbf(\theta)\,d{\theta} = RS_j(p), \tag{$EQ_j$} \label{EQ_j} \\
    &\;\;\;\;\;\;\;\; x_j(\theta,\cdot):B \rightarrow [0,1] \; \text{is non-decreasing for all } \theta, \tag{$MON_j$} \label{MON_j}
\end{align*}
where we have already used the fact that the insurer sets the utility of the lowest type $\underline{\theta}$ equal to her outside option.

It is straightforward to apply the techniques described in Sections~\ref{sec:quantity-space} and~\ref{sec:optimal-transport} to represent problem~\ref{P1_j} in quantity space and to rewrite it as an optimal transport problem, along with its corresponding dual formulation. Moreover, Lemma~\ref{lemma:optimal-menu} allows us to characterize the solutions to this optimal transport problem. We conclude that the insurer offers each risk type~$\theta$ either a simple deductible contract or a limited coverage deductible contract:

\begin{corollary}[Optimal Contracts Tiered Services]
    There exists an $x^*$ that solves the problem~\ref{P1_j} such that the contract offered to risk type $\theta$ for service $j$ takes one of the following forms: a \textbf{simple deductible contract} $\{(1, D_j(\theta))\}$ with $D_j(\theta) \in [0,p_j]$, or a \textbf{limited coverage deductible contract} $\{ (\alpha_j,D_j(\theta)),(1, M_j(\theta))\}$ where  $0\leq D_j(\theta) < M_j(\theta) \leq p_j$ and $\alpha_j \in (0,1)$. 
\end{corollary}\label{coro:optimal-menu-TS}

Finally, we can also decompose the second-stage problem by service $j$. In particular, the problem is given by:
\begin{equation}
    \max_{p_j \in \left[p_j^N, p_j^F\right]} \; \Pi_j(p), \tag{$P2_j$} \label{P2_j},
\end{equation}
where using Theorem \ref{thm:envelope} we can characterize the derivative of $\Pi_j$:

\begin{corollary}[Insurer's Profit Derivative Tiered Services]
   The function $\Pi_j$ is differentiable in $p$ and its derivative is given by:
    \begin{align*} \Pi_j'(p)=&\int_{\underline{\theta}}^{\overline{\theta}}x_j^*(\theta,b_j+p_j)J_j(\theta,b_j+p_j)h_{\theta}(p)f(\theta)\,d\theta- \int_{\underline{\theta}}^{\bar{\theta}}\int_{b_j}^{b_j+p_j}x_j^*(\theta,b)h_{\theta}(b)\,dbf(\theta)\,d\theta\\
    &-\lambda^* \left( S_j'(p) + \int_{\underline{\theta}}^{\overline{\theta}}(1-x_j^*(\theta,b_j+p_j))h_\theta(p)f(\theta)d\theta\right).
    \end{align*} 
\end{corollary}\label{coro:envelope-TS}

\section{Conclusion}\label{sec:conclusion}
This paper develops a mechanism design framework to study how insurance contracts interact with third-party service markets through their impact on service prices. We analyze a setting in which a monopolistic insurer offers contracts to risk-averse agents with dual utility preferences and sequential two-dimensional private information. Crucially, the insurer internalizes the effect of coverage on the demand for services, which in turn determines equilibrium prices. We approach this problem in two stages: first, we characterize optimal insurance menus for any given service price; second, we identify the equilibrium price that maximizes the insurer’s profits.

Our first main result shows that optimal contracts take a simple and realistic form: either full coverage after a deductible or limited coverage with an out-of-pocket maximum. The limited coverage design emerges as a tool not only to manage risk but also to control service utilization and price levels. Our second main result identifies three forces that shape the insurer’s pricing decision---surplus extraction, higher payouts, and equilibrium effects---offering sufficient conditions under which the insurer sets prices above or below a price-taking benchmark. 

We also extend the model to settings with redistributive objectives, negotiated prices, and pricing tiers. These extensions preserve the core insights of the baseline model while broadening its applicability to both policy-oriented and institutional contexts. Overall, our results contribute to a better understanding of how insurers may strategically design contracts in environments where prices and coverage are jointly determined through market interactions.

\newpage
% Bibliography
\bibliography{references.bib}
\bibliographystyle{aer}

% Appendix
\newpage
\appendix
\section{Proofs of Results in Section~\ref{sec:implementable}} \label{sec:proof-implmentable}
\begin{proof}[Proof of Lemma~\ref{lemma:IR-contract}] 
       Suppose that a contract satisfies~\ref{IC2} and that $D(\theta,b) > b_p$ for some $b$. In this case, the agent will not use her insurance and her loss is given by $L(b;\theta)=b_p$. 
       
       Now, suppose the insurer modifies the contract by setting $D'(\theta,b)=b_p$ and $x'(\theta,b)=0$. Under this modified contract, the loss remains unchanged at $L'(b;\theta)=b_p$. The treatment cost for the insurer remains equal to 0, the demand for hospital services remains 1 if $b > p$ and 0 otherwise. Finally, for every $b' \neq b$, $L'(b',b;\theta)=b_p \geq L(b';\theta)$. Therefore,~\ref{IC2} is still satisfied.
\end{proof}

\begin{proof}[Proof of Lemma~\ref{lemma:effective-b-type}]
    Suppose that a contract satisfies~\ref{IC2}. Towards a contradiction, assume that $L(b;\theta)>L(p;\theta)$ for some $b>p$. Then: 
    \[
        L(b,p;\theta)=x(\theta,p)\min\{D(\theta,p),p\}+(1-x(\theta,p))p= L(p;\theta)< L(b;\theta),
    \]
    which contradicts the assumption that \ref{IC2} is satisfied. Therefore, we must have $L(b;\theta) \leq L(p;\theta)$. A similar argument shows that $L(b;\theta) \geq L(p;\theta)$, hence $L(p;\theta)=L(b;\theta)$ for all $b > p$.

     To see why we can, with out loss of generality, set $x(\theta,b)=1$ for all $b > p$, suppose instead that $x(\theta,b)<1$ for some $b > p$. Consider a modified contract where $x'(\theta,b)=1$ and $D'(\theta,b)=L(p;\theta)$ for all $b > p$. Then, $L'(b;\theta)=\min\{L(p;\theta),p\}=L(p;\theta)=L(b;\theta)$. Thus, the loss remains unchanged. Moreover, $p-D'(\theta,b)=p-L(p;\theta)=p-L(b;\theta)=x(\theta,b)(p-D(\theta,b))$, meaning the treatment cost to the insurer also remains the same. The demand for hospital services is still 1. Finally, for any $b \leq p$, we have:
    \begin{align*}
        L(b;\theta) &\leq L(b,p;\theta) \leq x(\theta,p)\min\{D(\theta,p),p\} +(1-x(\theta,p))p= L(p;\theta).
    \end{align*}
    Additionally, we know that $L(b;\theta) \leq b$. Hence, $ L(b;\theta) \leq \min\{L(p;\theta),b\}$. Since $L'(b,p;\theta)=\min\{L(p;\theta),b\}$,~\ref{IC2} is still satisfied.
\end{proof}

\begin{proof}[Proof of the Necessity part of Lemma~\ref{lemma:IC2-contract}]
    Suppose that a contract satisfies~\ref{IC2}. Then, for any $b,b' \in [0,p]$, we have:
    \begin{align*}
        L(b;\theta)\leq L(b,b';\theta)= (1-x(\theta,b'))b+x(\theta,b')\min\{D(\theta,b'),b\} \leq (1-x(\theta,b'))b+x(\theta,b')D(\theta,b').
    \end{align*}
    
    Thus:
    \[
    L(b;\theta)=\min_{b'\in [0,p]} \; \; x(\theta,b')D(\theta,b')+(1-x(\theta,b'))b.
    \]
    
    Using the envelope theorem (see~\citealp{milgrom2002envelope} Theorem 2), it follows that: 
    \[
    L(b;\theta)=L(0;\theta)+\int_0^{b}(1-x(\theta,l))\,dl.
    \]
    
    Furthermore, since $0 \leq L(b;\theta) \leq  b$, we have $L(0;\theta)=0$. Thus:
    \[
    L(b;\theta)=\int_0^{b}(1-x(\theta,l))dl.
    \]
    
    We now show Condition~\ref{condition:monotonicity-contract} of Lemma~\ref{lemma:IC2-contract}. From~\ref{IC2}, for $p \geq b > b'$, we have:
    \begin{align*}
        (1-x(\theta,b))b+x(\theta,b)D(\theta,b) &\leq (1-x(\theta,b'))b+x(\theta,b')\min\{D(\theta,b'),b\}  \\
        &= (1-x(\theta,b'))b+x(\theta,b')D(b',\theta),\\
        (1-x(\theta,b'))b'+x(\theta,b')D(\theta,b') &\leq (1-x(\theta,b))b'+x(\theta,b)\min\{D(\theta,b),b'\}  \\
        &\leq (1-x(\theta,b))b'+x(\theta,b)D(b,\theta). 
    \end{align*}
    
    If we combine both inequalities we get:
    \begin{align*}
        (1-x(\theta,b))(b-b') &\leq (1-x(\theta,b'))(b-b') \\
        x(\theta,b') &\leq x(\theta,b).
    \end{align*}
\end{proof}

\begin{proof}[Proof of the Sufficiency part of Lemma~\ref{lemma:IC2-contract}]
    Suppose a contract satisfies Conditions~\ref{condition:envelope-contract} and~\ref{condition:monotonicity-contract} of Lemma~\ref{lemma:IC2-contract}. For $b,b' \in [0,p]$ with $b' \leq b$, we have:
    \begin{align*}
        L(b,b';\theta)&=x(b',\theta)D(b',\theta) + (1-x(b',\theta))b  \\
        &=L(b';\theta)+(b-b')(1-x(\theta,b'))\\
        &=\int_0^{b'}(1-x(\theta,l))\,dl+\int_{b'}^{b}(1-x(\theta,b'))\,dl \nonumber \\
        &\geq \int_0^{b'}(1-x(\theta,l))\,dl+\int_{b'}^{b}(1-x(\theta,l))\,dl \\
        &= L(b;\theta), 
    \end{align*}
    where the inequality follows from Condition~\ref{condition:monotonicity-contract}, and the last equality from Condition~\ref{condition:envelope-contract}. This shows~\ref{IC2} is satisfied for $b,b' \in [0,p]$ with $b' \leq b$.
    
    Similarly, for $b,b' \in [0,p]$ with $b \leq b'$, $L(b,b';\theta)$ can take two possible values. If $D(\theta,b')>b$, $L(b,b';\theta)=b \geq L(b;\theta)$ and~\ref{IC2} is trivially satisfied. If $C(\theta,b') \leq b$, we have:
    \begin{align*}
       L(b,b';\theta)&=L(b';\theta)-(b'-b)(1-x(\theta,b')) \nonumber \\
        &=\int_0^{b'}(1-x(\theta,l))\,dl-\int_{b}^{b'}(1-x(\theta,b'))\,dl \\
        &\geq \int_0^{b'}(1-x(\theta,l))\,dl-\int_{b}^{b'}(1-x(\theta,l))\,dl  \\
        &= L(b;\theta),
    \end{align*}
    This shows~\ref{IC2} is satisfied for $b,b' \in [0,p]$ with $b \leq b'$.

    Next, for $b>p$ and $b' \in [0,p]$, we have that $L(b,b';\theta)=L(p,b';\theta)$ and $L(p;\theta) = L(b;\theta)$. From earlier arguments we know that $L(p;\theta) \leq L(p,b';\theta)$. Then, $L(b;\theta) \leq L(b,b';\theta)$ and~\ref{IC2} is satisfied for $b>p$ and $b' \in [0,p]$.
    
    Finally, for $b \in [0,p]$ and $b'>p$, we have that $x(\theta,b')=1$ and $L(b';\theta)=L(p,\theta)$. Thus, $L(b,b';\theta)=\min\{L(p;\theta),b\}$. Furthermore, we know that $L(b;\theta) \leq b $ and that:
    \begin{align*}
        L(p;\theta) &= x(\theta,p)D(\theta,p)+(1-x(\theta,p))p\\
        &\geq x(\theta,p)D(\theta,p)+(1-x(\theta,p))b \\
        &= L(b,p;\theta) \\
        &\geq L(b,\theta),
    \end{align*}
    where the last inequality follows from earlier arguments. Therefore, $L(b,;\theta) \leq \min\{L(p;\theta),b\}=L(b,b';\theta)$ and~\ref{IC2} is satisfied for $b \in [0,p]$ and $b'>p$.
\end{proof}

\begin{proof}[Proof of Condition~\ref{condition:IC1-necessity} of Lemma~\ref{lemma:IC1-menu}] 
    Suppose that a menu satisfies~\ref{IC1}. Using the envelope theorem (see \citet{milgrom2002envelope} Theorem 2), the agent's certainty equivalent is given by:
    \[
        U(\theta)=U(\underline{\theta})+\int_{\underline{\theta}}^{\theta}\int_0^{p}(1-x(s,b))g'(H_s(b))\frac{\partial H_s(b)}{\partial s}\,db\,ds.
    \]
\end{proof}

\begin{proof}[Proof of Condition~\ref{condition:IC1-sufficiency} of Lemma~\ref{lemma:IC1-menu}]  
    Suppose that a menu satisfies Condition~\ref{condition:IC1-necessity} of Lemma~\ref{lemma:IC1-menu} and $x(\cdot,b):B \rightarrow [0,1]$ is a non-decreasing function for every $\theta$. Then, for any $\theta > \theta'$, we have that:
    \begin{align*}
        U(\theta,\theta')&=-t(\theta')-\int_{0}^{p}(1-x(\theta',b))[1-g(H_{\theta}(b))] \, db\\
        &= U(\theta') + \int_0^{p}(1-x(\theta',b))[1-g(H_{\theta'}(b))]\, db-\int_{0}^{\bar{b}}(1-x(\theta',b))[1-g(H_{\theta}(b))] \, db\\
        &= U(\theta')+\int_{\theta'}^{\theta}\int_{0}^{p}(1-x(\theta',b))g'(H_{s}(b))\frac{\partial H_{s}(b)}{\partial s}\, db \, ds \\
        &= U(\theta)+\int_{\theta'}^{\theta}\int_0^{p}[x(s,b)-x(\theta',b)]g'(H_s(b))\frac{\partial H_s(b)}{\partial s}\, db \, ds \\
        &\leq U(\theta),
    \end{align*}
    where the second equality follows from the definition of $U(\theta')$, while the fourth equality is a consequence of Condition~\ref{condition:envelope-contract}. The inequality stems from the assumptions that $x(\cdot,b)$ is a non-decreasing function, $g$ is increasing, and $\frac{\partial H_{\theta}(b)}{\partial \theta} \leq 0$ for all $b$. This last fact arises because $H_{\theta}$ increases in the sense of first-order stochastic dominance. An equivalent argument shows that same inequality holds when $\theta < \theta'$.
\end{proof}

\begin{proof}[Proof of Lemma~\ref{lemma:IR-menu}] 
    The necessity part follows trivially. In order to show the sufficiency part, suppose that a menu satisfies~\ref{IC1} and that~\ref{IR} holds for type $\underline{\theta}$. Additionally, we assume that we are in the \textit{endogenous case} (an equivalent proof applies to the \textit{exogenous case}) such that the certainty equivalent that type $\theta$ gets from not participating is:
    \begin{align*}
         U_{NP}(\theta)&=-\int_0^{\beta+p}[1-g(H_{\theta}(b))]\,db \\
         &= U_{NP}(\underline{\theta})+\int_{\underline{\theta}}^{\theta}\int_0^{\beta+p}g'(H_s(b))\frac{\partial H_s(b)}{\partial s}\,db\,ds.
    \end{align*}
    
    Then, by Lemma~\ref{lemma:IC1-menu} Condition~\ref{condition:IC1-necessity} we get that:
    \begin{align*}
    U(\theta)-U_{NP}(\theta)=&U(\underline{\theta})-U_{NP}(\underline{\theta})-\int_{\underline{\theta}}^{\theta}\int_0^px(s,b)g'(H_s(b))\frac{\partial H_s(b)}{\partial s}\,db\,ds\\
    &-\int_{\underline{\theta}}^{\theta}\int_p^{\beta+p}g'(H_s(b))\frac{\partial H_s(b)}{\partial s}\,db\,ds.
    \end{align*}
    
    By assumption $U(\underline{\theta})-U_{NP}(\underline{\theta}) \geq 0$. Furthermore, $x(\theta,b) \geq 0$, $g$ is increasing and $\frac{\partial H_{\theta}(b)}{\partial \theta} \leq 0$. Hence: $U(\theta)-U_{NP}(\theta) \geq 0$.
\end{proof}

%\begin{proof}[Proof of Lemma~\ref{lemma:IR-menu}] 
    %The necessity part follows trivially. In order to show the sufficiency part, suppose that a menu satisfies~\ref{IC1} and that~\ref{IR} holds for type $\underline{\theta}$. Additionally, we assume that we are in the \textit{endogenous case} (an equivalent proof applies to the \textit{exogenous case}) such that the certainty equivalent that type $\theta$ gets from not participating is:
    %\begin{align*}
         %U_{NP}(\theta)&=-\int_0^{p}[1-g(H_{\theta}(b))]\,db \\
         %&= U_{NP}(\underline{\theta})+\int_{\underline{\theta}}^{\theta}\int_0^pg'(H_s(b))\frac{\partial H_s(b)}{\partial s}\,db\,ds.
    %\end{align*}
    
    %Then, by Lemma~\ref{lemma:IC1-menu} Condition~\ref{condition:IC1-necessity} we get that:
    %\[
    %U(\theta)-U_{NP}(\theta)=U(\underline{\theta})-U_{NP}(\underline{\theta})-\int_{\underline{\theta}}^{\theta}\int_0^px(s,b)g'(H_s(b))\frac{\partial H_s(b)}{\partial s}\,db\,ds.
    %\]
    
    %By assumption $U(\underline{\theta})-U_{NP}(\underline{\theta}) \geq 0$. Furthermore, $x(\theta,b) \geq 0$, $g$ is increasing and $\frac{\partial H_{\theta}(b)}{\partial \theta} \leq 0$. Hence: 
    %[\int_{\underline{\theta}}^{\bar{\theta}}\int_{0}^{p}x(s,b)g'(H_s(b))\frac{\partial H_s(b)}{\partial s}\,db\,ds \leq 0.
    %\]
    
    %Therefore, $U(\theta)-U_{NP}(\theta) \geq 0$.
%\end{proof}

\begin{proof}[Proof of the Necessity part of Lemma~\ref{lemma:feasible-prices}]
   Suppose that for some $x$ satisfying Lemma~\ref{lemma:IC2-contract}, Condition~\ref{condition:monotonicity-contract}, and some $p \in \mathbb{R}_+$, we have $\mathcal{D}(x;p) = S(p)$. 

    Since $\mathcal{D}(x;p) \leq \int_{\underline{\theta}}^{\bar{\theta}}[1-H_{\theta}(0)]f(\theta)\,d\theta$, $S\left(p^F\right) = \int_{\underline{\theta}}^{\bar{\theta}}[1-H_{\theta}(0)]f(\theta)\,d\theta$, and $S(p)$ is strictly increasing, it immediately follows that $p \leq p^F$.

    Recall that $\mathcal{D}(x;p) = S(p)$ can also be expressed as
    \[
        \int_{\underline{\theta}}^{\bar{\theta}} \int_{0}^{p} x(\theta,b) h_{\theta}(b) \, db \, f(\theta) \, d\theta = RS(p).
    \]
    The smallest achievable value on the left-hand side of the equation is zero. For the right-hand side, we know that $RS\left(p^N\right) = 0$, and our technical assumptions imply that $RS(p)$ is strictly increasing. Therefore, the equality cannot hold for any $p < p^N$. Thus, $p \geq p^N$.
\end{proof}

\begin{proof}[Proof of the Sufficiency part of Lemma~\ref{lemma:feasible-prices}]
    Suppose that $p \in \left(p^N, p^F\right)$, and for any $\alpha \in [0,1]$, define the following function:
    \[
        ED(\alpha;p) = \int_{\underline{\theta}}^{\bar{\theta}} \int_{0}^{p} \alpha h_{\theta}(b) \, db \, f(\theta) \, d\theta - RS(p).
    \]

    Note that $ED(\alpha; p)$ is continuous in $\alpha$ and satisfies $ED(0; p) < 0$ and $ED(1; p) > 0$. By the intermediate value theorem, there exists an $\alpha^* \in (0,1)$ such that $ED(\alpha^*; p) = 0$.

    Next, define $x(\theta, b) = \alpha^*$ for all $b \in [0, p]$. It follows that $x$ satisfies~\ref{IC2}, Condition~\ref{condition:monotonicity-contract}, and~\ref{EQ}.
\end{proof}

\section{Proofs of Results in Section~\ref{sec:insurer-problem}} \label{sec:proof-insurer-problem}
\begin{proof}[Proof of Lemma~\ref{lemma:objective}]
    If an agent has risk type $\theta$ and disease level $b$, the insurance provider pays $x(\theta,b)(p-D(\theta,b))$ to the hospital. Furthermore, we know that 
    \[
        x(\theta,b)D(\theta,b) = L(b;\theta) - (1 - x(\theta,b))b,
    \] 
    and by Lemma~\ref{lemma:IC2-contract}, we have:
    \[
    L(b;\theta) = \int_0^b (1 - x(\theta,l)) \, dl \; \text{for all } b \in [0,p], \; L(p;\theta) = L(b;\theta) \text{ for all } b > p, \; \text{and } x(\theta,b) = 1 \text{ for all } b > p.
    \]
    
    Thus, the insurer's cost of providing the contract for risk type $\theta$ simplifies as follows:
    \begin{align*}
        \int_0^{\bar{b}}x(\theta,b)(D(\theta,b)-p)h_{\theta}(b)\,db
        =&\int_{0}^{\bar{b}}x(\theta,b)(b-p)h_{\theta}(b)\,db+\int_0^{\bar{b}}(L(b;\theta)-b)h_{\theta}(b)\,db\\
        =&\int_{0}^{p}x(\theta,b)(b-p)h_{\theta}(b)\,db+\int_{p}^{\bar{b}}(b-p)h_{\theta}(b)\,db\\
        &-\int_0^{p}\left[\int_0^bx(\theta,l)\,dl\right]h_{\theta}(b)\,db+\int_{p}^{\bar{b}}(L(p;\theta)-b)h_{\theta}(b)\,db\\
        =&\int_{0}^{p}x(\theta,b)(b-p)h_{\theta}(b)\,db-p(1-H_{\theta}(p))\\
        &-\int_0^px(\theta,b)\,dbH_{\theta}(p)+\int_{0}^px(\theta,b)H_\theta(b)\,db+\int_p^{\bar{b}}L(p;\theta)h_{\theta}(b)\,db\\
        =&\int_{0}^{p}x(\theta,b)(b-p)h_{\theta}(b)\,db-p(1-H_{\theta}(p))\\
        &+(L(p;\theta)-p)H_{\theta}(p)+\int_{0}^px(\theta,b)H_\theta(b)\,db+L(p;\theta)(1-H_{\theta}(p))\\
        =&\int_{0}^{p}x(\theta,b)(b-p)h_{\theta}(b)\,db+\int_{0}^px(\theta,b)H_\theta(b)\,db+L(p;\theta)-p\\
        =&\int_{0}^{p}x(\theta,b)(b-p)h_{\theta}(b)\,db+\int_{0}^px(\theta,b)H_\theta(b)db-\int_0^px(\theta,b)\,db\\
        =&\int_{0}^{p}x(\theta,b)\left[b-p-\frac{1-H_\theta(b)}{h_{\theta}(b)}\right]h_{\theta}(b)\,db,
    \end{align*}
    where in the third equality we used integration by parts to express the term $\int_0^{p}\left[\int_0^bx(\theta,l)\,dl\right]h_{\theta}(b)\,db$ as:
    \[
        \int_0^px(\theta,b)\,dbH_{\theta}(p)-\int_0^px(\theta,b)H_{\theta}(b)\,db.
    \]
    
    Using Lemma~\ref{lemma:IC1-menu}, $t(\theta)$ can be expressed as:
    \begin{align*}
        t(\theta)=&-U(\underline{\theta})-\int_{0}^{p}(1-x(\theta,b))[1-g(H_{\theta}(b))]\,db-\int_{\underline{\theta}}^{\theta}\int_0^{p}(1-x(s,b))g'(H_s(b))\frac{\partial H_s(b)}{\partial s}\,db\,ds\\
        =&-U(\underline{\theta})-\int_{0}^{p}[1-g(H_{\underline{\theta}}(b))]\,db+\int_{0}^{p}x(\theta,b)[1-g(H_{\theta}(b))]\,db\\
        &+\int_{\underline{\theta}}^{\theta}\int_0^{p}x(s,b)g'(H_s(b))\frac{\partial H_s(b)}{\partial s}\,db\,ds.
    \end{align*}
    
    Integration by parts shows:
    \[
    \int_{\underline{\theta}}^{\bar{\theta}}\int_0^{p}\int_{\underline{\theta}}^{\theta}x(s,b)g'(H_s(b))\frac{\partial H_s(b)}{\partial s}\,ds\,dbf(\theta)\,d\theta=\int_{\underline{\theta}}^{\bar{\theta}}\int_0^{p}x(\theta,b)\frac{1-F(\theta)}{f(\theta)}g'(H_{\theta}(b))\frac{\partial H_{\theta}(b)}{\partial s}\,dbf(\theta)\,d\theta.
    \]
    
    Thus, we have that:
    \begin{align*}
        \int_{\underline{\theta}}^{\bar{\theta}}t(\theta)f(\theta)\,d\theta=&-U(\underline{\theta})-\int_{0}^{p}[1-g(H_{\underline{\theta}}(b))]\,db\\
        &+\int_{\underline{\theta}}^{\bar{\theta}}\int_0^{p}x(\theta,b)\left[1-g(H_{\theta}(b))
        +\frac{1-F(\theta)}{f(\theta)}g'(H_{\theta}(b))\frac{\partial H_{\theta}(b)}{\partial b}\right]\,dbf(\theta)\,d\theta.
    \end{align*}
    
    Then, the insurer's profits from offering a menu:
    \[
    \int_{\underline{\theta}}^{\bar{\theta}}\left[t(\theta)-\int_0^{\bar{b}}x(\theta,b)(p-D(\theta,b))\,dH_{\theta}(b)\right]f(\theta)\,d\theta,
    \]
    simplify to:
    \[
        \int_{\underline{\theta}}^{\bar{\theta}} \int_{0}^{p} x(\theta,b) J(\theta,b)h_{\theta}(b) \, db f(\theta) \, d\theta 
        - U(\underline{\theta}) - \int_{0}^{p} [1 - g(H_{\underline{\theta}}(b))] \, db,
    \]
    where:
    \[
        J(\theta,b) = \left[1 - g(H_{{\theta}}(b)) + \frac{1 - F(\theta)}{f(\theta)} g'(H_{\theta}(b)) \frac{\partial H_{\theta}(b)}{\partial \theta}\right]\frac{1}{h_{\theta}(b)} +b - p - \frac{1 - H_{\theta}(b)}{h_{\theta}(b)}.
    \]
\end{proof}

\section{Proofs of Results in Section~\ref{sec:optimal-transport}}\label{sec:proof-optimal-transport}
\begin{proof}[Proof of Lemma~\ref{lemma:optimal-menu}]
    Let $\pi'$ be any solution to the primal problem~\ref{OT}, whose existence is guaranteed by Theorem~\ref{thm:zero-duality}. Define $q_{\theta}$ as the expected value of $q$ under the conditional distribution $\pi'(\cdot \mid \theta)$, i.e., $q_{\theta} = \int_0^1 q \, d\pi^*(q \mid \theta)$.
    
    Since $q_{\theta}$ lies in the convex hull of $\supp \pi'(\cdot \mid \theta) \subset \mathbb{R}$, Carathéodory's theorem ensures the existence of a distribution $\pi''(\cdot \mid \theta)$ such that $\supp \pi''(\cdot \mid \theta)$ contains at most two elements, satisfies  $\supp \pi''(\cdot \mid \theta) \subseteq \supp \pi'(\cdot \mid \theta)$, and $q_{\theta} = \int_0^1 q \, d\pi''(q \mid \theta)$.
    
    Define $V$ as the set of risk types for which $\supp \pi''(\cdot \mid \theta)$ contains exactly two elements. That is, for any $\theta \in V$, we have $\supp \pi''(\cdot \mid \theta) = \{q_{1,\theta}, q_{2,\theta}\}$, where, without loss of generality, $q_{1,\theta} < q_{2,\theta}$, and:  
    \[
    q_{\theta} = q_{1,\theta} \pi''(q_{1,\theta} \mid \theta) + q_{2,\theta} (1 - \pi''(q_{1,\theta} \mid \theta)).
    \]  

    Next, define:
    \[
    q_V = \int_{V} [H_{\theta}(p) - H_{\theta}(0)] q_{\theta} f(\theta) \, d\theta.
    \]
    
    Since: 
    \[
    q_V \in \left( \int_{V} [H_{\theta}(p) - H_{\theta}(0)] q_{1,\theta} f(\theta) \, d\theta, \int_{V} [H_{\theta}(p) - H_{\theta}(0)] q_{2,\theta} f(\theta) \, d\theta \right),
    \]  
    there exists some $\alpha \in (0,1)$ such that: 
    \[
    q_V = \alpha \int_{V} [H_{\theta}(p) - H_{\theta}(0)] q_{1,\theta} f(\theta) \, d\theta + (1 - \alpha) \int_{V} [H_{\theta}(p) - H_{\theta}(0)] q_{2,\theta} f(\theta) \, d\theta.
    \]

    Define a new distribution $\pi^*$ as follows: For all $\theta \in V$, set:  
    \[
        \pi^*(\theta, q_{1,\theta}) = \alpha f(\theta), \quad \pi^*(\theta, q_{2,\theta}) = (1 - \alpha) f(\theta),
    \]
  and $\pi^*(\theta, q) = 0$ for all other $q$. For all $\theta \in V^c$, set  
  \[
  \pi^*(\theta, q_{\theta}) = f(\theta),
  \]
  and $\pi^*(\theta, q) = 0$ for all other $q$.  

    By construction, $\pi^*$ satisfies condition~\ref{BP}. Moreover,  
    \begin{align*}
    \int_{\underline{\theta}}^{\bar{\theta}} \int_0^1 [H_{\theta}(p) - H_{\theta}(0)] q \, d\pi^*(\theta,q)
    &= q_V + \int_{V^c} [H_{\theta}(p) - H_{\theta}(0)] q_{\theta} f(\theta) \, d\theta \\
    &= \int_{\underline{\theta}}^{\bar{\theta}} \int_0^1 [H_{\theta}(p) - H_{\theta}(0)] q_{\theta} f(\theta) \, d\theta \\
    &= \int_{\underline{\theta}}^{\bar{\theta}} \int_0^1 [H_{\theta}(p) - H_{\theta}(0)] q \, d\pi'(\theta,q) \\
    &= RS(p).
    \end{align*}  
    Thus, $\pi^*$ also satisfies condition~\ref{EQ}.  

    Finally, by Theorem~\ref{thm:zero-duality}, there exists $(v^*,\lambda^*)$ solving the dual problem~\ref{D}. By Corollary~\ref{coro:complementary-slackness}, for $f$-almost every $\theta$, the distribution $\pi'(\cdot \mid \theta)$ is supported on~\ref{supp}.
    Since $\supp \pi^*(\cdot \mid \theta) \subseteq \supp \pi'(\cdot \mid \theta)$, it follows that $f$-almost every $\theta$ also has $\pi^*(\cdot \mid \theta)$ supported~\ref{supp}. Using Corollary~\ref{coro:complementary-slackness} again, we conclude that $\pi^*$ is a solution to the primal problem.
\end{proof}

\begin{proof}[Proof of Lemma~\ref{lemma:simple-contract}]
    Fix a risk type $\theta$ and suppose that $\frac{\partial \phi(\theta, q)}{\partial q} \leq 0$ for all $q$. This assumption implies that the objective in~\ref{supp} is concave in $q$.  

    Let $\pi$ be any solution to the primal problem~\ref{OT}, whose existence is guaranteed by Theorem~\ref{thm:zero-duality}. By Corollary~\ref{coro:complementary-slackness}, we know that $\pi(\cdot \mid \theta)$ is supported on~\ref{supp}. Moreover, by the concavity of the objective in~\ref{supp}, it follows that the average $q$ induced by $\pi(\cdot \mid \theta)$ also belongs to~\ref{supp}. Consequently, for the measure $\pi^*$ constructed in the proof of Lemma~\ref{lemma:optimal-menu}, the support of $\pi^*(\cdot \mid \theta)$ contains only a single element.
\end{proof}

\begin{proof}[Proof of Lemma~\ref{lemma:implementable-menu}]
    Suppose that $\frac{\partial \phi(\theta,q)}{\partial \theta} > 0$ for all $q$, and let $\pi^*$ be the solution to the primal problem~\ref{OT}.  

    By Corollary~\ref{coro:complementary-slackness}, we know that $\pi^*(\cdot \mid \theta)$ is supported on~\ref{supp}. Furthermore, multiplying the objective on~\ref{supp} by $\frac{1}{H_{\theta}(p) - H_{\theta}(0)}$—which is independent of $q$—does not change the set of maximizers. Thus, $\pi^*(\cdot \mid \theta)$ is also supported on :
    \[
        \arg\max_{q \in [0,1]} \quad \frac{\Phi(\theta,q)}{H_{\theta}(p) - H_{\theta}(0)} + \lambda^*\left(q - \frac{RS(p)}{H_{\theta}(p) - H_{\theta}(0)}\right).
    \]  
    The derivative of this objective with respect to $q$ is:  
    \[
    \phi(\theta,q) + \lambda^*,
    \]  
    and its derivative with respect to $\theta$ is:  
    \[
        \frac{\partial \phi(\theta,q)}{\partial \theta},
    \]  
    which, by assumption, is positive. Therefore, the objective is supermodular, and by Topkis’s monotonicity theorem, the desired result follows.
\end{proof}

\section{Proof of Results in Section~\ref{sec:second-stage}}\label{sec:proof-second-stage}
\begin{proof}[Proof of Theorem~\ref{thm:envelope}]
To prove the result, we will show that our setting satisfies all the assumptions needed to apply Theorem 4 of \cite{milgrom2002envelope}. For this purpose, let $\mathcal{M}^+(\Theta \times [0,1])$ denote the set of positive measures over $\Theta \times [0,1]$, and let $L^1(f)$ denote the set of all measurable functions $v:\Theta \to \mathbb{R}$ that satisfy:
\[
\int_{\underline{\theta}}^{\bar{\theta}} \lvert v(\theta) \lvert f(\theta)\;d\theta < \infty.
\]

Then, define the function $\mathcal{L} : \mathcal{M}^+(\Theta \times [0,1]) \times (L^1(f) \times \mathbb{R}) \times [p^N,p^F] \to \mathbb{R}$ by:
\[
\mathcal{L}(\pi, (v,\lambda), p) = \int_{\underline{\theta}}^{\bar{\theta}} \int_0^1 \Phi(\theta,q)\, d\pi(\theta,q)
+ \int_{\underline{\theta}}^{\bar{\theta}} v(\theta) \Bigl[f(\theta)\, d\theta - \int_0^1 d\pi(\theta,q)\Bigr]
+ \lambda\, ED(\pi,p),
\]
where:
\[
ED(\pi,p) = \int_{\underline{\theta}}^{\bar{\theta}} \int_0^1 [H_{\theta}(p) - H_{\theta}(0)]\, q\, d\pi(\theta,q) - RS(p).
\]

Fix any $p \in [p^N,p^F]$ and let $\pi^*$ be a solution to the primal problem~\ref{OT} and $(v^*, \lambda^*)$ be a solution to the dual problem~\ref{D}. By Lemma~\ref{lemma:zero-duality}, $(\pi^*, (v^*, \lambda^*))$ is a saddle point of $\mathcal{L}$:
\[
\mathcal{L}(\pi, (v^*, \lambda^*), p) \leq \mathcal{L}(\pi^*, (v^*, \lambda^*), p)=\Pi(p) \leq \mathcal{L}(\pi^*, (v, \lambda), p),
\]
for all $\pi \in \mathcal{M}^+(\Theta \times [0,1])$ and $(v,\lambda) \in L^1(f) \times \mathbb{R}$.

It is easy to show that there exist finite positive measures $\overline{\pi}$ and $\underline{\pi}$ such that $ED(\overline{\pi},p) > 0$ and $ED(\underline{\pi},p) < 0$ for all $p \in [p^N,p^F]$. Since $(\pi^*, (v^*, \lambda^*))$ is a saddle point, it follows that:
\[
\Pi(p) 
\geq \mathcal{L}(\overline{\pi}, (v^*, \lambda^*), p) = \int_{\underline{\theta}}^{\bar{\theta}} \int_0^1 \Phi(\theta,q)\, d\overline{\pi}(\theta,q)
+ \int_{\underline{\theta}}^{\bar{\theta}} v(\theta) \Bigl[f(\theta)\, d\theta - \int_0^1 d\overline{\pi}(\theta,q)\Bigr]
+ \lambda^*\, ED(\overline{\pi},p).
\]

Hence:
\[
\lambda^* \leq \overline{\lambda} = \sup_{p \in [p^N,p^F]} \frac{\Pi(p)
- \int_{\underline{\theta}}^{\bar{\theta}} \int_0^1 \Phi(\theta,q)\, d\overline{\pi}(\theta,q)
- \int_{\underline{\theta}}^{\bar{\theta}} v^*(\theta) \Bigl[f(\theta)\, d\theta - \int_0^1 d\overline{\pi}(\theta,q)\Bigr]}
{ED(\overline{\pi},p)}.
\]

Since the numerator is bounded and the denominator is bounded away from zero by the definition of $\bar{\pi}$ and continuity of $ED(\overline{\pi}, \cdot)$, we conclude that $\lambda^*$ is bounded above by $\overline{\lambda} < \infty$. A similar argument with $\underline{\pi}$ shows that $\lambda^*$ is bounded below by some $\underline{\lambda} > -\infty$. Thus, we can assume $\lambda \in [\underline{\lambda}, \overline{\lambda}]$ and $\lvert \lambda \lvert \leq M$ for some $M< \infty$.

Additionally, since $\pi^* \in \Delta(\Theta \times [0,1])$, $(\pi^*, (v^*, \lambda^*))$ is a saddle point even when we restrict to probability measures. Therefore, from now on we restrict to $\pi \in \Delta(\Theta \times [0,1])$.

Next, by our technical assumptions, $\mathcal{L}(\pi, (v,\lambda), \cdot)$ is differentiable with derivative:
\[
\frac{\partial \mathcal{L}(\pi, (v,\lambda), p)}{\partial p} = \mathcal{L}_p(\pi, (v,\lambda), p)
= \int_{\underline{\theta}}^{\bar{\theta}} \int_0^1 \frac{\partial \Phi(\theta,q)}{\partial p}\, d\pi(\theta,q)
- \lambda \Bigl(S'(p) + \int_{\underline{\theta}}^{\bar{\theta}} \int_0^1 h_{\theta}(p) (1 - q)\, d\pi(\theta,q)\Bigr).
\]

Recall:
\[
\Phi(\theta,q) = [H_{\theta}(p) - H_{\theta}(0)] \int_0^q \phi(\theta,s)\, ds,
\]
where $\phi(\theta,q) = J(\theta, b_{\theta}(q))$ and:
\[
b_{\theta}(q) = H_{\theta}^{-1}\Bigl((1 - q)[H_{\theta}(p) - H_{\theta}(0)] + H_{\theta}(0)\Bigr).
\]

A careful calculation yields:
\begin{align*}
    \frac{\partial \Phi(\theta,q)}{\partial p}&= h_{\theta}(p) \int_0^q J(\theta, b_{\theta}(s))\, ds+[H_{\theta}(p)-H_{\theta}(0)]\int_0^q\frac{\partial J(\theta,b_{\theta}(s))}{\partial b}\frac{\partial b_{\theta}(s)}{\partial p}\,ds - [H_{\theta}(p) - H_{\theta}(0)] q \\
    &= h_{\theta}(p) \int_0^q J(\theta, b_{\theta}(s))\, ds-h_{\theta}(p)\int_0^q(1-s)\frac{\partial J(\theta,b_{\theta}(s))}{\partial b}\frac{\partial b_{\theta}(s)}{\partial s}\,ds - [H_{\theta}(p) - H_{\theta}(0)] q\\
    &=h_{\theta}(p)[J(\theta,p)-(1-q)J(\theta,q)]- [H_{\theta}(p) - H_{\theta}(0)] q.
\end{align*}

Then:
\[
\mathcal{L}_p(\pi, (v,\lambda), p) = \int_{\underline{\theta}}^{\bar{\theta}} \int_0^1 T(\theta,q,p)\, d\pi(\theta,q)
- \lambda \int_{\underline{\theta}}^{\bar{\theta}} \int_0^1 Y(\theta,q,p)\, d\pi(\theta,q),
\]
where:
\[
T(\theta,q,p) = h_{\theta}(p)\Bigl[J(\theta,p) - (1 - q) J(\theta,b_{\theta}(q))\Bigr] - [H_{\theta}(p) - H_{\theta}(0)] q,
\]
\[
Y(\theta,q,p) = S'(p) + h_{\theta}(p)(1 - q).
\]

By our technical assumptions, $T$ and $Y$ are continuous functions. Moreover, since both functions are defined on the compact domain $\Theta \times [0,1] \times \left[p^N, p^F\right]$, they are uniformly continuous. Hence, for any $\varepsilon > 0$, there exists $\delta > 0$ such that:
\[
\int_{\underline{\theta}}^{\bar{\theta}} \int_0^1 |T(\theta,q,p) - T(\theta,q,p')|\, d\pi(\theta,q) \leq \sup_{(\theta,q) \in \Theta \times [0,1]} \lvert T(\theta,q,p)-T(\theta,q,p')\lvert < \varepsilon,
\]
\[
|\lambda| \int_{\underline{\theta}}^{\bar{\theta}} \int_0^1 |Y(\theta,q,p) - Y(\theta,q,p')|\, d\pi(\theta,q) \leq \sup_{(\theta,q) \in \Theta \times [0,1]} M\lvert Y(\theta,q,p)-Y(\theta,q,p')\lvert) < \varepsilon,
\]
for all $|p - p'| < \delta$. Therefore, $\mathcal{L}_p(\pi, (v,\lambda), \cdot)$ is uniformly continuous and the family $\{\mathcal{L}_p(\pi, (v,\lambda), \cdot)\}_{(\pi,(v,\lambda))}$ is uniformly equicontinuous. Hence, $\{\mathcal{L}(\pi, (v,\lambda), \cdot)\}_{(\pi,(v,\lambda))}$ is equidifferentiable.

Since $\Delta(\Theta \times [0,1])$ and $L^1(f) \times [\underline{\lambda}, \overline{\lambda}]$ are second-countable topological spaces, all assumptions of Theorem 4 of \cite{milgrom2002envelope} are satisfied. Thus:
\[
\frac{\partial \Pi(p)}{\partial p} = \Pi'(p) = \mathcal{L}_p(\pi^*, (v^*, \lambda^*), p).
\]

To conclude the proof, we simplify $\mathcal{L}_p$ and express it in terms of $x^*$. Since $\pi^*$ satisfies~\ref{BP}, $\pi^*(\cdot \mid \theta) \equiv \frac{\pi^*(\theta,\cdot)}{f(\theta)}:[0,1] \to [0,1]$ is a well-defined cumulative distribution function. Moreover, by Corollary~\ref{coro:complementary-slackness}, any $q \in (0,1)$ in the support of $\pi^*(\cdot \mid \theta)$ must satisfy the first-order condition:
\[
J(\theta, b_{\theta}(q)) + \lambda^* = 0.
\]

Then:
\begin{align*}
\mathcal{L}_p(\pi^*, (v^*, \lambda^*), p) =& \int_{\underline{\theta}}^{\bar{\theta}}\left[(1 - \pi^*(0 \mid \theta)) J(\theta,p) h_{\theta}(p)-[H_{\theta}(p) - H_{\theta}(0)]\int_0^1 (1 - \pi^*(q \mid \theta))\, dq\right]\, f(\theta)\, d\theta\\
&- \lambda^* \Bigl(S'(p) + \int_{\underline{\theta}}^{\bar{\theta}} h_{\theta}(p) \pi^*(0 \mid \theta) f(\theta)\, d\theta\Bigr).
\end{align*}

Recall:
\[
x^*(\theta, b_{\theta}(q)) = 1 - \pi^*(q \mid \theta).
\]

Then, changing variables $b = b_{\theta}(q)$, we get:
\begin{align*}
\mathcal{L}_p(\pi^*, (v^*, \lambda^*), p) =& \int_{\underline{\theta}}^{\bar{\theta}} x^*(\theta,p) J(\theta,p) h_{\theta}(p) f(\theta)\, d\theta
- \int_{\underline{\theta}}^{\bar{\theta}} \int_0^p x^*(\theta,b) h_{\theta}(b)\, db\, f(\theta)\, d\theta\\
&- \lambda^* \Bigl(S'(p) + \int_{\underline{\theta}}^{\bar{\theta}} h_{\theta}(p) [1 - x^*(\theta,p)] f(\theta)\, d\theta\Bigr).
\end{align*}
\end{proof}

\begin{proof}[Proof of Corollary~\ref{coro:endo-exo}]
Let $\beta>\beta'$ and $p_{\beta}$ be a solution to the second-stage problem~\ref{P2} under $\beta$, and let $p_{\beta'}$ denote the highest price that solves~\ref{P2} under $\beta'$. By Corollary~\ref{coro:maximum-price}, $p_{\beta}$ exists. Moreover, since the set of maximizers in the second-stage problem~\ref{P2} is nonempty and compact (also by Corollary~\ref{coro:maximum-price}), $p_{\beta'}$ is well defined.

Suppose, for the sake of contradiction, that $p_{\beta} > p_{\beta'}$. Let $\Pi_{\beta}$ and $\Pi_{\beta'}$ denote the insurer's profit functions under $\beta$ and $\beta'$, respectively. Then,
\[
\Pi_{\beta}(p_{\beta}) - \Pi_{\beta}(p_{\beta'}) = \int_{p_{\beta'}}^{p_{\beta}} \Pi_{\beta}'(p)\,dp \geq 0,
\]
since $p_{\beta}$ is a maximizer of $\Pi_{\beta}$.

However, by Theorem~\ref{thm:envelope}, we know that $\Pi_{\beta}'(p) \leq \Pi'_{\beta'}(p)$ for all $p$. This inequality holds because the only difference between the two derivatives is:
\[
-g(H_{\underline{\theta}}(p+\beta) \leq -g(H_{\underline{\theta}}(p+\beta').
\]

Therefore, we obtain:
\[
\Pi_{\beta'}(p_{\beta}) - \Pi_{\beta'}(p_{\beta'}) = \int_{p_{\beta'}}^{p_{\beta}} \Pi_{\beta'}'(p)\,dp \geq \int_{p_{\beta'}}^{p_{\beta}} \Pi_{\beta}'(p)\,dp \geq 0,
\]
which contradicts the fact that $p_{\beta'}$ is the highest maximizer of $\Pi_{\beta'}$. 
\end{proof}

%\begin{proof}[Proof of Corollary~\ref{coro:endo-exo}]
%Let $p_{\text{exo}}$ be a solution to the second-stage problem~\ref{P2} in the exogenous case, and let $p_{\text{endo}}$ denote the highest price that solves~\ref{P2} in the endogenous case. By Corollary~\ref{coro:maximum-price}, $p_{\text{exo}}$ exists. Moreover, since the set of maximizers in the endogenous case is nonempty and compact (also by Corollary~\ref{coro:maximum-price}), $p_{\text{endo}}$ is well defined.

%Suppose, for the sake of contradiction, that $p_{\text{exo}} > p_{\text{endo}}$. Let $\Pi_{\text{exo}}$ and $\Pi_{\text{endo}}$ denote the insurer's profit functions in the exogenous and endogenous cases, respectively. Then,
%\[
%\Pi_{\text{exo}}(p_{\text{exo}}) - \Pi_{\text{exo}}(p_{\text{endo}}) = \int_{p_{\text{endo}}}^{p_{\text{exo}}} \Pi_{\text{exo}}'(p)\,dp \geq 0,
%\]
%since $p_{\text{exo}}$ is a maximizer of $\Pi_{\text{exo}}$.

%However, by Theorem~\ref{thm:envelope}, we know that $\Pi_{\text{exo}}'(p) \leq \Pi_{\text{endo}}'(p)$ for all $p$. This inequality holds because the only difference between the two derivatives is the term $K'(p)$, which is less than or equal to zero in the exogenous case.

%Therefore, we obtain:
%\[
%\Pi_{\text{endo}}(p_{\text{exo}}) - \Pi_{\text{endo}}(p_{\text{endo}}) = \int_{p_{\text{endo}}}^{p_{\text{exo}}} \Pi_{\text{endo}}'(p)\,dp \geq \int_{p_{\text{endo}}}^{p_{\text{exo}}} \Pi_{\text{exo}}'(p)\,dp \geq 0,
%\]
%which contradicts the fact that $p_{\text{endo}}$ is the highest maximizer of $\Pi_{\text{endo}}$. 
%\end{proof}

\begin{proof}[Proof of Proposition~\ref{prop:comparative-statics}]
    We begin by recalling key properties that any optimal multiplier $\lambda^{*}(p)$ and optimal deductible $D(\theta;p)$ must satisfy. Define:
    \[
        \hat{J}(\theta,b) = J(\theta,b) + p,
    \]
    which is independent of $p$; since $J(\theta,\cdot)$ is strictly increasing, so is $\hat{J}(\theta,\cdot)$.

    Let:
    \[
        q(\theta;p) = 1 - H_{\theta}(D(\theta;p) \mid p).
    \]
    By Complementary Slackness (Corollary~\ref{coro:complementary-slackness}), every optimal $q(\theta;p)$ lies in the set~\ref{supp} and must satisfy:
   \begin{itemize}
        \item \textbf{Interior case} ($q(\theta;p) \in (0,1)$, i.e., $D(\theta;p) \in (0,p)$):
            \[
                \phi(\theta,q(\theta;p)) = J(\theta,b_{\theta}(q(\theta;p))) = \hat{J}(\theta,D(\theta;p)) - p = -\lambda^{*}(p).
            \]
        \item \textbf{Upper boundary} ($q(\theta;p) = 0$, i.e., $D(\theta;p) = p$):
            \[
                \hat{J}(\theta,p) \le -\lambda^{*}(p) + p.
            \]
        \item \textbf{Lower boundary} ($q(\theta;p) = 1$, i.e., $D(\theta;p) = 0$):
            \[
                \hat{J}(\theta,0) \ge -\lambda^{*}(p) + p.
            \]
    \end{itemize}
    Since $\hat{J}(\theta,\cdot)$ is strictly increasing and continuous, it is invertible. Therefore, we can define the function:
    \[
    D\left(\theta;\left(p,\hat{\lambda}\right)\right)=\min\left\{p,\max\left\{0,\hat{J}^{-1}\left(\theta,-\hat{\lambda}\right)\right\}\right\}.
    \]
    
    Then, from the first-order conditions stated above, it is easy to see that the optimal deductible for risk type $\theta$ is given by:
        \[
            D(\theta;p) =D\left(\theta;\left(p,\hat{\lambda}^*(p)\right)\right),
        \]
    where $\hat{\lambda}^{*}(p) = \lambda^{*}(p) - p$.
    
    \medskip
    
    \noindent\textbf{Step 1: Monotonicity of the adjusted multiplier.}
    We will prove that $\hat{\lambda}^{*}(\cdot)$ is strictly increasing. Assume, for contradiction, that $\hat{\lambda}^{*}(p') \le \hat{\lambda}^{*}(p)$ for some $p' > p$.
    
    \begin{enumerate}
    \item \textbf{Deductibles must be weakly larger at the higher price.}
   This immediately follows from the fact that $D\left(\theta; \left(p, \cdot \right)\right)$ is non-increasing, since the inverse function $\hat{J}^{-1}(\theta, \cdot)$ is strictly increasing, and from the fact that $D\left(\theta; \left(\cdot, \hat{\lambda} \right)\right)$ is non-decreasing. Thus:
    \[
    D(\theta;p)=D\left(\theta;\left(p,\hat{\lambda}^*(p)\right)\right) \leq D\left(\theta;\left(p',\hat{\lambda}^*(p')\right)\right)=D(\theta;p').
    \]
    \item \textbf{Demand–supply contradiction.}

    Since $D(\theta;p') \ge D(\theta;p)$, we have:
    \[
        \mathcal{D}(p') = \int_{\underline{\theta}}^{\bar{\theta}} \left[1 - H_{\theta}(D(\theta;p'))\right]f(\theta)\,d\theta \le \int_{\underline{\theta}}^{\bar{\theta}} \left[1 - H_{\theta}(D(\theta;p))\right]f(\theta)\,d\theta = \mathcal{D}(p).
    \]
    But supply is strictly increasing, so $S(p') > S(p)$. Since the market clears at $p$ (i.e., $\mathcal{D}(p) = S(p)$), we get:
    \[
        S(p') > \mathcal{D}(p'),
    \]
    which contradicts market clearing.
    \end{enumerate}
    
    Thus, $\hat{\lambda}^{*}(p') > \hat{\lambda}^{*}(p)$. Since $\lambda^{*}(p') = \hat{\lambda}^{*}(p') + p'$, we conclude:
    \[
        \lambda^{*}(p') - \lambda^{*}(p) = \hat{\lambda}^{*}(p') - \hat{\lambda}^{*}(p) + (p' - p) > p' - p > 0,
    \]
    so $\lambda^{*}(\cdot)$ is also strictly increasing.

    \medskip

    \noindent\textbf{Step 2: Monotonicity of the deductible.}
    \begin{itemize}
        \item \textbf{Interior case} ($0 < D(\theta;p) < p$):
        In this case:
        \[
        D(\theta;p)=D\left(\theta;\left(p,\hat{\lambda}^*(p)\right)\right)=\hat{J}^{-1}\left(\theta,-\hat{\lambda}^*(p)\right).
        \]
        
        Since $\hat{J}^{-1}(\theta, \cdot)$ is strictly increasing, it follows that:
        \[
        \hat{J}^{-1}\left(\theta,-\hat{\lambda}^*(p')\right)<\hat{J}^{-1}\left(\theta,-\hat{\lambda}^*(p)\right),
        \]
        which leads to:
        \[
        D(\theta;p')=D\left(\theta;\left(p,\hat{\lambda}^*(p')\right)\right)=\hat{J}^{-1}\left(\theta,-\hat{\lambda}^*(p')\right)<D(\theta;p).
        \]

        \item \textbf{Boundary case} ($D(\theta;p) = 0$):
        In this case:
        \[
           \hat{J}^{-1}\left(\theta,-\hat{\lambda}^*(p')\right)<\hat{J}^{-1}\left(\theta,-\hat{\lambda}^*(p)\right)\leq0,
        \]
        which allows us to conclude that $D(\theta;p') = 0$.
    \end{itemize}
    
    Hence, $D(\theta;\cdot)$ is strictly decreasing whenever $D(\theta;p) \in (0,p)$, and weakly decreasing when $D(\theta;p) = 0$.
        
    \medskip

    \noindent\textbf{Step 3: Uniqueness of the multiplier.}
    Define the following demand function, which depends not only on $p$ but also on $\hat{\lambda}$:
    \[
        \mathcal{D}\left(p,\hat{\lambda}\right) = \int_{\underline{\theta}}^{\bar{\theta}} \left[1 - H_{\theta}\left(D\left(\theta; (p, \hat{\lambda})\right)\right)\right] f(\theta)\, d\theta.
    \]
    Since $H_{\theta}(\cdot)$ is strictly increasing, it follows that $\mathcal{D}(p, \cdot)$ is non-decreasing. Moreover, whenever $\hat{\lambda} > \hat{\lambda}'$ and there exists a set of types with positive measure such that $D\left(\theta; (p, \hat{\lambda})\right) < D\left(\theta; (p, \hat{\lambda}')\right)$, we have $\mathcal{D}\left(p, \hat{\lambda}\right) > \mathcal{D}\left(p, \hat{\lambda}'\right)$.

    Now fix $p \in (p^N, p^F)$ and suppose that $\hat{\lambda} > \hat{\lambda}'$. Since $\hat{J}^{-1}(\cdot, -\hat{\lambda})$ is continuous, two cases arise: Either, for all $\theta$, $D\left(\theta; (p, \hat{\lambda})\right) = D\left(\theta; (p, \hat{\lambda}')\right) = 0$ (or equal to $p$), in which case $\mathcal{D}(p, \hat{\lambda}) \ne S(p)$, so the market does not clear. Or there exists a positive mass of types for which $D\left(\theta; (p, \hat{\lambda})\right) < D\left(\theta; (p, \hat{\lambda}')\right)$, in which case $\mathcal{D}(p, \hat{\lambda}) > \mathcal{D}(p, \hat{\lambda}')$.

    Together with the continuity of $\mathcal{D}(p, \cdot)$, this implies that if there exists a multiplier $\hat{\lambda}$ such that the market clears at price $p$, then this multiplier must be unique. Since we know that at $\hat{\lambda}^*(p)$, the market clears, i.e., $\mathcal{D}\left(p, \hat{\lambda}^*(p)\right) = S(p)$, we conclude that $\hat{\lambda}^*(p)$ is unique.

    \medskip

    \noindent\textbf{Step 4: Continuity of the multiplier.}
    Let $p_0 \in \left(p^N, p^F\right)$, and suppose, towards a contradiction, that $\hat{\lambda}^*$ is discontinuous from the left at $p_0$:
    \[
        \hat{\lambda}^-(p_0) := \lim_{p \to p_0^-} \hat{\lambda}^*(p) < \hat{\lambda}^*(p_0).
    \]

    Since the function $\mathcal{D}$ (as defined in Step 3) is continuous in all its arguments, and $S$ is also continuous, it follows that:
    \[
        \lim_{p \to p_0^-} \left[ \mathcal{D}(p, \hat{\lambda}^*(p)) - S(p) \right] = \mathcal{D}\left(p_0, \hat{\lambda}^-(p_0)\right) - S(p_0).
    \]

    But since $\mathcal{D}(p, \hat{\lambda}^*(p)) = S(p)$ for all $p$, the left-hand side is zero. Therefore,
    \[
    \mathcal{D}(p_0, \hat{\lambda}^-(p_0)) =    S(p_0),
    \]
    which contradicts the uniqueness of $\hat{\lambda}^*(p_0)$ established in Step 3.

    A similar argument applies to rule out discontinuity from the right.

    Finally, since $\hat{\lambda}(\cdot)$ is continuous, it follows that $\lambda(\cdot)$ is also continuous.
    \medskip

    \noindent\textbf{Step 5: Continuity of the deductibles.}
    The continuity of $D(\theta; \cdot)$ follows immediately from the fact that $D(\theta; (p, \cdot))$ is continuous and that $\hat{\lambda}(\cdot)$ is continuous.
\end{proof}

\begin{proof}[Proof of Lemma~\ref{lemma:price-taking}]
    We have three possibilities:
    \begin{enumerate}
        \item $\lambda^*(p) \geq 0$ for all $p$. In this case, since $\lambda^*(\cdot)$ is strictly increasing (see Proposition~\ref{prop:comparative-statics}), it follows that $\lambda^*(p) > 0$ for all $p > p^N$, and the unique price-taking benchmark is $p^N$. Indeed, at $p^N$, the optimal deductible for all types is $D(\theta; p^N) = p^N$. Then, by the first-order conditions (see the proof of Proposition~\ref{prop:comparative-statics}), we have:
        \[
            \hat{J}(\theta, p^N) \leq -\lambda^*(p^N) + p^N \leq 0 + p^N,
        \]
        so $\lambda^*(p^N) = 0$ is also an optimal multiplier for price $p^N$.

        \item $\lambda^*(p) \leq 0$ for all $p$. In this case, $\lambda^*(p) < 0$ for all $p < p^F$, and the unique price-taking benchmark is $p^F$. Indeed, at $p^F$, the optimal deductible for all types is $D(\theta; p^F) = 0$. Then, by the first-order conditions, we have:
        \[
            \hat{J}(\theta, 0) \geq -\lambda^*(p^F) + p^F \geq 0 + p^F,
        \]
        it follows that $\lambda^*(p^F) = 0$ is also an optimal multiplier for price $p^F$.

        \item $\lambda^*(p^N) < 0$ and $\lambda^*(p^F) > 0$. In this case, since $\lambda^*(\cdot)$ is strictly increasing and continuous (see Proposition~\ref{prop:comparative-statics}), the Intermediate Value Theorem implies that there exists a unique $p^U \in (p^N, p^F)$ such that $\lambda^*(p^U) = 0$.
\end{enumerate}
\end{proof}

\begin{proof}[Proof of Proposition~\ref{prop:price-taking}]
    The desired result follows immediately by showing that for all $p < (>)\ p^U$, we have $\Pi'\left(p^U\right) < (>)\ \Pi'(p)$.

    Suppose $p < p^U$. Then, by Proposition~\ref{prop:price-taking}, we know that $\lambda^*(p) < \lambda^*\left(p^U\right)$ and $\hat{\lambda}^*(p) < \hat{\lambda}^*\left(p^U\right)$, where $\hat{\lambda}^*(p) = \lambda^*(p) - p$. Moreover, $D(\theta; p) > D\left(\theta; p^U\right)$ whenever $D(\theta; p) \in (0, p)$, and $D(\theta; p) = D\left(\theta; p^U\right)$ whenever $D(\theta; p) = 0$.

    This leads to the following cases:
    \begin{enumerate}
        \item $D(\theta;p) \in [0,p)$ and $D(\theta;p) \geq D\left(\theta;p^U\right)$:
        \begin{align*}
            \mathbf{1}_{\left[D(\theta;p)<p\right]}\left[J_1(\theta,p)-(1-H_{\theta}(D(\theta;p))\right]&=J_1(\theta,p)-(1-H_{\theta}(D(\theta;p)) \\
            &\geq J_1\left(\theta,p^U\right)-\left(1-H_{\theta}\left(D\left(\theta,p^U\right)\right)\right) \\
            &=\mathbf{1}_{\left[D\left(\theta;p^U\right)<p^U\right]}\left[J_1\left(\theta,p^U\right)-\left(1-H_{\theta}\left(D\left(\theta,p^U\right)\right)\right)\right],
        \end{align*}
        where the inequality follows from the fact that $J_1(\theta,\cdot)$ is non-increasing and that $H_{\theta}(D(\theta;p)) \geq H_{\theta}\left(D(\theta; p^U)\right)$.
        \item  $D(\theta;p)=p$ and $D\left(\theta;p^U\right)=p^U$:
        \begin{align*}
            \mathbf{1}_{\left[D(\theta;p)=p\right]}[-\lambda^*(p)h_{\theta}(p)]&=-\lambda^*(p)h_{\theta}(p) \\
            &>0 \\
            &=-\lambda^*\left(p^U\right)h_{\theta}\left(p^U\right) \\
            &=\mathbf{1}_{\left[D\left(\theta;p^U\right)=p^U\right]}\left[\lambda^*\left(p^U\right)h_{\theta}\left(p^U\right)\right],
        \end{align*}
        where the inequality follows from the fact that $\lambda^*\left(p^U\right) = 0 > \lambda^*(p)$.
        \item $D(\theta;p)=p$ and $D\left(\theta;p^U\right)<p$:
        \begin{align*}
            \mathbf{1}_{\left[D(\theta;p)=p\right]}[-\lambda^*(p)h_{\theta}(p)]&=-\lambda^*(p)h_{\theta}(p) \\
            &\geq J_1(\theta,p)-(1-H_{\theta}(p)) \\
            &\geq J_1\left(\theta,p^U\right)-\left(1-H_{\theta}\left(D\left(\theta,p^U\right)\right)\right) \\
            &=\mathbf{1}_{\left[D\left(\theta;p^U\right)<p^U\right]}\left[J_1\left(\theta,p^U\right)-\left(1-H_{\theta}\left(D\left(\theta,p^U\right)\right)\right)\right],
        \end{align*}
        where the first inequality follows from the fact that $D(\theta; p) = p$ must satisfy the first-order condition:
        \[
            J(\theta, p) = \frac{J_1(\theta, p) - \left(1 - H_{\theta}(p)\right)}{h_{\theta}(p)} \leq -\lambda^*(p),
        \]
        and the second inequality follows from the fact that $J_1(\theta; \cdot)$ is non-increasing and $H_{\theta}(p) \geq H_{\theta}\left(D\left(\theta, p^U\right)\right)$.
        \item $D(\theta;p)=p$ and $D\left(\theta;p^U\right) \in \left(p,p^U\right)$:
            \begin{align*}
            \mathbf{1}_{\left[D(\theta;p)=p\right]}[-\lambda^*(p)h_{\theta}(p)]&=-\lambda^*(p)h_{\theta}(p) \\
            &\geq -\lambda^*(p)h_{\theta}\left(D\left(\theta,p^U\right)\right) \\
            &\geq -\lambda^*(p)h_{\theta}\left(D\left(\theta,p^U\right)\right)-\left(D\left(\theta,p^U\right)-p\right)h_{\theta}\left(D\left(\theta,p^U\right)\right)\\
            &=\left[-\hat{\lambda}^*(p)-D\left(\theta,p^U\right)\right]h_{\theta}\left(D\left(\theta,p^U\right)\right)\\
            & \geq \left[-\hat{\lambda}^*\left(p^U\right)-D\left(\theta,p^U\right)\right]h_{\theta}\left(D\left(\theta,p^U\right)\right)\\
            &= J_1\left(\theta,p^U\right)-\left(1-H_{\theta}\left(D\left(\theta,p^U\right)\right)\right) \\
            &=\mathbf{1}_{\left[D\left(\theta;p^U\right)<p^U\right]}\left[J_1\left(\theta,p^U\right)-\left(1-H_{\theta}\left(D\left(\theta,p^U\right)\right)\right)\right],
        \end{align*}
        where the first inequality follows from the fact that $-\lambda^*(p) > 0$ and that $h_{\theta}(\cdot)$ is non-decreasing. The third inequality follows from the fact that $\hat{\lambda}^*(\cdot)$ is strictly increasing. Finally, the second-to-last equality follows from the first-order condition $J\left(\theta, p^U\right) = -\lambda^*\left(p^U\right)$.
    \end{enumerate}
    Finally, since $S'(p) > 0$ and $-\lambda^*(p) > -\lambda^*\left(p^U\right) = 0$, it follows that:
    \[
        -\lambda^*(p) S'(p) > -\lambda^*\left(p^U\right) S'\left(p^U\right).
    \]
    This establishes that $\Pi'(p) > \Pi'\left(p^U\right)$. An analogous argument shows that if $p > p^U$, then $\Pi'\left(p^U\right) > \Pi'(p)$.
\end{proof}

\begin{proof}[Proof of Proposition~\ref{prop:concave-profits}]
    First, since $H_{\theta}(\cdot)$ and $S$ are affine functions, we have $S'(p) = \beta$ for all $p$ and $h_{\theta}(b) = \beta(\theta)$ for all $b$, where $\beta > 0$ is a constant and $\beta(\theta) > 0$ is a strictly positive function for all $\theta$.

    Suppose that $p < p'$. Then, by Proposition~\ref{prop:price-taking}, we know that $\lambda^*(p) < \lambda^*(p')$ and $\hat{\lambda}^*(p) < \hat{\lambda}^*(p')$, where $\hat{\lambda}^*(p) = \lambda^*(p) - p$. Moreover, $D(\theta; p) > D(\theta; p')$ whenever $D(\theta; p) \in (0, p)$, and $D(\theta; p) = D(\theta; p')$ whenever $D(\theta; p) = 0$.

    This leads to the following cases:
    \begin{enumerate}
        \item $D(\theta;p) \in [0,p)$ and $D(\theta;p) \geq D(\theta;p')$:
        \begin{align*}
            \mathbf{1}_{\left[D(\theta;p)<p\right]}\left[J_1(\theta,p)-(1-H_{\theta}(D(\theta;p))\right]&=J_1(\theta,p)-(1-H_{\theta}(D(\theta;p)) \\
            &\geq J_1(\theta,p')-(1-H_{\theta}(D(\theta,p'))) \\
            &=\mathbf{1}_{\left[D(\theta;p')<p^U\right]}[J_1(\theta,p')-(1-H_{\theta}(D(\theta,p')))],
        \end{align*}
        where the inequality follows from the fact that $J_1(\theta,\cdot)$ is non-increasing and that $H_{\theta}(D(\theta;p)) \geq H_{\theta}(D(\theta; p'))$.
        \item  $D(\theta;p)=p$ and $D(\theta;p')=p'$:
        \begin{align*}
            \mathbf{1}_{\left[D(\theta;p)=p\right]}[-\lambda^*(p)h_{\theta}(p)]&=-\lambda^*(p)\beta(\theta) \\
            &>-\lambda^*(p')\beta(\theta) \\
            &=\mathbf{1}_{\left[D(\theta;p')=p'\right]}\left[\lambda^*(p')h_{\theta}(p')\right],
        \end{align*}
        where the inequality follows from the fact that $\lambda^*(p') > \lambda^*(p)$ and $\beta(\theta)>0$.
        \item $D(\theta;p)=p$ and $D(\theta;p')<p$:
        \begin{align*}
            \mathbf{1}_{\left[D(\theta;p)=p\right]}[-\lambda^*(p)h_{\theta}(p)]&=-\lambda^*(p)h_{\theta}(p) \\
            &\geq J_1(\theta,p)-(1-H_{\theta}(p)) \\
            &\geq J_1(\theta,p')-(1-H_{\theta}(D(\theta,p'))) \\
            &=\mathbf{1}_{\left[D(\theta;p')<p'\right]}[J_1(\theta,p')-(1-H_{\theta}D(\theta,p')))],
        \end{align*}
        where the first inequality follows from the fact that $D(\theta; p) = p$ must satisfy the first-order condition:
        \[
            J(\theta, p) = \frac{J_1(\theta, p) - \left(1 - H_{\theta}(p)\right)}{h_{\theta}(p)} \leq -\lambda^*(p),
        \]
        and the second inequality follows from the fact that $J_1(\theta; \cdot)$ is non-increasing and $H_{\theta}(p) \geq H_{\theta}(D(\theta, p'))$.
        \item $D(\theta;p)=p$ and $D(\theta;p') \in (p,p')$:
        \begin{align*}
            \mathbf{1}_{\left[D(\theta;p)=p\right]}[-\lambda^*(p)h_{\theta}(p)]&=-\lambda^*(p)\beta(\theta) \\
            &> \left[-\hat{\lambda}^*(p)-D(\theta,p')\right]\beta(\theta)\\
            &> \left[-\hat{\lambda}^*(p')-D(\theta,p')\right]\beta(\theta)\\
            &= J_1(\theta,p')-(1-H_{\theta}(D(\theta,p'))) \\
            &=\mathbf{1}_{\left[D(\theta;p')<p'\right]}\left[J_1(\theta,p')-(1-H_{\theta}(D(\theta,p')))\right],
        \end{align*}
        where the first inequality follows from the fact that:
        \[
            -\hat{\lambda}^*(p) -D(\theta,p')=-\lambda^*(p)-(D(\theta,p')-p)<-\lambda^*(p).
            \]
        The third inequality follows from the fact that $\hat{\lambda}^*(\cdot)$ is strictly increasing. Finally, the second-to-last equality follows from the first-order condition $J(\theta, p') = -\lambda^*(p')$.
    \end{enumerate}
    Finally, since $S'(p)=\beta>0$ and $-\lambda^*(p) > -\lambda^*(p')$, it follows that:
    \[
        -\lambda^*(p)\beta > -\lambda^*(p')\beta.
    \]
    This establishes that $\Pi'(p) > \Pi'(p')$. 
\end{proof}

\section{Technical Proofs in Section \ref{sec:optimal-transport}}\label{sec:proof-technical}
The proof of Theorem~\ref{thm:zero-duality} is a consequence of the following three lemmas:

\begin{lemma}\label{lemma:zero-duality}
    The value of the primal problem~\ref{OT} equals the value of the dual problem~\ref{D}.
\end{lemma}

\begin{proof}
    Let $V$ be given by:
    \[
        V=\sup_{\pi \in \mathcal{M}^+(\Theta \times [0,1])} \quad  \int_{\underline{\theta}}^{\bar{\theta}}\int_0^{1} \Phi(\theta,q)\,d\pi(\theta,q) + A(\pi), 
    \]
    where $\mathcal{M}^+(\Theta \times [0,1])$ is the set of positive measures over $\Theta \times [0,1]$, and $A$ is given by:
    \[
    A(\pi)=\inf_{(v,\lambda) \in L^1(f) \times \mathbb{R}} \quad \int_{\underline{\theta}}^{\bar{\theta}}v(\theta)\left[f(\theta)\,d\theta-\int_0^1\,d\pi(\theta,q)\right]+\int_{\underline{\theta}}^{\bar{\theta}}\int_0^1\lambda\left([H_{\theta}(p)-H_{\theta}(0)]q-RS(p)\right)\,d\pi(\theta,q),
    \]
    where $L^1(f)$ denotes the set of all measurable functions $v:\Theta \to \mathbb{R}$ that satisfy:
    \[
        \int_{\underline{\theta}}^{\bar{\theta}} \lvert v(\theta) \lvert f(\theta)\;d\theta < \infty.
    \]
    
    Note that if $\pi$ satisfies~\ref{BP} and~\ref{EQ}, $A(\pi)=0$. However, if $\pi$ does not satisfy~\ref{BP} or~\ref{EQ}, we can find a function $v$ or multiplier $\lambda$ such that $A(\pi)=-\infty$. This implies that the value of the primal problem~\ref{OT} is equal to $V$. We can also express $V$ as:
    \begin{align*}
        V=\sup_{\pi \in \mathcal{M}^+(\Theta \times [0,1])} \inf_{(v,\lambda) \in L^1(f) \times \mathbb{R}} \; \;  &\int_{\underline{\theta}}^{\bar{\theta}}v(\theta)f(\theta)\,d\theta\\
        &+\int_{\underline{\theta}}^{\bar{\theta}}\int_0^1\left[\Phi(\theta,q)+\lambda\left([H_{\theta}(p)-H_{\theta}(0)]q-RS(p)\right)-v(\theta)\right]\,d\pi(\theta,q).
    \end{align*}
    
    Here, it holds that $\sup \inf=\inf \sup$ (for a rigorous argument see the proof of Theorem 5.10 in \citet{villani2009optimal}), which yields:
    \[
        V=\inf_{(v,\lambda) \in L^1(f) \times \mathbb{R}} \; \;  \int_{\underline{\theta}}^{\bar{\theta}}v(\theta)f(\theta)\,d\theta+B(v,\lambda),
    \]
    where:
    \[
    B(v,\lambda)=\sup_{\pi \in \mathcal{M}^+(\Theta \times [0,1])} \; \; \int_{\underline{\theta}}^{\bar{\theta}}\int_0^1\left[\Phi(\theta,q)+\lambda\left([H_{\theta}(p)-H_{\theta}(0)]q-RS(p)\right)-v(\theta)\right]\,d\pi(\theta,q).
    \]
    
    Observe that if $v(x) < \Phi(\theta,q) + \lambda\left([H_{\theta}(p)-H_{\theta}(0)]q-RS(p)\right)$ for some $(\theta,q)$, we can find a positive measure $\pi$ such that $B(v,\lambda)=+\infty$. On the other hand, if $v(x) \geq  \Phi(\theta,q)+ \lambda\left([H_{\theta}(p)-H_{\theta}(0)]q-RS(p)\right)$ for all $(\theta,q)$, then $B(v,\lambda)=0$. This implies that the value of $V$ is equal to the value of the dual problem~\ref{D}, which shows that there is zero duality gap.
\end{proof}

\begin{lemma}\label{lemma:existence-primal}
    A solution to the primal problem~\ref{OT} exists.
\end{lemma}
\begin{proof}
    Let $\Pi$ denote the set of probability measures on $\Theta \times [0,1]$ that satisfy conditions~\ref{BP} and~\ref{EQ}. Formally, $\Pi$ is given by:
    \[
    \Pi=\left\{\pi \in \Delta\left(\Theta \times [0,1]\right) \; \Big \lvert \; \pi \text{ satisfies \ref{BP} and \ref{EQ}}\right\}.
    \]
    
    We first show that $\Pi$ is compact under the weak topology on $\Delta\left(\Theta \times [0,1]\right)$. For this, we invoke \textbf{Prokhorov's theorem}, which states:
    
    \textit{Let $Y$ be a Polish space, and let $\Delta(Y)$ denote the space of probability measures on $Y$. A set $P \subset \Delta(Y)$ is precompact under the weak topology if and only if it is tight; that is, for any $\varepsilon > 0$, there exists a compact set $K_{\varepsilon} \subset Y$ such that $\mu(K_{\varepsilon}) > 1 - \varepsilon$ for all $\mu \in P$.}

    Since $\Theta \times [0,1]$, equipped with the Euclidean metric, is a compact metric space, any collection of probability measures in $\Delta\left(\Theta \times [0,1]\right)$ is tight. Indeed, for any $\varepsilon > 0$, we can set $K_{\varepsilon} = \Theta \times [0,1]$. By Prokhorov's theorem, it follows that $\Pi$ is precompact. To conclude that $\Pi$ is compact, it remains to show that $\Pi$ is closed.

    Recall that a sequence of probability measures $\{\pi_n\}_{n=1}^{\infty} \subset \Pi$ converges to $\pi^*$ in the weak topology if, for every continuous function $g: \Theta \times [0,1] \rightarrow \mathbb{R}$:
    \[
        \lim_{n \rightarrow \infty} \int_{\underline{\theta}}^{\bar{\theta}}\int_0^1 g(\theta,q)\, d\pi_n(\theta,q) = \int_{\underline{\theta}}^{\bar{\theta}}\int_0^1 g(\theta,q)\, d\pi^*(\theta,q).
    \]

    In particular, since $H_{\theta}$ is continuous in $\theta$, we must have that if $\{\pi_n\}_{n=1}^{\infty} \rightarrow \pi^*$, then:
    \[
    \lim_{n \rightarrow \infty} \int_{\underline{\theta}}^{\bar{\theta}}\int_0^1 [H_{\theta}(p)-H_{\theta}(0)]q\, d\pi_n(\theta,q) = \int_{\underline{\theta}}^{\bar{\theta}}\int_0^1 [H_{\theta}(p)-H_{\theta}(0)]q\, d\pi^*(\theta,q).
    \]

    Using the fact that the following condition holds for all $n$:
    \[
    \int_{\underline{\theta}}^{\bar{\theta}}\int_0^1[H_{\theta}(p)-H_{\theta}(0)]   q\, d\pi_n(\theta,q) = RS(p),
    \]
    it follows that $\pi^*$ must satisfy~\ref{EQ}. A similar argument shows that if  $\{\pi_n\}_{n=1}^{\infty} \rightarrow \pi^*$, then $\pi^*$ must also satisfy~\ref{BP}. Thus, $\pi^* \in \Pi$, proving that $\Pi$ is closed.

    Finally, compactness of $\Pi$ implies the existence of a solution to~\ref{OT}. Indeed, let $\{\pi_n\}_{n=1}^{\infty}$ be any maximizing sequence, which admits a cluster point $\pi^* \in \Pi$. By our technical assumptions, $\Phi$ is a bounded and continuous function, then it can be written as the limit of a non-increasing sequence $\{\Phi_k\}_{k=1}^{\infty}$ of bounded continuous functions. By invoking successively the monotone convergence theorem, the fact that $\pi^*$ is a cluster point, the inequality $\Phi_k(\theta,q) \geq \Phi(\theta,q)$ and the maximizing property of $\{\pi_n\}_{n=1}^{\infty}$, we obtain:
\begin{align*}
    \int_{\underline{\theta}}^{\bar{\theta}}\int_0^1\Phi(\theta,q)\,d\pi^*(\theta,q) &= \lim_{k \rightarrow \infty} \int_{\underline{\theta}}^{\bar{\theta}}\int_0^1\Phi_k(\theta,q)\,d\pi^*(\theta,q) \\
    &\geq \lim_{k \rightarrow \infty} \lim_{n \rightarrow \infty} \inf\int_{\underline{\theta}}^{\bar{\theta}}\int_0^1\Phi_k(\theta,q)\,d\pi_n(\theta,q) \\
    &\geq \lim_{n \rightarrow \infty} \inf\int_{\underline{\theta}}^{\bar{\theta}}\int_0^1\Phi(\theta,q)\,d\pi_n(\theta,q) \\
    &= \sup_{\pi \in \Pi(f,p)} \; \int_{\underline{\theta}}^{\bar{\theta}}\int_0^1\Phi(\theta,q)d\pi(\theta,q).
\end{align*}
    Therefore, $\pi^*$ is a solution to~\ref{OT}.
\end{proof}

\begin{lemma}\label{lemma:existence-dual}
    A solution to the dual problem~\ref{D} exists.
\end{lemma}
\begin{proof}
    Let $\pi^*$ be any solution to the primal problem~\ref{OT}, whose existence is guaranteed by Lemma~\ref{lemma:existence-primal}. We can express the value of~\ref{OT} as:
    \[
        V = \int_{\underline{\theta}}^{\bar{\theta}}\int_0^1\Phi(\theta,q) \, d\pi^*(\theta,q)
        = \int_{\underline{\theta}}^{\bar{\theta}}\int_0^1\Phi(\theta,q)\, d\pi^*(q \mid \theta)f(\theta) \, d\theta,
    \]
    where $\pi^*(q \mid \theta) = \frac{\pi^*(\theta,q)}{f(\theta)}$ is the conditional distribution of $q$ given $\theta$. Indeed, since $\pi^*$ satisfies~\ref{BP}, we know that $\pi^*(\cdot \mid \theta) \in \Delta([0,1])$. For the rest of the proof, with a slight abuse of notation, we denote by $\pi$ the collection of conditional distributions, $\{\pi(\cdot \mid \theta)\}_{\theta \in \Theta}$, where each $\pi(\cdot \mid \theta) \in \Delta([0,1])$.

    We now define two key functions: 
    \[
    y(\pi) = \int_{\underline{\theta}}^{\bar{\theta}}\int_0^1\Phi(\theta,q) \, d\pi(q \mid \theta)f(\theta) \, d\theta, 
    \quad
    g(\pi) = \int_{\underline{\theta}}^{\bar{\theta}}\int_0^1 \left([H_{\theta}(p)-H_{\theta}(0)]q-RS(p)\right) \, d\pi(q \mid \theta)f(\theta) \, d\theta.
    \]
    
    From the definitions of $y$ and $g$, we immediately have the following reformulation of the original problem:
    \[
        V = \max_{\pi} \quad y(\pi) \quad \text{s.t} \quad g(\pi) = 0,
    \]
    and $\pi^*$ solves this optimization problem:
    \[
    \pi^* \in \arg\max_{\pi} \quad y(\pi) \quad \text{s.t} \quad g(\pi) = 0.
    \]

    By standard Lagrangian duality, we can argue that there exists a multiplier $\lambda^* \in \mathbb{R}_+$ such that the constrained optimization problem can be reformulated as: 
    \begin{align*}
        V &= \max_{\pi} \quad y(\pi) + \lambda^* g(\pi), \\
        \pi^* &\in \arg\max_{\pi} \quad y(\pi) + \lambda^* g(\pi).
    \end{align*}

    To rigorously justify the existence of $\lambda^*$, we use the separating hyperplane theorem. Define the following sets in $\mathbb{R}^2$:
    \begin{align*}
        A &= \left\{(r, z) \;\middle|\; r \leq y(\pi) \text{ and } z = g(\pi) \text{ for some } \pi \in \Delta([0,1])^\Theta \right\},\\
        B &= \left\{(r, 0) \;\middle|\; r \geq V\right\}.
    \end{align*}

    Since both $y$ and $g$ are linear in $\pi$, the sets $A$ and $B$ are convex. Furthermore, $A$ contains no interior points of $B$. Therefore, by the separating hyperplane theorem, there exists a non-zero vector $(r_0, \lambda^*) \in \mathbb{R}^2$ such that:
    \[
        r_0 r_1 + \lambda^* z_1 \leq r_0 r_2,
    \]
    for all $(r_1, z_1) \in A$ and $(r_2, 0) \in B$. We now show that $r_0 > 0$.

    \textbf{Case 1 $r_0 < 0$}: In this case, we can find a point $(r_2, 0) \in B$ with sufficiently large $r_2$ such that the inequality $r_0 r_1 + \lambda^* z_1 > r_0 r_2$ holds, leading to a contradiction.

    \textbf{Case 2 $r_0 = 0$}: Here, $\lambda^* \neq 0$. Since, for all $p$, $g$ can be positive for some $\pi$ and negative for other, i.e. it is always possible to have excess demand or excess supply. Then, $A$ contains two points, $(r,z)$ and $(r',z')$, such that $z>0$ and $z'<0$. However, if $\lambda^*>0$ ($\lambda^*<0$), then $\lambda^*z \leq 0$ ($\lambda^*z' \leq 0$), leading to a contradiction. 

    Thus, we must have $r_0 > 0$, and without loss of generality, we can normalize $r_0 = 1$. Therefore, we have:
    \begin{align*}
    V = \max_{(r, z) \in A} \quad r + \lambda^* z=\max_{\pi} \quad y(\pi) + \lambda^* g(\pi), 
    \end{align*}
    which proves the existence of $\lambda^*$.

    We now exploit some properties of $\pi^*$. Since $\pi^*$ maximizes $y(\pi)+\lambda^*g(\pi)$, for $f$-almost every $\theta$ we must have that:
    \[
        \int_0^1 \left[\Phi(\theta,q) + \lambda^*\left([H_{\theta}(p)-H_{\theta}(0)]q-RS(p)\right)\right] \, d\pi^*(q \mid \theta),
    \]
    is equal to:
    \[
        \max_{q \in [0,1]} \quad \Phi(\theta,q)+ \lambda^*\left([H_{\theta}(p)-H_{\theta}(0)]q-RS(p)\right).
    \]
    
    The maximization problem is well-defined because $\Phi(\theta,\cdot)$ is a continuous function and we are maximizing over the compact set $[0,1]$. Then, define the function $v^*: \Theta \to \mathbb{R}$ by:
    \[
    v^*(\theta) = \int_0^1 \left[\Phi(\theta,q) + \lambda^*\left([H_{\theta}(p)-H_{\theta}(0)]q-RS(p)\right)\right] \, d\pi^*(q \mid \theta).
    \]
    
    Then, the pair $(v^*, \lambda^*)$ satisfies~\ref{SP}. Finally, observe that:
    \[
    V = \int_{\underline{\theta}}^{\bar{\theta}} v^*(\theta) f(\theta) \, d\theta.
    \]
    
    This shows that $(v^*, \lambda^*)$ achieves the same value as the primal problem~\ref{OT}, and by Lemma~\ref{lemma:zero-duality}, we conclude that $(v^*, \lambda^*)$ is a solution to the dual problem~\ref{D}.
\end{proof}

\begin{proof}[Proof of the Necessity Part of Corollary~\ref{coro:complementary-slackness}]
    First, since $(v^*,\lambda^*)$ is a solution of the dual problem~\ref{D} it immediately follows that for $f-$almost every $\theta$ we must have:
    \[
    v^*(\theta)=\max_{q \in [0,1]} \quad \Phi(\theta,q)+\lambda^*\left([H_{\theta}(p)-H_{\theta}(0)]q-RS(p)\right),
    \]
    where the maximum is well-defined because $\Phi(\theta,\cdot)$ is a continuous function and we are maximizing over the compact set $[0,1]$. 

    The above implies that if any joint distribution $\pi$ satisfies conditions~\ref{BP} and~\ref{EQ}, but there is a positive measure of risk types $\theta$ for which $\pi(\cdot \mid \theta)$ is not supported on~\ref{supp}, then the following inequality holds:
    \begin{align*}
        \int_{\underline{\theta}}^{\bar{\theta}}v^*(\theta)f(\theta)\,d\theta>&\int_{\underline{\theta}}^{\bar{\theta}}\int_{0}^1\left[\Phi(\theta,q)+\lambda^*\left([H_{\theta}(p)-H_{\theta}(0)]q-RS(p)\right)\right]\,d\pi(q \mid\theta)f(\theta)\,d\theta\\
        =&\int_{\underline{\theta}}^{\bar{\theta}}\int_{0}^1\Phi(\theta,q)\,d\pi(q \mid\theta)f(\theta)\,d\theta\\
        &+\lambda^*\int_{\underline{\theta}}^{\bar{\theta}}\int_{0}^1\left([H_{\theta}(p)-H_{\theta}(0)]q-RS(p)\right)\,d\pi(q \mid\theta)f(\theta)\,d(\theta)\\
        =&\int_{\underline{\theta}}\int_{0}^1\Phi(\theta,q)\,d\pi(q,\theta).
    \end{align*}
    
    By Theorem \ref{thm:zero-duality}, this implies that $\pi$ is not a solution of the primal problem~\ref{OT}.
\end{proof}
    
\begin{proof}[Proof of the Sufficiency Part of Corollary~\ref{coro:complementary-slackness}]
   Suppose that $\pi$ satisfies~\ref{BP},~\ref{EQ} and for $f-$almost every $\theta$, the conditional distribution $\pi(\cdot \mid \theta)$ is supported on~\ref{supp}. Then, the following equality holds:
    \begin{align*}
        \int_{\underline{\theta}}^{\bar{\theta}}v^*(\theta)f(\theta)\,d\theta=&\int_{\underline{\theta}}^{\bar{\theta}}\int_{0}^1\left[\Phi(\theta,q)+\lambda^*\left([H_{\theta}(p)-H_{\theta}(0)]q-RS(p)\right)\right]\,d\pi(q \mid\theta)f(\theta)\,d\theta\\
        =&\int_{\underline{\theta}}^{\bar{\theta}}\int_{0}^1\Phi(\theta,q)\,d\pi(q \mid\theta)f(\theta)\,d\theta\\
        &+\lambda^*\int_{\underline{\theta}}^{\bar{\theta}}\int_{0}^1\left([H_{\theta}(p)-H_{\theta}(0)]q-RS(p)\right)\,d\pi(q \mid\theta)f(\theta)\,d(\theta)\\
        =&\int_{\underline{\theta}}\int_{0}^1\Phi(\theta,q)\,d\pi(q,\theta).
    \end{align*}
    
    By Theorem~\ref{thm:zero-duality}, it follows that $\pi$ is a solution to the primal problem~\ref{OT}.
\end{proof}

\end{document}